\documentclass{article}

\usepackage[a4paper, total={6in, 8in}]{geometry}
\usepackage{longtable}
\usepackage{fancyhdr}
\usepackage{color}
\usepackage{array}
\usepackage{mdwmath}
\usepackage{mdwtab}
\usepackage{amsmath,amssymb}
\usepackage{cite}
\usepackage{tikz}
\usetikzlibrary{arrows,decorations.pathmorphing,backgrounds,positioning,fit,petri}
\usetikzlibrary{automata}
\usetikzlibrary{graphs}
\usetikzlibrary{shadows}
\tikzstyle{blackboxState}=[state,fill=black,text=white]

\usepackage{graphicx}
\usepackage{listings}
\usepackage{colortbl}
\usepackage{subfig}
\usepackage{booktabs}
\usepackage{amsthm}
\usepackage{latexsym}
\usepackage{color}
\usepackage{rotating}
\usepackage{multirow}
\usepackage{paralist}
\usepackage{bibentry}
\PassOptionsToPackage{hyphens}{url}\usepackage[bookmarks=true,pdftex=false,bookmarksopen=true,pdfborder={0,0,0}]{hyperref}
\usepackage{url}
\usepackage{lscape}
\usepackage{algorithm}
\usepackage{longtable}
\usepackage[T1]{fontenc} 
\usepackage[latin1]{inputenc}
\usepackage{color}
\usepackage{amsmath}
\usepackage{ltl}
\usepackage{authblk}

\usepackage{amsthm}

\usepackage{algpseudocode}
\usepackage{hhline}

\newcommand{\IFSAPartialExtended}{Incomplete FSAs}
\newcommand{\IBASPartialExtended}{Incomplete BAs}
\newcommand{\IBAPartialExtended}{Incomplete BA}
\newcommand{\IBASExtended}{Incomplete B\"uchi Automata}
\newcommand{\IBAExtended}{Incomplete B\"uchi Automaton}
\newcommand{\IFSAS}{IFSAs}
\newcommand{\IFSA}{IFSA}
\newcommand{\IFSASExtended}{Incomplete Finite State Automata}
\newcommand{\IFSAExtended}{Incomplete Finite State Automaton}
\newcommand{\IBAS}{IBAs}
\newcommand{\IBA}{IBA}

\usepackage{amsthm,thmtools}
\declaretheorem[numberwithin=section,name=Proposition]{proposition}
\declaretheorem[numberwithin=section,name=Definition]{mydef}

\declaretheorem[numberwithin=section,name=Lemma]{lemma}

\declaretheorem[numberwithin=section,name=Theorem]{theorem}

\newcommand{\subproperty}{$\mathcal{S}=\langle  \mathcal{P},$ $\Delta^{in\mathcal{S}},$ $ \Delta^{out\mathcal{S}},$ $G,$ $R,$ $K,$ $\Gamma_\mathcal{M},$ $\Gamma_{\mathcal{A}_{\neg \phi}} \rangle$ }

\newcommand{\replacement}{$\mathcal{R}=\langle \mathcal{T},$ $\Delta^{{inR}},$ $ \Delta^{{outR}} \rangle$}
\newcommand{\claimautomata}{$\mathcal{A}_{\neg \phi}$}
\newcommand{\model}{\mathcal{M}}
\newcommand{\replacementtext}{$\mathcal{R}$}
\newcommand{\subpropertypossiblytext}{$S_p$}
\newcommand{\subpropertytext}{$S$}
\newcommand{\constraintf}{$\mathcal{C}=\langle S, S_p\rangle$}
\newcommand{\constrainttext}{$\mathcal{C}$}
\newcommand{\flag}{\mathcal{Y}}
\newcommand{\underapproximation}{\mathcal{U}}
\newcommand{\overapproximation}{$\mathcal{O}$}
\newcommand{\refinement}{\mathcal{N}}
\newcommand{\propertytext}{\phi}

\title{Modeling, refining and analyzing\\ \IBASExtended} % Title

\author[1]{Claudio Menghi}
\author[2]{Paola Spoletini}
\author[1]{Carlo Ghezzi}
\affil[1]{DEIB, Politecnico di Milano, Italy}
\affil[ ]{\textit {\{claudio.menghi, carlo.ghezzi\}@polimi.it}}
\affil[2]{Kennesaw State University, USA}
\affil[ ]{\textit {pspoleti@kennesaw.edu}}

\date{} 

\makeatletter
\def\@maketitle{%
  \newpage
  \null
  \vskip 2em%
  \begin{center}%
  \let \footnote \thanks
    {\Large\bfseries \@title \par}%
    \vskip 1.5em%
    {\normalsize
      \lineskip .5em%
      \begin{tabular}[t]{c}%
        \@author
      \end{tabular}\par}%
    \vskip 1em%
    {\normalsize \@date}%
  \end{center}%
  \par
  \vskip 1.5em}
\makeatother
\begin{document}
\maketitle
Software development is an iterative process which includes a set of development steps that transform the initial high level specification of the system into its final, fully specified, implementation~\cite{wirth1971program}.
The modeling formalisms used in this refinement process depend on the properties of the system that are of interest. 
This report discusses the theoretical foundations that allow  \IBASExtended\ (IBAs) to be used in the iterative development of a sequential system. 
Section~\ref{Ch:ModelingIncompleteness} describes the IBA modeling formalism and its properties. Section~\ref{sec:modelingclaim} specifies the semantic of LTL formulae over \IBAS. Section~\ref{sec:checkingIBAs} describes the model checking algorithm for \IBA\ and proofs its correctness. Section~\ref{sec:ComputingConstraints} describes the constraint computation algorithm. Finally, section~\ref{sec:replacementChecking} describes the replacement checking procedure and its properties.

\section{Modeling and refinining \IBASExtended}
\label{Ch:ModelingIncompleteness}
\thispagestyle{empty}

Section~\ref{sec:ModelingIncompleteSystems} describes \IFSASExtended\ (\IFSA) and \IBASExtended\ (\IBAS) which extend classical Finite State Automata and B\"uchi Automata with black box states. 
Section~\ref{sec:refinement} describes how these two modeling formalisms can be used in the refinement process, i.e., how the initial, incomplete, high level specification can be iteratively refined. 

\subsection{Modeling incomplete systems}
\label{sec:ModelingIncompleteSystems}

\subsubsection{Incomplete Finite State Automata}
\label{sec:IncompleteFSA}
\IFSAPartialExtended\ (\IFSAS) are a state based modeling formalism that extends FSAs by partitioning the set of the states $Q$ in two sets: the set of \emph{regular}  states $R$ and the set of \emph{black box states} $B$\footnote{Black box states have been also identified in other works as transparent states, such as in~\cite{sharifloo2013lover}.}.  
Regular states correspond to classical automata states, while black box states are placeholders for configurations in which the behavior of the system is currently unspecified.
Black box states will be later replaced by other automata, other \IFSAS. 
In the rest of this report black box states are often abbreviated as black boxes or boxes.

\begin{mydef}[\IFSAExtended] Given a finite set of atomic propositions $AP$, a non-deterministic \IFSAPartialExtended\ (\IFSA) $\model$ is a tuple $\langle \Sigma, R, B, Q, \Delta, Q^0,$ $F \rangle$, where:
\begin{inparaenum}[\itshape a\upshape)]
\item $\Sigma=2^{AP}$ is the finite alphabet;
\item $R$ is the finite set of \emph{regular} states;
\item $B$ is the finite set of \emph{box} states;
\item $Q$ is the finite set of states such that $Q= B \cup R$ and $B \cap R= \emptyset$;
\item $\Delta \subseteq Q\times \Sigma \times Q$ is the transition relation;
\item $Q^0 \subseteq Q$ is the set of initial states; 
\item $F \subseteq Q$ is the set of final states. 
\end{inparaenum}
\end{mydef}

Graphically, boxes are filled with black, initial states are marked by an incoming arrow, and final states are double circled.
Note that the transition relation allows the definition of transitions that connect states of $Q$ irrespective of their type. 
An example of \IFSA\ defined over the set of propositions $AP=\left \{start, \right.$  $fail,$ $ok,$ $success,$ $\left. abort \right \}$ is shown in Figure~\ref{Fig:IFSAExample}. 
This automaton is a well known example of incompleteness in the context of software development and has been presented in~\cite{Alur:2001:MCH:503502.503503}.
$Q=\left \{ q_1,\right.$ $send_1,$ $send_2,$ $q_2,$ $\left. q_3 \right \}$, $Q^0=\left \{ q_1 \right \}$, $F=\left \{ q_2, q_3\right \}$ and $B=\left \{ send_1, send_2\right \}$ are the set of the states, of the initial states, of the final states and of the boxes, respectively.

\begin{figure}[ht]
\centering
\includegraphics[scale=0.7]{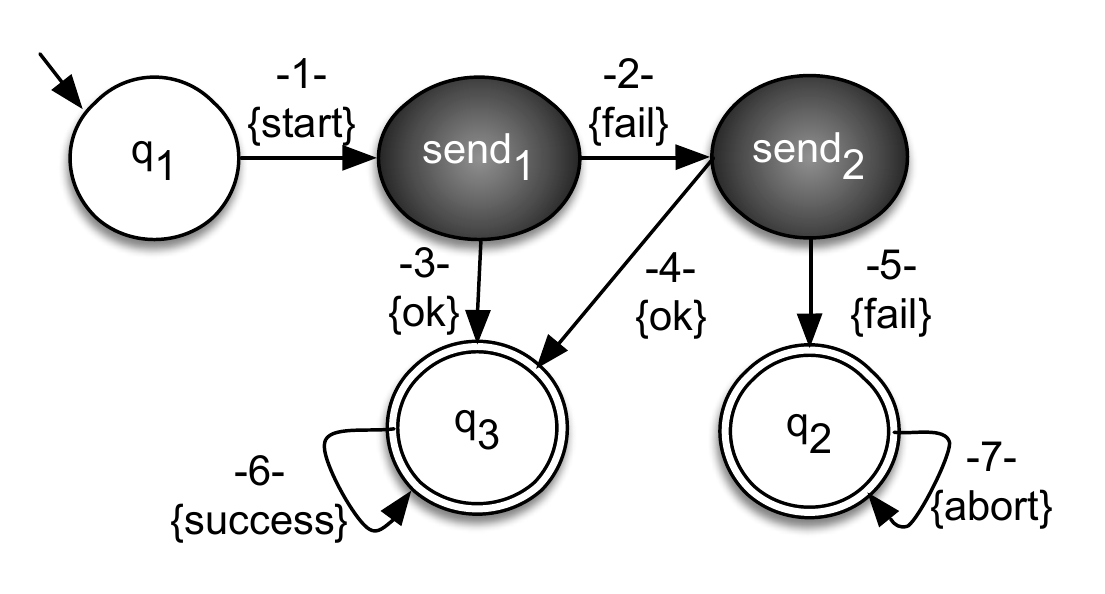}
\caption{An example of \IFSA.}
\label{Fig:IFSAExample}
\end{figure}

Given a word $v\in \Sigma^\ast$ of length $|v|$ a run defines the sequences of states traversed by the automaton to recognize $v$.
\begin{mydef}[\IFSA\ run]
Given a set of atomic propositions $AP$, an \IFSA\ $\model=\langle \Sigma, R, B, Q, \Delta,$ $Q^0,$ $F \rangle$, such that  $\Sigma=2^{AP}$,  a set of atomic propositions $AP^\prime$, such that $AP \subseteq AP^\prime$ and $\Sigma^\prime=2^{AP^\prime}$, and a \emph{word} $v=v_0 v_1 v_2\ldots v_{|v-1|}$ of length $|v|$ in $\Sigma^{\prime^\ast}$,  a \emph{run}  over the word $v$ is a mapping $\rho^\ast:  \left \{ 0,1,2 \ldots |v|  \right \} \rightarrow Q$ such that: 
\begin{inparaenum}[\itshape a\upshape)]
\item $\rho^\ast(0) \in Q^0$;
\item for all $0 \leq i < |v| $, $(\rho^\ast(i)$, $v_{i}$, $\rho^\ast(i+1))  \in \Delta$ or $\rho^\ast(i) \in B$ and $\rho^\ast(i)=\rho^\ast(i+1)$.
\end{inparaenum}
\end{mydef}

A run $\rho^\ast$ corresponds to a path in the \IFSA\ $\model$, such that the first state $\rho^\ast(0)$ of the path is an initial state of $\model$, i.e., it is in the set $Q^0$, and either the system moves form a state $\rho^\ast(i)$ to the next state $\rho^\ast(i+1)$ by reading the character $v_i$, or the state $\rho^\ast(i)$ is a box ($\rho^\ast(i) \in B$) and the character $v_i$ is recognized ``inside" the box $\rho^\ast(i)=\rho^\ast(i+1)$. 
For example, the finite word \{$start$\}.\{$send$\}.\{$fail$\} can be associated with the run $\rho^\ast(0)=q_1$, $\rho^\ast(1)=send_1$, $\rho^\ast(2)=send_1$  and $\rho^\ast(3)=send_2$ or with the run $\rho^\ast(0)=q_1$, $\rho^\ast(1)=send_1$, $\rho^\ast(2)=send_1$ and $\rho^\ast(3)=send_1$ .

\begin{mydef}[\IFSA\ definitely accepting and possibly accepting run]
A run $\rho^\ast$ is \emph{definitely accepting} if and only if $\rho^\ast(|v|) \in F$ and for all $0 \leq i \leq |v|,\ \rho^\ast(i) \in R$. 
A run $\rho^\ast$ is \emph{possibly accepting} if and only if $\rho^\ast(|v|) \in F$ and there exists $0 \leq i \leq |v|$ such that $\rho^\ast(i) \in B$. 
A run $\rho^\ast$ is not accepting otherwise.
\end{mydef}

Informally, a run $\rho^\ast$ is \emph{definitely accepting} if and only if ends in a final state of $\model$ and all the states of the run are regular, it is \emph{possibly accepting} if and only if it ends in a final state of $\model$ and there exists at least a state of the run which is a box, it is not accepting otherwise.

\begin{mydef} [\IFSA\ definitely accepted and possibly accepted word]
An \IFSA\ $\model$ \emph{definitely accepts} a word $v$ if and only if there exists a definitely accepting run of $\model$ on $v$. 
$\model$ \emph{possibly accepts} a word $v$ if and only if it does not definitely accept $v$ and there exists at least a possibly accepting run of $\model$ on $v$. 
Finally, $\model$ \emph{does not accept} $v$ iff it does not contain any definitely accepting or possibly accepting run for $v$.  
\end{mydef}

Note that possibly accepted words describe \emph{possible behaviors}. 
For example, the word \{$start$\}. \{$send$\}.\{$ok$\} is possibly accepted by the automaton presented in Figure~\ref{Fig:IFSAExample} since no definitely accepting run exists, while there exists a possibly accepting run described by the function $\rho^\ast$, such that $\rho^\ast(0)=q_1$, $\rho^\ast(1)=send_1$, $\rho^\ast(2)=send_1$ and $\rho^\ast(3)=q_3$.

\begin{mydef} [\IFSA\ definitely accepted and possibly accepted language]
Given a finite set of atomic propositions $AP^\prime$, such that $AP \subseteq AP^\prime$, and the alphabet $\Sigma^\prime=2^{AP^\prime}$, the language $\mathcal{L}^\ast(\model) \subseteq \Sigma^{\prime^\ast}$ definitely accepted by an \IFSA\ $\model$ contains all the words $v_1, v_2 \ldots v_n \in \Sigma^{\prime^\ast}$ definitely accepted by $\model$. 
The possibly accepted language $\mathcal{L}^\ast_p(\model) \subseteq \Sigma^{\ast}$ of $\model$ contains all the words $v_1, v_2 \ldots v_n \in \Sigma^{\prime^\ast}$ possibly accepted  by $\model$. 
\end{mydef}

Given an \IFSA\ $\model$ it is possible to define its completion  $\model_c$ as the FSA obtained by removing its boxes and their incoming and outgoing transitions.

\begin{mydef} [Completion of an \IFSA] Given an \IFSA\ $\model=\langle \Sigma, R, B,$ $ Q, \Delta,  Q^0, F \rangle$ the completion of $\model$ is the FSA $\model_c=\langle \Sigma, R, \Delta_c, Q^0 \cap R, F \cap R \rangle$, such as $\Delta_c=\left \{ (s, a, s^\prime) \mid (s, a, s^\prime) \in \Delta \right.$ and  $s \in R$ and $\left. s^\prime \in R\right \}$.
\end{mydef}

Lemma~\ref{lem:languageIFSA}  proves that the completion of an \IFSA\ recognizes its definitely accepted language.
\begin{lemma}[Language of the completion of an \IFSA]
\label{lem:languageIFSA} Given an \IFSA\ $\model=\langle \Sigma, R, B, Q, \Delta, $ $ Q^0, F \rangle$ the completion $\model_c$ of $\model$ recognizes the definitely accepted language $\mathcal{L}^\ast(\model)$.
\end{lemma}

\begin{proof}
\label{proof:languageIFSA}
To prove Lemma~\ref{lem:languageIFSA} it is necessary to demonstrate that a word is recognized by the completion if and only if it belongs to the definitely accepted  language of $\model$, i.e., $v \in \mathcal{L}^\ast(\model) \Leftrightarrow v \in \mathcal{L}^\ast(\model_c)$.

($\Rightarrow$) Each word $v$ accepted by $\model$ is associated with an accepting run $\rho^\ast$ which contains only regular states. Since $\model_c$ contains all the regular states of $\model$ and the same transitions between these states, it is possible to simulate the run $\rho^\ast$ of  $\model $ on the automaton $\model_c$. 
Furthermore, the regular and final states of $\model$ are also final for the automaton $\model_c$.
This implies that $v$ is definitely accepted by $\model_c$.

($\Leftarrow$) is proved by contradiction. Imagine that there exists a word $v$ in $\mathcal{L}^\ast(\model_c)$ which is not in $\mathcal{L}^\ast(\model)$. This implies that there exists a run $\rho^\ast$ in $\model_c$ which does not correspond to a run $\rho^{\ast^\prime}$ in $\model$. Given one of the states $\rho^\ast(i)$ it can be associated to the corresponding state of  $\model$.
Given two states $\rho^\ast(i)$ and $\rho^\ast(i+1)$ of the run and the transition $(\rho^\ast(i), a, \rho^\ast(i+1)) \in \Delta_c$, it is possible to ``simulate" the transition by performing the corresponding transition of $\model$ since $\Delta_c \subseteq \Delta$. 
Furthermore, every final state of $\model_c$ is also final for $\model$.
This implies that $v$ is also accepted by $\model$, and therefore $v$ is in the language $\mathcal{L}^\ast(\model)$ contradicting the hypothesis.
\end{proof}

The size $|\model|$ of an \IFSA\ $\model$ is the sum of the cardinality of the set of its states and the cardinality of the set of its transitions.
\begin{mydef}[Size of an \IFSA]  The size $|\model|$ of an IFSA $\model=\langle \Sigma, R, B, Q, \Delta, $ $ Q^0, F \rangle$ is equal to $|Q|+|\Delta|$.
\end{mydef}

\subsubsection{Incomplete B\"uchi Automata}
\label{sec:IncompleteBA}
Software systems are usually not designed to stop during their execution, thus infinite models of computation are usually considered. 
B\"uchi Automata (BAs) are one of the most used infinite models of computation. 
This section introduces \IBASPartialExtended\ (\IBAS) an extended version of BAs that support incompleteness.

\begin{mydef} [\IBASExtended] 
A non-deterministic \IBAExtended\ (\IBA) is an \IFSA\ $\langle \Sigma, R, B, Q, \Delta, Q^0, F \rangle$, where the set of final states $F$ of the \IFSA\ is used to define the acceptance condition for infinite words (also called $\omega$-\emph{words}). 
The set $F$ identifies the \emph{accepting states} of the IBA.  
\end{mydef}

Given an $\omega$-\emph{word} $v=v_0 v_1 v_2\ldots$  a run defines an execution of the \IBA\ (sequence of states).
\begin{mydef} [\IBA\ run]
\label{incompleterun}
Given a set of atomic propositions $AP$, an \IFSA\ $\model=\langle \Sigma, R, B, Q, \Delta, Q^0,$ $F \rangle$, such that  $\Sigma=2^{AP}$,  a set of atomic propositions $AP^\prime$, such that $AP \subseteq AP^\prime$ and $\Sigma^\prime=2^{AP^\prime}$, and a word $v \in \Sigma^{\prime^\omega}$, a \emph{run} $\rho^\omega: \left \{ 0,1,2, \ldots  \right \}\rightarrow Q$ over $v$ is defined for an \IBA\ as follows:
\begin{inparaenum}[\itshape a\upshape)]
\item $\rho^\omega(0) \in Q^0$;
\item for all $i \geq 0$, $(\rho^\omega(i)$, $v_{i}$, $\rho^\omega(i+1)) \in \Delta $ or $\rho^\omega(i) \in B$ and $\rho^\omega(i)= \rho^\omega(i+1)$.
\end{inparaenum}
\end{mydef}

Informally, a character $v_{i}$ of the word $v$ can be recognized by a transition of the \IBA, changing the state of the automaton from $\rho^\omega(i)$ to $\rho^\omega(i+1)$, or it can be recognized by a transition of the \IBA\ that will replace the box $\rho^\omega(i) \in B$. 
In the latter, the state $\rho^\omega(i+1)$ of the automaton after the recognition of $v_{i}$ corresponds to $\rho^\omega(i)$, since the control remains to the automaton which will replace the box $\rho^\omega(i)$. 
For example, the infinite word \{$start$\}.\{$send$\}.\{$ok$\}.\{$success$\}$^\omega$ can be associated with the run $\rho^\omega(0)=q_1$, $\rho^\omega(1)=send_1$ and $\rho^\omega(2)=send_1$ and $\forall i \geq 3, \rho^\omega(i)=q_3$ of the automaton described in Figure~\ref{Fig:IFSAExample} when it is interpreted as an \IBA. 
The character $send$ is recognized by the box $send_1$.

Let $inf(\rho^\omega)$ be the set of states that appear infinitely often in the run $\rho^\omega$.
\begin{mydef} [\IBA\ definitely accepted and possibly accepted run]
A run $\rho^\omega$ of an \IBA\ $\model$ is:
\begin{inparaenum}[\itshape a\upshape)]
\item  \emph{definitely accepting} if and only if $inf(\rho^\omega) \cap F \neq \emptyset$ and for all $i \geq 0, \rho^\omega(i) \in R$;
\item  \emph{possibly accepting} if and only if $(inf(\rho^\omega) \cap F \neq \emptyset)$ and there exists $i \geq 0$ such that $\rho^\omega(i) \in B$;
\item \emph{not accepting} otherwise.
\end{inparaenum}
\end{mydef}

Informally, a run is definitely accepting if some accepting states appear in $\rho^\omega$ infinitely often and all states of the run are regular states, it is possibly accepting if some accepting states appear in $\rho^\omega$ infinitely often and there is at least one state in the run that is a box, not accepting otherwise.

\begin{mydef} [\IBA\ definitely accepted and possibly accepted word]
An automaton $\model$ \emph{definitely accepts} a word $v$ if and only if there exists a definitely accepting run of $\model$ on $v$. 
$\model$ \emph{possibly accepts} a word $v$ if and only if it does not definitely accept $v$ and there exists at least a possibly accepting run of $\model$ on $v$. 
Finally, $\model$ \emph{does not accept} $v$ if and only if it does not contain any accepting or possibly accepting run for $v$.  
\end{mydef}

As for \IFSA, possibly accepted words describe \emph{possible behaviors}. 
For example, the automaton described in Figure~\ref{Fig:IFSAExample} (when it is interpreted as an IBA) possibly accepts the infinite word \{$start$\}.\{$send$\}.\{$ok$\}.\{$success$\}$^\omega$ since a definitely accepting run does not exist but there exists a run which is possibly accepting.

\begin{mydef} [\IBA\ definitely accepted and possibly accepted language]
Given a finite set of atomic propositions $AP^\prime$, such that $AP \subseteq AP^\prime$, and the alphabet $\Sigma^\prime=2^{AP^\prime}$, the language $\mathcal{L}^\omega(\model) \in \Sigma^{\prime^\omega}$ definitely accepted by an \IBA\ $\model$ contains all the words definitely accepted by $\model$. 
The possibly accepted language $\mathcal{L}^{\omega}_p(\model) \in \Sigma^{\prime^\omega}$ of $\model$ contains all the words possibly accepted  by $\model$. 
\end{mydef}

The language $\mathcal{L}^\omega(\model)$ can be defined by considering the BA $\model_c$ obtained from $\model$ by removing its boxes and their incoming and outgoing transitions.

\begin{mydef}[Completion of an \IBA] 
\label{def:IBACompletion}
Given an \IBA\ $\model=\langle \Sigma, R, B, Q,$ $ \Delta,  Q^0, F \rangle$ the completion of $\model$ is the BA $\model_c=\langle \Sigma, R, \Delta_c, Q^0 \cap R, F \cap R \rangle$, such as $\Delta_c=\left \{ (s, a, s^\prime) \mid (s, a, s^\prime) \in \Delta \right.$ and $s \in R$ and $\left. s^\prime \in R \right \}$.
\end{mydef}

As for \IFSA, it is possible to prove that the completion of an \IBA\ recognizes its definitely accepted language.

\begin{lemma}[Language of the completion of an \IBA] 
\label{th:languageMc} 
Given an \IBA\ $\model=\langle \Sigma, R, B, Q, \Delta, $ $ Q^0, F \rangle$ the completion $\model_c$ of $\model$ recognizes the definitely accepted language $\mathcal{L}^\omega(\model)$.
\end{lemma}
\begin{proof}
The proof of Lemma~\ref{th:languageMc} is similar to the proof of Lemma~\ref{lem:languageIFSA} and requires to demonstrate that $v \in \mathcal{L}^\omega(\model) \Leftrightarrow v \in \mathcal{L}^\omega(\model_c)$. 

($\Rightarrow$) Each word $v$ definitely accepted by $\model$ is associated to a definitely accepting run $\rho^\omega$ which only contains regular states. 
Since $\model_c$ contains all the regular states of $\model$ and the same transitions between these states, it is possible to simulate the run $\rho^\omega$ of  $\model $ on the automaton $\model_c$. Furthermore, the regular and accepting states of $\model$ are also accepting for the automaton $\model_c$. This implies that $v$ is also accepted by $\model_c$.

($\Leftarrow$) is proved by contradiction. 
Imagine that there exists a word $v$ in $\mathcal{L}^\omega(\model_c)$ which is not in $\mathcal{L}^\omega(\model)$. 
This implies that there exists a run $\rho^\omega$ in $\model_c$ which does not correspond to a run $\rho^{\omega^\prime}$ in $\model$. 
Consider the run $\rho^\omega$, each state $\rho^\omega(i)$  can be associated to the corresponding state of  $\model$.
Given two states $\rho^\omega(i)$ and $\rho^\omega(i+1)$ of the run and the transition $(\rho^\omega(i), a, \rho^\omega(i+1)) \in \Delta_c$ it is possible to ``simulate" the transition by performing the corresponding transition of $\model$ since $\Delta_c \subseteq \Delta$. 
Furthermore, every accepting state of $\model_c$ is also accepting for $\model$.
This implies that $v$ is also accepted by $\model$, and therefore $v$ is in the language $\mathcal{L}^\omega(\model)$, which violates the hypothesis.
\end{proof}

The size $|\model|$ of an \IBA\ $\model$ is the sum of the cardinality of the set of its states and the set of its transitions. 
\begin{mydef}[Size of an \IBA]  
The size $|\model|$ of an \IBA\ $\model=\langle \Sigma, R, B, Q, \Delta, $ $ Q^0, F \rangle$ is $|Q|+|\Delta|$.
\end{mydef}

\subsection{Refining incomplete models}
\label{sec:refinement}
The development activity is an iterative and incremental process through which the initial high level specification $\model$ is iteratively refined. 
After having designed the initial high level specification $\model$, the modeling activity proceeds through a set of refinement rounds $\mathcal{R}\mathcal{R}$. 
At each refinement round $r \in \mathcal{R}\mathcal{R}$, a box $b$ of $\model$ is refined.
We use the term \emph{refinement} to capture the notion of model elaboration, i.e., the model $\refinement$ is a refinement  $\model$  if it is  obtained from $\model$ by adding knowledge about the behavior of the system inside one of its boxes. 
We call \emph{replacement} the sub-automaton which specifies the behavior of  the system inside a specific box.

\subsubsection{Refining \IBASExtended}
The refinement relation $\preceq$ allows the iterative concretization of the model of the system by replacing boxes with other \IBAS. These \IBAS\ are called \emph{replacements}. The definition of the refinement relation $\preceq$ has been inspired from~\cite{shoham2004monotonic}.

\begin{mydef}[Refinement]
\label{def:refinement}  Let $\wp_\model$ be the set of all possible \IBAS. An \IBA\ $\refinement \in \wp_\model$ is a refinement of an \IBA\ $\model \in \wp_\model$, i.e., $\model \preceq \refinement$, iff $\Sigma_\model \subseteq\Sigma_\refinement $ and there exists some refinement relation $\Re \in Q_\model \times Q_\refinement$, such that:
\begin{enumerate}
\item \label{def:refinementRegularMRegularN} for all $q_\model \in R_\model$ there exists  exactly one $q_\refinement \in  R_\refinement$ such that  $(q_\model, q_\refinement) \in \Re$;
\item  \label{def:stateNstateM} for all $q_\refinement \in Q_\refinement$ there exists exactly one $q_\model \in Q_\model$ such that  $(q_\model, q_\refinement)$ $\in \Re$;
\item \label{def:initianlNinitialM} for all $(q_\model, q_\refinement) \in \Re$, if $q_\refinement \in Q^0_\refinement$ then $q_\model \in Q^0_\model$;
\item \label{def:boxNboxM} for all $(q_\model, q_\refinement) \in \Re$, if $q_\refinement \in B_\refinement$ then $q_\model \in B_\model$;
\item \label{def:acceptingNacceptingM} for all $(q_\model, q_\refinement) \in \Re$, if $q_\refinement \in F_\refinement$ then $\ q_\model \in F_\model$;
\item \label{def:initialRegularMinitialRegularN} for all $(q_\model, q_\refinement) \in \Re$, if $q_\model \in Q^0_\model \cap  R_\model$ then $q_\refinement \in Q^0_\refinement \cap  R_\refinement$;
\item \label{def:acceptingRegularMacceptingRegularN} for all $(q_\model, q_\refinement) \in \Re$, if  $q_\model \in F_\model \cap R_\model$ then $q_\refinement \in F_\refinement \cap R_\refinement$;
\item \label{def:refinement6} for all $ (q_\model, q_\refinement) \in \Re$ and $\forall a \in \Sigma_\refinement$, if $(q_\model, a, q_\model^\prime) \in \Delta_\model$ then there exists $q_\refinement^\prime \in  Q_\refinement$ such that one of the following is satisfied:
\begin{itemize}
\item $(q_\refinement, a, q_\refinement^\prime) \in \Delta_\refinement$ and $(q_\model^\prime, q_\refinement^\prime)\in \Re$;
\item $q_\model \in B_{\model}$ and there exists $q_\refinement^{\prime\prime} \in Q_{\refinement}$ such that $(q_\model, q_\refinement^{\prime\prime}) \in \Re$ and $(q_\refinement^{\prime\prime}, a, q_\refinement^{\prime}) \in  \Delta_\refinement$ and $(q_\model, q_\refinement^{\prime}) \in \Re$;
\end{itemize}  
\item \label{def:refinement7} for all $(q_\model, q_\refinement) \in \Re$ and $\forall a \in \Sigma_\refinement$, if $(q_\refinement, a, q_\refinement^\prime) \in \Delta_\refinement $ one of the following holds:
\begin{itemize}
\item there exists $q_\model^\prime \in Q_\model$ such that $(q_\model^\prime, q_\refinement^\prime) \in \Re$ and $(q_\model, a, q_\model^\prime) \in \Delta_\model$;
\item  $q_\model \in B_\model$ and $(q_\model, q^\prime_\refinement)\in \Re$.
\end{itemize}
\end{enumerate}
\end{mydef}
The idea behind the refinement relation is that every definite behavior of $\model$ \emph{must be preserved} in its refinement $\refinement$, and every behavior of $\refinement$ must correspond to a behavior of $\model$. 

Condition~\ref{def:refinementRegularMRegularN} imposes that each \emph{regular state} of $\model$ is associated with exactly one regular state of the refinement $\refinement$. When $q_\model$ is a box several states (or none) of $\refinement$ can be associated with $q_\model$.
Condition~\ref{def:stateNstateM} imposes that each state (regular or black box) of the refinement $\refinement$ is associated with exactly one state of the model $\model$. 
Condition~\ref{def:initianlNinitialM} specifies that any initial state of the refinement  $\refinement$ is  associated with an initial state of the model $\model$.
Condition~\ref{def:boxNboxM} guarantees that any \emph{box} in the refinement $\refinement$ is  associated with a box of the model $\model$, i.e., it is not possible to refine a regular state into a box.
Condition~\ref{def:acceptingNacceptingM} specifies that each \emph{accepting state} of $\refinement$ corresponds with an accepting state of $\model$.
Condition~\ref{def:initialRegularMinitialRegularN} forces each initial and regular state of the model $\model$ to be associated with an initial and regular state $\refinement$.
Condition~\ref{def:acceptingRegularMacceptingRegularN} specifies that each accepting and regular state of $\model$ is associated with an accepting and regular state of  $\refinement$.
Finally,  conditions~\ref{def:refinement6} and~\ref{def:refinement7} constrain the \emph{transition relation}. 
Given a state $q_\model$ in $\model$ and a corresponding state $q_\refinement$ of the refined automaton $\refinement$, condition~\ref{def:refinement6} specifies that for each transition $(q_\model, a, q_\model^\prime)$ either there exists a state $q_\refinement^\prime$ that follows $q_\refinement$ through a transition labeled with $a$, or the state $q_\model$ is a box and another transition $(q_\refinement^{\prime\prime}, a, q_\refinement^{\prime})$ that exits the state  $q_\refinement^{\prime\prime}$ of the replacement of the box $q_\model$ is associated with the transition $(q_\model, a, q_\model^\prime)$\footnote{Note that the state $q_\refinement^{\prime\prime}$ must not be necessarily reachable in the replacement of the state $q_\refinement$.}. 
Condition~\ref{def:refinement7}  guarantees that each transition $(q_\refinement, a, q_\refinement^\prime)$ in the refinement $\refinement$ must be associated with a transition $(q_\model, a, q_\model^\prime)$ of $\model$ or it is a transition performed inside box $q_\model$, i.e., $q_\model \in B_\model$.

Consider for example the automaton $\model$ presented in Figure~\ref{Fig:IFSAExample} and the automaton $\refinement$ of Figure~\ref{Fig:FSAExample}, $\model \preceq \refinement$, through the relation $\Re=\left \{ (q_1, q_1),\right.$ $(send_1, q_4),$ $(send_1, q_5),$ $(send_1, q_6),$ $(send_1, q_7),$ $(send_1, q_8),$ $(send_2, q_9),$ $(send_2, q_{10}),$ $(send_2, q_{11}),$ $(send_2, q_{12}),$ $(send_2, q_{13}),$ $(q_2, q_2),$ $\left. (q_3,q_3) \right \}$.

\begin{figure}[ht]
\centering
\includegraphics[scale=0.4]{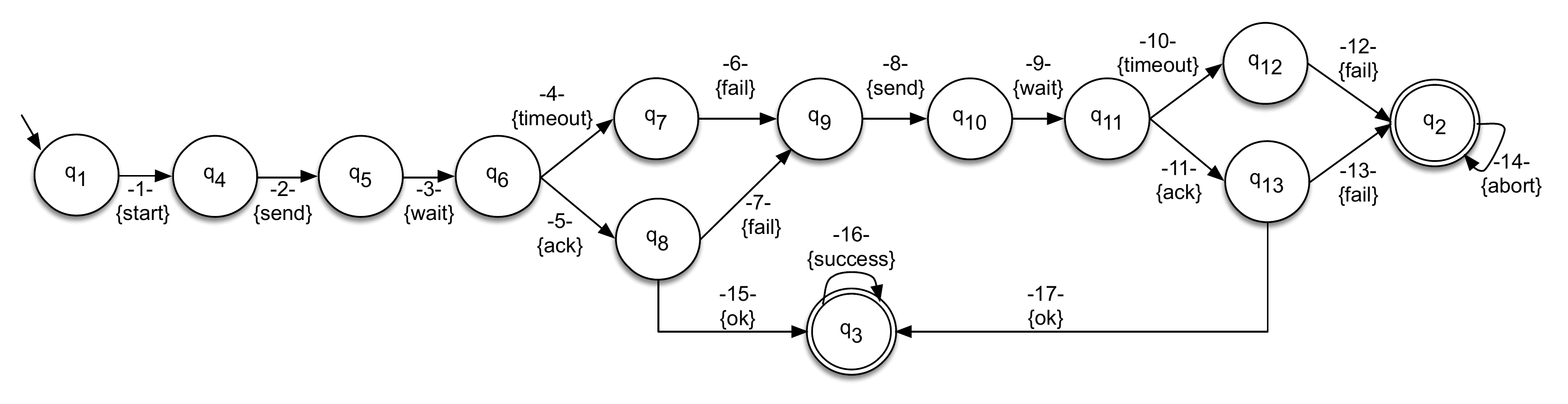}
\caption{An example of BA.}
\label{Fig:FSAExample}
\end{figure}

\begin{mydef} [Implementation] 
A BA $\refinement$ is an implementation of an \IBA\ $\model$ if and only if $\model \preceq \refinement$.
\end{mydef} 

The automaton $\refinement$ presented in Figure~\ref{Fig:FSAExample} is also an implementation of the automaton $\model$  described in Figure~\ref{Fig:IFSAExample}.

It is important to notice that the refinement relation preserves the language  containment relation, i.e., a possibly accepted word of $\model$ can be definitely accepted, possibly accepted or not accepted in the refinement, but every definitely accepted and not accepted word remains accepted or not accepted in $\refinement$.

\begin{lemma}[The refinement relation is reflexive] 
\label{lem:reflexive} 
Given an \IBA\ $\model$, $\model \preceq \model$.
\end{lemma}
\begin{proof}
We prove that there exists a relation $\Re_{\model \preceq \model}$ between the states of $\model$ and its refinement $\model$ that satisfies the conditions specified in Definition~\ref{def:refinement}.  
More precisely, the relation $\Re_{\model \preceq \model}$, such that for all $q \in Q_{\model}$,  $(q,q) \in \Re_{\model \preceq \model}$, satisfies the conditions of Definition~\ref{def:refinement}, since the states and transitions of $\model$ are the same states and transitions of its refinement $\model$. 
\end{proof}

\begin{lemma}[The refinement relation is transitive] 
\label{lem:transitivity} 
Given three IBAs $\model$, $\refinement$ and $\mathcal{O}$, if $\model \preceq \refinement$ and   $\refinement \preceq \mathcal{O}$ then $\model \preceq \mathcal{O}$.
\end{lemma}

\begin{proof}

Since  $\model \preceq \refinement$ and $\refinement \preceq \mathcal{O}$, there exist a refinement relation $\Re_{\model \preceq \refinement}$  between the states of $\model$ and $\refinement$ and  a refinement relation $\Re_{\refinement \preceq \mathcal{O}}$ between the states of  $\refinement$  and $\mathcal{O}$, respectively. 
To prove lemma~\ref{lem:transitivity} we need to show that exists a refinement relation $\Re_{\model \preceq \mathcal{O}}$ between the states of $\model$ and $\mathcal{O}$ that satisfies the conditions specified in Definition~\ref{def:refinement}.

Let us consider the relation $\Re_{\model \preceq \mathcal{O}}$  such that $( (q, q^{\prime\prime}) \in \Re_{\model \preceq \refinement} \text{ and } (q^{\prime\prime}, q^\prime) \in \Re_{\refinement \preceq \mathcal{O}}) \Leftrightarrow (q, q^\prime) \in \Re_{\model \preceq \mathcal{O}}$. 
We prove that $\Re_{\model \preceq \mathcal{O}}$ satisfies the conditions specified in Definition~\ref{def:refinement}.

Condition~\ref{def:refinementRegularMRegularN}. 
\emph{Every regular state of $\model$ must also be contained in its refinement $\mathcal{O}$}. 
Since $\model \preceq \refinement$, each regular state $q_{\model}$  must also be contained in $\refinement$, i.e., it must exists a regular state  $q_{\refinement}$ such that $(q_{\model}, q_{\refinement}) \in \Re_{\model \preceq \refinement}$. 
Since  $\refinement \preceq \mathcal{O}$  each regular state $q_{\refinement}$  must also be contained in $\mathcal{O}$, i.e., it must exists a regular state $q_{\mathcal{O}}$ such that $(q_{\refinement}, q_{\mathcal{O}}) \in \Re_{\refinement \preceq \mathcal{O}}$.
This implies that $(q_{\model}, q_{\mathcal{O}}) \in \Re_{\model \preceq \mathcal{O}}$, i.e., for every regular state $q_{\model}$ it exists a regular state $q_{\mathcal{O}}$ associated to $q_{\model}$ through the relation $\Re_{\model \preceq \mathcal{O}}$.

Condition~\ref{def:stateNstateM}.
\emph{Each state of the refinement $\mathcal{O}$ must be associated with exactly one state of  $\model$}. 
Since $\refinement \preceq \mathcal{O}$, each state $q_{\mathcal{O}}$ of $\mathcal{O}$ is associated with exactly one state $q_{\refinement}$ of $\refinement$ through the relation $\Re_{\refinement \preceq \mathcal{O}}$.
Since $\model \preceq \refinement$, each state $q_{\refinement}$ of $\refinement$ is associated with exactly one state $q_{\model}$ of $\model$, through the relation $\Re_{\model \preceq \refinement}$.
This implies that $(q_{\model}, q_{\mathcal{O}}) \in \Re_{\model \preceq \mathcal{O}}$ for construction, i.e.,  for each  each $q_{\mathcal{O}}$ there exists exactly one $q_{\model}$ such that $(q_{\model}, q_{\mathcal{O}}) \in \Re_{\model \preceq \mathcal{O}}$.

Condition~\ref{def:initianlNinitialM}. 
\emph{Each initial state of the refinement $\mathcal{O}$ must be associated with an initial state of the model  $\model$}.  
Since $\refinement \preceq \mathcal{O}$, each state $q_{\mathcal{O}}$ of $\mathcal{O}$ which is an initial state is associated with exactly one state $q_{\refinement}$  of $\refinement$ which is also an initial state through the relation $\Re_{\refinement \preceq \mathcal{O}}$.
Since $\model \preceq \refinement$, each state $q_{\refinement}$ of $\refinement$ which is an initial state is associated with exactly one state $q_{\mathcal{M}}$ of $\model$, which is also an initial state through the relation $\Re_{\model \preceq \refinement}$.
This implies that whenever $(q_{\model}, q_{\mathcal{O}}) \in \Re_{\model \preceq \mathcal{O}}$, if  $q_{\mathcal{O}}$ is an initial state, then $q_{\model}$ is also an initial state.

Condition~\ref{def:boxNboxM}. 
\emph{Each box of the refinement $\mathcal{O}$ must be associated with a box of the model  $\model$}.  
Since $\refinement \preceq \mathcal{O}$, each state $q_{\mathcal{O}}$ of $\mathcal{O}$ which is a box is associated with exactly one state $q_{\refinement}$  of $\refinement$ which is also a box through the relation $\Re_{\refinement \preceq \mathcal{O}}$.
Since $\model \preceq \refinement$, each state $q_{\refinement}$ of $\refinement$ which is a box is associated with exactly one state $q_{\model}$ of $\model$, which is also a box through the relation $\Re_{\model \preceq \refinement}$.
This implies that whenever $(q_{\model}, q_{\mathcal{O}}) \in \Re_{\model \preceq \mathcal{O}}$, if  $q_{\mathcal{O}}$ is a box, then $q_{\model}$ is a box.

Condition~\ref{def:acceptingNacceptingM}. 
\emph{Each accepting state of the refinement $\mathcal{O}$ must be associated with a accepting state of the model  $\model$}.
Since $\refinement \preceq \mathcal{O}$, each state $q_{\mathcal{O}}$ of $\mathcal{O}$ which is an accepting state is associated with exactly one state $q_{\refinement}$  of $\refinement$ which is also an accepting state through the relation $\Re_{\refinement \preceq \mathcal{O}}$.
Since $\model \preceq \refinement$, each state $q_{\refinement}$ of $\refinement$ which is an accepting state is associated with exactly one state $q_{\model}$ of $\model$, which is also an accepting state through the relation $\Re_{\model \preceq \refinement}$.
This implies that whenever $(q_{\model}, q_{\mathcal{O}}) \in \Re_{\model \preceq \mathcal{O}}$, if  $q_{\mathcal{O}}$ is an accepting state, then $q_{\model}$ is an accepting state.

Condition~\ref{def:initialRegularMinitialRegularN}. 
\emph{Each initial and regular state of $\model$ must be associated with an initial state of the refinement  $\mathcal{O}$}.
Since $\model \preceq \refinement$, each initial and regular state $q_{\model}$  must also be contained in $\refinement$, i.e., it must exists an initial and regular state  $q_{\refinement}$ such that $(q_{\model}, q_{\refinement}) \in \Re_{\model \preceq \refinement}$. 
Since  $\refinement \preceq \mathcal{O}$  each initial and regular state $q_{\refinement}$  must also be contained in $\mathcal{O}$, i.e., it must exists an initial and regular state $q_{\mathcal{O}}$ such that $(q_{\refinement}, q_{\mathcal{O}}) \in \Re_{\refinement \preceq \mathcal{O}}$.
This implies that $(q_{\model}, q_{\mathcal{O}}) \in \Re_{\model \preceq \mathcal{O}}$, i.e., for every initial and regular state  $q_{\model}$ it exists an initial and regular state $q_{\mathcal{O}}$ associated to $q_{\model}$ through the relation $\Re_{\model \preceq \mathcal{O}}$.

Condition~\ref{def:acceptingRegularMacceptingRegularN}. 
\emph{Each accepting and regular state of $\model$ must be associated with an accepting state of the refinement  $\mathcal{O}$}.
Since $\model \preceq \refinement$, each accepting and regular state $q_{\model}$  must also be contained in $\refinement$, i.e., it must exists an accepting and regular state  $q_{\refinement}$ such that $(q_{\model}, q_{\refinement}) \in \Re_{\model \preceq \refinement}$. 
Since  $\refinement \preceq \mathcal{O}$  each accepting and regular state $q_{\refinement}$  must also be contained in $\mathcal{O}$, i.e., it must exists an accepting and regular state $q_{\mathcal{O}}$ such that $(q_{\refinement}, q_{\mathcal{O}}) \in \Re_{\refinement \preceq \mathcal{O}}$.
This implies that $(q_{\model}, q_{\mathcal{O}}) \in \Re_{\model \preceq \mathcal{O}}$, i.e., for every accepting and regular state  $q_{\model}$ it exists an accepting and regular state $q_{\mathcal{O}}$ associated to $q_{\model}$ through the relation $\Re_{\model \preceq \mathcal{O}}$.

Condition~\ref{def:refinement6}. 
\emph{A transition of the model starting from a state $q_{\model}$ is associated with a transition of the refinement 
which starts from a state that refines $q_{\model}$.}
Let us consider a couple $(q_{\model},  q_{\mathcal{O}}) \in \Re_{\model \preceq \mathcal{O}}$ and a transition $(q_{\model}, a, q^\prime_{\model}) \in \Delta_{\model}$, we have to prove that it exists a $q^\prime_{\mathcal{O}}$ that satisfies the conditions specified by the condition~\ref{def:refinement6}. 
Since $\model \preceq \refinement$ it exists a $q_{\refinement}$ such that $(q_{\model}, q_{\refinement}) \in \Re_{\model \preceq \refinement}$ and a $q^\prime_{\refinement}$, such that $(q^\prime_{\model}, q^\prime_{\refinement}) \in \Re_{\model \preceq \refinement}$, and  $q^\prime_{\refinement}$ satisfies one of the statements specified in Definition~\ref{def:refinement} condition~\ref{def:refinement6}. 
\begin{itemize}
\item Assume that the first statement of condition~\ref{def:refinement6} is satisfied, i.e., $(q_{\refinement}, a, q^\prime_{\refinement}) \in \Delta_\refinement$. Since $\refinement \preceq \mathcal{O}$, it must exists a $q_{\mathcal{O}}$ and a $q_{\mathcal{O}}^{\prime}$ such that  $(q_{\refinement}, q_{\mathcal{O}})  \in \Re_{\refinement \preceq \mathcal{O}}$, $(q^\prime_{\refinement}, q^\prime_{\mathcal{O}})  \in \Re_{\refinement \preceq \mathcal{O}}$ and  one of the statements specified in Definition~\ref{def:refinement} condition~\ref{def:refinement6} is satisfied. \\
If $(q_{\mathcal{O}}, a, q^\prime_{\mathcal{O}}) \in \Delta_\mathcal{O}$, then there exists $q^\prime_{\mathcal{O}}$ such that $(q^\prime_{\model}, q^\prime_{\mathcal{O}}) \in \Re_{\model \preceq \mathcal{O}}$ and $(q_{\mathcal{O}}, a, q^\prime_{\mathcal{O}}) \in \Delta_\mathcal{O}$.\\
If instead $q_{\refinement}$ is a box, then there exist a $q_{\mathcal{O}}^{\prime\prime}$  such that $(q_{\refinement},q_{\mathcal{O}}^{\prime\prime})\in \Re_{\refinement \preceq \mathcal{O}}$  and $(q_{\mathcal{O}}^{\prime\prime}, a, q_{\mathcal{O}}^{\prime}) \in \Delta_\mathcal{O}$. 
Since $q_{\refinement}$ is a box by condition~\ref{def:boxNboxM} also $q_{\model}$ is a box and $(q_{\model},q_{\mathcal{O}})\in \Re_{\model \preceq \mathcal{O}}$.
Furthermore, $(q_{\model},q_{\mathcal{O}}^{\prime\prime})\in \Re_{\model \preceq \mathcal{O}}$ and $(q_{\model}^\prime,q_{\mathcal{O}}^{\prime})\in \Re_{\model \preceq \mathcal{O}}$. Thus, it exists a $q_{\mathcal{O}}^{\prime\prime}$ such that $(q_{\model},q_{\mathcal{O}}^{\prime\prime})\in \Re_{\model \preceq \mathcal{O}}$, $(q_{\mathcal{O}}^{\prime\prime}, a, q_{\mathcal{O}}^{\prime}) \in \Delta_\mathcal{O}$ and $(q_{\model}^\prime,q_{\mathcal{O}}^{\prime})\in \Re_{\model \preceq \mathcal{O}}$.
\item Assume that the second statement of condition~\ref{def:refinement6} is satisfied, i.e., $q_{\model}\in B_{\model}$ and there exists a $q_{\refinement}^{\prime\prime}$ such that $(q_\model, q_\refinement^{\prime\prime}) \in \Re_{\model \preceq \refinement}$ and $(q_{\refinement}^{\prime\prime} , a, q_{\refinement}^{\prime}) \in \Delta_{\refinement}$.  
Since $\refinement \preceq \mathcal{O}$, it must exists a $q_{\mathcal{O}}^{\prime\prime}$, such that $(q_{\refinement}^{\prime\prime}, q_{\mathcal{O}}^{\prime\prime}) \in  \Re_{\refinement \preceq \mathcal{O}}$ and a $q_{\mathcal{O}}^{\prime}$, such that $(q^{\prime}_{\refinement}, q_{\mathcal{O}}^{\prime}) \in \Re_{\refinement \preceq \mathcal{O}}$.
Furthermore, $\refinement \preceq \mathcal{O}$ implies that $(q_{\refinement}^{\prime\prime}, q_{\mathcal{O}}^{\prime\prime})$ satisfies one of the  conditions of~\ref{def:refinement6}. \\
If the first statement of condition~\ref{def:refinement6} is satisfied then $(q^{\prime\prime}_{\mathcal{O}}, a, q^\prime_{\mathcal{O}}) \in \Delta_\mathcal{O}$. This proves that there exists $q^{\prime\prime}_{\mathcal{O}}$ such that $(q_{\model}, q^{\prime\prime}_{\mathcal{O}}) \in \Re_{\model \preceq \mathcal{O}}$ and $(q_{\mathcal{O}}^{\prime\prime}, a, q^\prime_{\mathcal{O}}) \in \Delta_\mathcal{O}$ and $(q_{\model}^\prime, q^{\prime}_{\mathcal{O}}) \in \Re_{\model \preceq \mathcal{O}}$ as required by the first statement of condition~\ref{def:refinement6}.\\
Otherwise, if  $q^{\prime\prime}_{\refinement}$ is a box,  it must exists a $q_{\mathcal{O}}^{\prime\prime\prime}$, such that $(q^{\prime\prime}_{\refinement}, q_{\mathcal{O}}^{\prime\prime\prime}) \in \Re_{\refinement \preceq \mathcal{O}}$, and $(q^{\prime\prime\prime}_{\mathcal{O}}, a, q_{\mathcal{O}}^{\prime}) \in \Delta_{\mathcal{O}}$. 
This proves that there exists $q^{\prime\prime\prime}_{\mathcal{O}}$ such that $(q_{\model}, q^{\prime\prime\prime}_{\mathcal{O}}) \in \Re_{\model \preceq \mathcal{O}}$ and $(q_{\mathcal{O}}^{\prime\prime\prime}, a, q^\prime_{\mathcal{O}}) \in \Delta_\mathcal{O}$ and $(q_{\model}^\prime, q^{\prime}_{\mathcal{O}}) \in \Re_{\model \preceq \mathcal{O}}$  as required by the second statement of condition~\ref{def:refinement6}
\end{itemize}

Condition~\ref{def:refinement7}.  \emph{A transition of the refinement is associated with a transition of the model or to one of its black box states}. It is necessary to prove that for all $(q_\model, q_\mathcal{O}) \in  \Re_{\model \preceq \mathcal{O}}$, if $(q_\mathcal{O}, a, q_\mathcal{O}^\prime) \in \Delta_{\mathcal{O}}$, one of the statements specified in condition~\ref{def:refinement7} is satisfied.  
Since $\refinement \preceq \mathcal{O}$, it must exists a $(q_\refinement, q_\mathcal{O}) \in  \Re_{\refinement \preceq \mathcal{O}}$ such that one of the two statements specified in condition~\ref{def:refinement7} is satisfied.\\
Let first consider the case in which the first statement is satisfied. Then, it must exists a state $q^\prime_\refinement$, such that $(q^\prime_\refinement, q^\prime_\mathcal{O}) \in  \Re_{\refinement \preceq \mathcal{O}}$ and $(q_{\refinement}, a, q^\prime_{\refinement}) \in \Delta_{\refinement}$.  Since  $\model \preceq \refinement$, it must exists a $(q_\model, q_\refinement) \in  \Re_{\model \preceq \refinement}$ that satisfies one of the two statements specified in condition~\ref{def:refinement7}.
\begin{itemize}
\item If it exists a state $q^\prime_\model$ such that $(q^\prime_\model, q^\prime_\refinement) \in  \Re_{\model \preceq \refinement}$ and $(q_{\model}, a, q^\prime_{\model}) \in \Delta_{\model}$, then we can conclude that $(q_\model, q_\mathcal{O}) \in  \Re_{\model \preceq \mathcal{O}}$, $(q^\prime_\model, q^\prime_\mathcal{O}) \in  \Re_{\model \preceq \mathcal{O}}$ and $(q_{\model}, a, q^\prime_{\model}) \in \Delta_{\model}$, which satisfies the first statement of condition~\ref{def:refinement7}.
\item If $q_\model \in B_\model$ and $(q_\model, q^\prime_\refinement) \in  \Re_{\model \preceq \refinement}$, then we can conclude that $(q_\model, q^\prime_\mathcal{O}) \in  \Re_{\model \preceq \mathcal{O}}$.
\end{itemize}
Let us then consider the case in which the second statement is satisfied. Then, it must exist a box $q_\refinement$, such that $(q_\refinement, q^\prime_\mathcal{O}) \in  \Re_{\refinement \preceq \mathcal{O}}$. Similarly it may also exists a $q_\model$, which is a box, such that $(q_\model, q_\refinement) \in  \Re_{\model \preceq \refinement}$. This implies that $(q_\model, q_\mathcal{O}) \in  \Re_{\model \preceq \refinement}$ and $(q_\model, q^\prime_\mathcal{O}) \in  \Re_{\model \preceq \refinement}$
\end{proof}

\begin{theorem}[Language preservation] 
\label{th:languagePreservation} 
Given an \IBA\ $\model$ and one of its refinements $\refinement$, for all $v^\omega \in \Sigma^\omega$:
\begin{enumerate}
\item  \label{satisfiedPreservation} if $v^\omega  \in \mathcal{L}^\omega(\model)$ then  $v^\omega \in  \mathcal{L}^\omega(\refinement)$
\item  \label{notsatisfiedPreservation} if $v^\omega  \not \in (\mathcal{L}^\omega_p(\model) \cup  \mathcal{L}^\omega(\model) )$ then  $v^\omega  \not \in (\mathcal{L}^\omega_p(\refinement) \cup  \mathcal{L}^\omega(\refinement) )$
\end{enumerate}
\end{theorem}
\begin{proof}
Let us first prove the statement~\ref{satisfiedPreservation} of Theorem~\ref{th:languagePreservation}. 
Since $v^\omega$ is definitely accepted by the \IBA\ $\mathcal{L}^\omega(\model)$, it must exists a definitely accepting run $\rho^\omega_\model$ of $\model$. 
Note that definitely accepting runs only contains states that are regular. 
Let us consider the initial state $\rho^\omega_\model(0)$. 
By Definition~\ref{def:refinement} conditions~\ref{def:refinementRegularMRegularN} and~\ref{def:initialRegularMinitialRegularN} it must exists a state $q^0_{\refinement} \in Q^0_{\refinement}$ such that $(\rho^\omega_\model(0), q^0_{\refinement}) \in \Re$. 
Let us identify with $\rho^\omega_\refinement$ a run which starts in this state and is iteratively obtained as follows. 
Consider a generic step step $i$.  
Given two states  $\rho^\omega_\model(i)$, $\rho^\omega_\model(i+1)$ of the run $\rho^\omega_\model$ it must exist a transition $(\rho^\omega_\model(i),a, \rho^\omega_\model(i+1)) \in \Delta_\model$. 
By Definition~\ref{def:refinement} condition~\ref{def:refinement6} it must exists a transition ($q^\omega_\refinement(i), a, q^\omega_\refinement(i+1)$) of $\Delta_\refinement$, where $(\rho^\omega_\model(i) ,\rho^\omega_\refinement(i)) \in \Re$ and $(\rho^\omega_\model(i+1) ,\rho^\omega_\refinement(i+1)) \in \Re$.
Condition~\ref{def:acceptingRegularMacceptingRegularN} imposes that a regular accepting state of $q_\model$ is associated with an accepting state of $q_\refinement$. 
Thus, since $\rho^\omega_\model$ and $\rho^\omega_\refinement$ move from $\rho^\omega_\model(i)$ and $\rho^\omega_\refinement(i)$ to  $\rho^\omega_\model(i+1)$ and $\rho^\omega_\refinement(i+1)$ by reading the same characters and for construction the corresponding runs are definitely accepting we conclude that $v^\omega \in \mathcal{L}^\omega(\refinement)$.

Let us now consider the statement~\ref{notsatisfiedPreservation} of Theorem~\ref{th:languagePreservation}. 
The proof is by contradiction. 
Imagine that there exists a word $v^\omega  \not \in (\mathcal{L}^\omega_p(\model) \cup  \mathcal{L}^\omega(\model) )$ and $v^\omega  \in (\mathcal{L}^\omega_p(\refinement) \cup  \mathcal{L}^\omega(\refinement))$. 
Since $v^\omega  \in (\mathcal{L}^\omega_p(\refinement) \cup  \mathcal{L}^\omega(\refinement))$, it must exists a definitely accepting or possibly accepting run $\rho^\omega_\refinement$ associated with this word. 
Let us consider the initial state $\rho^\omega_\refinement(0)$ of this run. 
By Definition~\ref{def:refinement}, condition~\ref{def:stateNstateM}, it must exists an initial state $q_\model \in Q_\model$ such that $(q_\model, \rho^\omega_\refinement(0))\in \Re$. 
Since $\rho_\refinement(0)$ is initial by Definition~\ref{def:refinement}, condition~\ref{def:initianlNinitialM},  we derive that $q_\model$ is also initial. 
Let us identify as $\rho^\omega_\model$ a run in $\model$ which starts from $q_\model$.  Given two states  $\rho^\omega_\refinement(i)$, $\rho^\omega_\refinement(i+1)$ of the run $\rho^\omega_\refinement$ it must exists a transition $(\rho^\omega_\refinement(i),a, \rho^\omega_\refinement(i+1)) \in \Delta_\refinement$. By Definition~\ref{def:refinement}, condition~\ref{def:refinement7}, either it exists a transition $(\rho^\omega_\model(i), a, \rho^\omega_\model(i+1))$ of $\Delta_\model$ or $\rho^\omega_\model(i) \in B_\model$. 
Finally,  condition~\ref{def:acceptingNacceptingM} imposes that an accepting state of $q_\refinement$ is associated with an accepting state of $q_\model$. Thus, since $\rho^\omega_\model$ and $\rho^\omega_\refinement$ moves from $\rho^\omega_\model(i)$ and $\rho^\omega_\refinement(i)$ to  $\rho^\omega_\model(i+1)$ and $\rho^\omega_\refinement(i+1)$, respectively, by reading the same characters, or $\rho^\omega_\model(i)=\rho^\omega_\model(i+1)$ and $\rho^\omega_\model(i) \in B_\model$, and by construction the corresponding runs are accepting, we conclude that $v^\omega \in \mathcal{L}^\omega(\model)$ or $v^\omega \in \mathcal{L}_p^\omega(\model)$ which contradict our hypothesis.
\end{proof}

\subsubsection{Replacements}
\label{Sec:ModelingReplacements}
Consider an \IBA\ $\model$. 
At each refinement round $i \in \mathcal{R}\mathcal{R}$, the developer designs a replacement $\mathcal{R}$\footnote{The term replacement is also used for example in~\cite{nopper2004approximate}.} for one of the boxes $b \in B_{\model_i}$ of  $\model_i$, where  $\model_i$ is the refinement of the automaton $\model$ before the refinement round $i$.

\begin{mydef}[Replacement]
\label{replacement}
 Given an \IBA\ $\model=\langle \Sigma_{\model}, R_{\model}, B_{\model}, Q_{\model},$ $ \Delta_{\model}, Q_{\model}^0, F_{\model} \rangle$, the replacement $\mathcal{R}$ of the box $b \in B_{\model}$ is defined as a triple 
 $\langle \mathcal{T}, \Delta^{inR},$ $ \Delta^{outR} \rangle$. 
 $\mathcal{T}=\langle \Sigma_{\mathcal{T}}, R_{\mathcal{T}}, B_{\mathcal{T}},$ $Q_{\mathcal{T}}, \Delta_{\mathcal{T}},$ $ Q_{\mathcal{T}}^0, F_{\mathcal{T}} \rangle$ is an \IBA, $\Delta^{inR} \subseteq \left \{ (q^\prime, a, q) \right. $ $ \left. \mid (q^\prime, a, b) \in  \Delta_{\model}\text{ and }q \in Q_{\mathcal{T}}\right \}$ and $ \Delta^{outR}$ $\subseteq \left \{ (q, a, q^\prime) \mid \right.$ $\left. (b, a, q^\prime) \in  \Delta_{\model}  \text{ and } \right.$ $\left.q \in Q_{\mathcal{T}}\right \}$ are its incoming and outgoing transitions, respectively.
$\mathcal{R}$ must  satisfy the following conditions: 
\begin{itemize}
\item if $b \not \in Q^0_{\model}$ then $Q^0_{\mathcal{T}}=\emptyset$;
\item if $b \not \in F_{\model}$ then $F_{\mathcal{T}}=\emptyset$;
\item if $(q^\prime, a, b) \in  \Delta_{\model}$ then it exists $(q^\prime, a, q)  \in \Delta^{inR}$, such that $q \in Q_{\mathcal{T}}$;
\item if $(b, a, q^\prime)   \in  \Delta_{\model}$ then it exists $(q, a, q^\prime)   \in \Delta^{outR}$, such that $q \in Q_{\mathcal{T}}$;
\item if $(b, a, b) \in  \Delta_{\model}$ then it exists $(q^\prime, a, q) \in  \Delta_{\mathcal{T}}$.
\end{itemize}
\end{mydef}

Informally, $\mathcal{T}$ is the \IBA\ to be substituted to the box $b$, $\Delta^{inR}$ and $\Delta^{outR}$ specify how the replacement is connected to the states of $\model$.
Consider for example the replacement $\mathcal{R}_{send_1}$ described in Figure~\ref{Fig:send1Replacement} which refers to the box $send_1$ of the model $\model$  described in Figure~\ref{Fig:IFSAExample} (the replacement assumes that $send_1$  is bot initial and accepting). 
The automaton $\model_{send_1}$ is defined over the set of atomic propositions $AP_{send_1}=\left \{start, \right.$ $\left.  booting, ready, send,  wait, timeout, ack, fail, ok \right \}$. The states $q_{14}$, $q_{15}$ and $q_{17}$ are the initial, accepting and a box of the replacement, respectively. Note that the initial/accepting states must be initial/accepting for the whole system, i.e., not only in the scope of the considered replacement. Furthermore, the destination/source of an incoming/outgoing transition is not considered as initial/accepting if they are not initial/accepting for $\model_{send_1}$.

\begin{figure}[ht]
\centering
\includegraphics[scale=0.5]{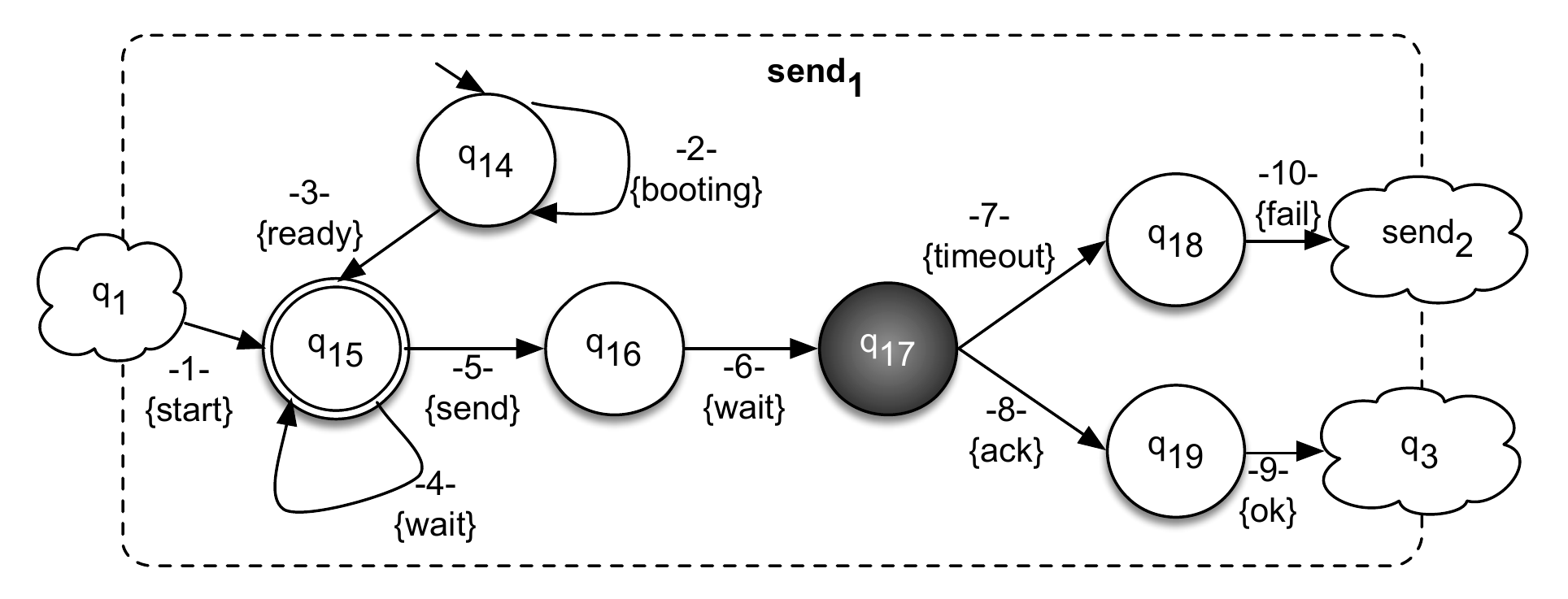}
\caption{The replacement of the box $send_1$.}
\label{Fig:send1Replacement}
\end{figure}

As for \IBAS\ we define the \emph{completion} of a replacement $\mathcal{R}_{c}$ as the replacement where the corresponding automaton is discharged from its boxes and their incoming and outgoing transitions.

 When a replacement is considered four different types of runs can be identified:
\begin{itemize}
\item \emph{finite internal runs}: are the runs which start from an initial state that is internal to the replacement and reach an outgoing transition of the replacement;
\item \emph{infinite internal runs}: are the runs that start from an initial state that is internal to the replacement and infinitely enter an internal accepting state  \emph{without} leaving the replacement;
\item \emph{finite external runs}: are the runs that start  from an incoming transition of the replacement and reach an outgoing transition of the replacement, i.e., they are finite paths that cross the component;
\item \emph{infinite external runs}: are the  runs that start from an incoming transition of the replacement and reach an accepting state which is internal to the replacement it-self \emph{without} leaving the replacement.
\end{itemize} 

\noindent We identify as $Q^{0inR}=\{q \in Q_{\model}$ such that there exist $ q^{\prime} \in Q_{\mathcal{T}}$ and an $a \in \Sigma_{\model}$ and  $(q, a, q^\prime)  \in  \Delta^{inR}  \}$ and  
$Q^{0outR}= \{q \in Q_{\mathcal{T}}$ such that there exist  $q^{\prime} \in  Q_{\mathcal{T}}$ and an $a \in \Sigma_{\model}$ and $(q^{\prime}, a,  q) \in  \Delta^{inR} \}$ the set of the states that are sources and destinations of incoming transitions, respectively.  
We indicate with $F^{inR}=\{q \in Q_{\mathcal{T}}$ such that there exist $q^{\prime} \in  Q_{\model}$ and an $a \in \Sigma_{\model}$ and $(q, a, q^{\prime}) \in  \Delta^{outR} \}$ and with $F^{outR}= \{q  \in Q_{\model} $ such that there exist $q^{\prime} \in  Q_{\mathcal{T}}$ and an $a \in \Sigma_{\model}$ and $(q^{\prime}, a, q) \in \Delta^{outR} \}$ the set of the states that are sources and destinations of outgoing transitions.

Infinite internal runs, finite internal runs, infinite external runs and finite external runs can then formally defined as in the following.

\begin{mydef} [Finite Internal Run] Given a replacement $\mathcal{R}=\langle \mathcal{T},$ $ \Delta^{inR},$ $  \Delta^{outR} \rangle$ defined over the automaton $\mathcal{T}=\langle \Sigma_{\mathcal{T}}, R_{\mathcal{T}}, B_{\mathcal{T}}, Q_{\mathcal{T}}, $ $ \Delta_{\mathcal{T}}, Q_{\mathcal{T}}^0, F_{\mathcal{T}} \rangle$ a finite internal  run $\rho_b^{f\ast}$ over a word $v \in \Sigma^\ast$ is a finite run  of the finite state automaton $\model^\prime=\langle \Sigma_{\mathcal{T}}, R_{\mathcal{T}}, B_{\mathcal{T}}, Q_{\mathcal{T}} \cup F^{outR} , \Delta_{\mathcal{T}} \cup \Delta^{outR},$ $Q^{0}_s, F^{outR} \rangle$.
\end{mydef}

A finite internal run is associated to the IFSA corresponding to the replacement where the initial states include only the internal initial states of the replacement and the final states are the destinations of its outgoing transitions. For example, the run $\rho_{send_1}^{f\ast}($\{$ready$\}.\{$send$\}.\{$wait$\}.\{$timeout$\}. \{$fail$\}), such that $\rho_{send_1}^{f\ast}(0)=q_{14}$, $\rho_{send_1}^{f\ast}(1)=q_{15}$, $\rho_{send_1}^{f\ast}(2)=q_{16}$, $\rho_{send_1}^{f\ast}(3)=q_{17}$, $\rho_{send_1}^{f\ast}(4)=q_{18}$, $\rho_{send_1}^{f\ast}(5)=send_2$, is a finite internal run of the replacement presented in Figure~\ref{Fig:send1Replacement}.

\begin{mydef} [Infinite Internal Run] Given a replacement $\mathcal{R}=\langle \mathcal{T}, $ $\Delta^{inR}, \Delta^{outR} \rangle$ defined over the automaton $\mathcal{T}=\langle \Sigma_{\mathcal{T}}, R_{\mathcal{T}}, B_{\mathcal{T}}, Q_{\mathcal{T}},$ $ \Delta_{\mathcal{T}}, Q_{\mathcal{T}}^0, F_{b} \rangle$ a infinite internal run $\rho^{i\omega}_b$ over a word $v \in \Sigma^\omega$ is an infinite run  of the (Incomplete) B\"uchi automaton $\model^\prime=\langle \Sigma_{\mathcal{T}}, R_{\mathcal{T}}, B_{\mathcal{T}}, Q_{\mathcal{T}}, $ $\Delta_{\mathcal{T}},Q_{\mathcal{T}}^0, F_{\mathcal{T}} \rangle$.
\end{mydef}

An infinite internal run refers to the IBA obtained from the automaton $\mathcal{T}$  where the initial and accepting states include only the  initial and accepting states of the automaton associated with the replacement. For example, the infinite internal run $\rho^{i\omega}_{send_1}($\{$ready$\}.\{$wait$\}$^\omega)$ is a function such that $\rho^{i\omega}_{send_1}(0)=q_{14}$, and $\forall i >1, \rho^{i\omega}_{send_1}(i)=q_{15}$.

\begin{mydef}[Finite External Run] 
\label{def:finiteExternalRun} Given a replacement $\mathcal{R}=\langle \mathcal{T},$ $ \Delta^{inR},$ $  \Delta^{outR} \rangle$ defined over the automaton $\mathcal{T}=\langle \Sigma_{\mathcal{T}}, R_{\mathcal{T}}, B_{\mathcal{T}}, Q_{\mathcal{T}},$ $ \Delta_{\mathcal{T}}, Q_{\mathcal{T}}^0, F_{\mathcal{T}} \rangle$ a finite external  run $\rho_b^{e\ast}$ over a word $v \in \Sigma^\ast$ is a finite run  of the finite state automaton $\model^\prime=\langle \Sigma_{\mathcal{T}}, R_{\mathcal{T}}, B_{\mathcal{T}}, Q_{\mathcal{T}} \cup Q^{0inR} \cup F^{outR}, \Delta_{\mathcal{T}}\cup \Delta^{inR} \cup \Delta^{outR},$ $Q^{0inR}, F^{outR} \rangle$.
\end{mydef}

A finite external run refers to the IFSA obtained from the automaton $\mathcal{T}$ where the initial and accepting states include only the sources and the destinations of the incoming and outgoing transitions, respectively. For example, the finite external run $\rho_{send_1}^{e\ast}($\{$start$\}.\{$send$\}.\{$wait$\}.\{$timeout$\}. \{$fail$\}$)$ is a function such that $\rho_{send_1}^{e\ast}(0)=q_1$, $\rho_{send_1}^{e\ast}(1)=q_{15}$, $\rho_{send_1}^{e\ast}(2)=q_{16}$, $\rho_{send_1}^{e\ast}(3)=q_{17}$, $\rho_{send_1}^{e\ast}(4)=q_{18}$ and $\rho_{send_1}^{e\ast}(5)=send_2$.

\begin{mydef} [Infinite External Run] Given a replacement $\mathcal{R}=\langle \mathcal{T}, $ $\Delta^{inR}, \Delta^{outR} \rangle$ defined over the automaton $\mathcal{T}=\langle \Sigma_{\mathcal{T}}, R_{\mathcal{T}}, B_{ \mathcal{T}}, Q_{\mathcal{T}}, $ $\Delta_{\mathcal{T}}, Q_{\mathcal{T}}^0, F_{\mathcal{T}} \rangle$ a infinite external run $\rho^{e\omega}_b$ over a word $v \in \Sigma^\omega$ is an infinite run  of the (Incomplete) B\"uchi automaton $\model^\prime_{b}=\langle \Sigma_{\mathcal{T}}, R_{\mathcal{T}}, B_{\mathcal{T}},$ $ Q_{\mathcal{T}} \cup Q^{0inR}, \Delta_{\mathcal{T}} \cup \Delta^{inR}, Q^{0inR}, F_{\mathcal{T}} \rangle$.
\end{mydef}

An infinite external run refers to the IBA obtained from the automaton $\mathcal{T}$ where the initial states include the source states of the incoming transitions and the accepting states contains only the accepting states of $\mathcal{T}$.
For example, the infinite external run $\rho^{e\omega}_{send_1}($\{$start$\}.\{$wait$\}$^\omega)$ is a function such that  $\rho^{e\omega}_{send_1}(0)=q_1$ and $\forall i \geq 1, \rho^{e\omega}_{send_1}(i)=q_{15}$.

Given the four types of runs previously described, which are defined over IFSA and IBA, it is possible to distinguish between the three types of finite/infinite runs described in Sections~\ref{sec:IncompleteFSA} and \ref{sec:IncompleteBA}: \emph{definitely accepting, possibly accepting and not accepting}. For example, the replacement presented in Figure~\ref{Fig:send1Replacement} contains two types of \emph{definitely accepting infinite} runs. The infinite \emph{internal} runs involve the states $q_{14}$ and $q_{15}$, i.e., they recognize all the words in the form \{$booting$\}$^\ast$.\{$ready$\}.\{$wait$\}$^\omega$. The infinite \emph{external} runs involve the states $q_1$ and $q_{15}$ and recognize all the words in the form \{$start$\}.\{$wait$\}$^\omega$. Furthermore, the replacement contains two types of \emph{possibly accepting finite} runs. The finite \emph{internal} possibly accepting runs includes all the runs which involve the states $q_{14}$, $q_{15}$, $q_{16}$, $q_{17}$ and $q_{18}$ or $q_{19}$, respectively. The finite \emph{external} possibly accepting runs includes all the runs which involve the states $q_{1}$, $q_{15}$, $q_{16}$, $q_{17}$ and $q_{18}$  or $q_{19}$, respectively.

Let us now discuss the language recognized by a replacement. The replacement $\mathcal{R}$ \emph{internally definitely accepts} the finite word $v \in \Sigma^\ast$ if and only if there exists an internal finite definitely accepting run of $\mathcal{R}$ on $v$. The language of the finite words internally definitely accepted by the replacement $\mathcal{R}$ is indicated as $\mathcal{L}^{i\ast}(\mathcal{R})$. The replacement $\mathcal{R}$ \emph{externally definitely accepts} the finite word $v \in \Sigma^\ast$ if and only if there exists an external finite definitely accepting run of $\mathcal{R}$ on $v$. The language of the finite words externally definitely accepted by the replacement $\mathcal{R}$ is indicated as $\mathcal{L}^{e\ast}(\mathcal{R})$. The replacement $\mathcal{R}$ \emph{internally definitely accepts} the infinite word $v \in \Sigma^\omega$ if and only if there exists an internal infinite definitely accepting run of $\mathcal{R}$ on $v$. The language of the infinite words internally definitely accepted by the replacement $\mathcal{R}$ is indicated as $\mathcal{L}^{i\omega}(\mathcal{R})$. The replacement $\mathcal{R}$ \emph{externally definitely accepts} the infinite word $v \in \Sigma^\omega$ if and only if there exists an external infinite definitely accepting run of $\mathcal{R}$ on $v$. The language of the infinite words externally definitely accepted by the replacement $\mathcal{R}$ is indicated as $\mathcal{L}^{e\omega}(\mathcal{R})$.

Let us now consider \emph{possibly accepting words}. The replacement $\mathcal{R}$ \emph{internally possibly accepts} the finite word $v \in \Sigma^\ast$ if and only if there exists an internal possibly finite accepting run of $\mathcal{R}$ on $v$. The language of the finite words internally possibly accepted by the replacement $\mathcal{R}$ is indicated as $\mathcal{L}_p^{i\ast}(\mathcal{R})$. The replacement $\mathcal{R}$ \emph{externally possibly accepts} the finite word $v \in \Sigma^\ast$ if and only if there exists an external finite possibly accepting run of $\mathcal{R}$ on $v$. The language of the finite words externally possibly accepted by the replacement $\mathcal{R}$ is indicated as $\mathcal{L}_p^{e\ast}(\mathcal{R})$. The replacement $\mathcal{R}$ \emph{internally possibly accepts} the infinite word $v \in \Sigma^\omega$ if and only if there exists an internal possibly infinite accepting run of $\mathcal{R}$ on $v$. The language of the infinite words internally possibly accepted by the replacement $\mathcal{R}$ is indicated as $\mathcal{L}_p^{i\omega}(\mathcal{R})$. The replacement $\mathcal{R}$ \emph{externally possibly accepts} the infinite word $v \in \Sigma^\omega$ if and only if there exists an external infinite possibly accepting run of $\mathcal{R}$ on $v$. The language of the infinite words externally possibly accepted by the replacement $\mathcal{R}$ is indicated as $\mathcal{L}_p^{e\omega}(\mathcal{R})$.

\begin{mydef}[Sequential composition]
\label{def:plugginarefinement}  
Given an \IBA\ $\model=\langle \Sigma_{\model}, R_{\model},$ $ B_{\model},Q_ {\model}, \Delta_{\model},$ $ Q_{\model}^0, F_{\model} \rangle$ and the replacement $\mathcal{R}=\langle \mathcal{T}, \Delta^{inR},$ $ \Delta^{outR} \rangle$ of the box $b \in B_{\model}$, the sequential composition $\model\Join\mathcal{R}$ is an \IBA\ $\langle \Sigma_{\model\Join\mathcal{R}}, R_{\model\Join\mathcal{R}}, B_{\model\Join\mathcal{R}},$ $Q_{\model\Join\mathcal{R}}, \Delta_{\model\Join\mathcal{R}},$ $ Q_{\model\Join\mathcal{R}}^0, F_{\model\Join\mathcal{R}} \rangle$ of $\model$ that satisfies the following conditions:
\begin{enumerate}
\item \label{replacementAlphabet} $\Sigma_{\model\Join\mathcal{R}}=\Sigma_\model \cup \Sigma_{\mathcal{T}}$;
\item  \label{replacementRegular} $R_{\model\Join\mathcal{R}}=R_{\model} \cup R_{\mathcal{T}}$;
\item \label{replacementTransparent} $B_{\model\Join\mathcal{R}}=B_{\model} \setminus \left \{ b\right \} \cup B_{\mathcal{T}}$;
\item \label{replacementStates} $Q_{\model\Join\mathcal{R}}=R_{\model\Join\mathcal{R}} \cup B_{\model\Join\mathcal{R}}$;
\item \label{transitionRelation}$\Delta_{\model\Join\mathcal{R}}=(\Delta_{\model} \setminus 
\left \{ (q_{\model}, a, q_{\model}^\prime)  \in \Delta_\model \mid q_{\model} =b \lor q_{\model}^\prime =b    \right \} )\cup \Delta_{\mathcal{T}}  \cup \Delta^{inR} \cup \Delta^{outR}  $;
\item $Q_{\model\Join\mathcal{R}}^0=(Q_\model^0 \cup Q^0_{\mathcal{T}}) \cap Q_{\model\Join\mathcal{R}}$;
\item $F_{\model\Join\mathcal{R}}=(F_\model \cup F_{\mathcal{T}}) \cap Q_{\model\Join\mathcal{R}}$.
\end{enumerate}
\end{mydef}

Definition~\ref{def:plugginarefinement} condition~\ref{replacementAlphabet}  specifies that the alphabet of the refinement $\model\Join\mathcal{R}$ is the union of the alphabet of the original \IBA\ $\model$ and the alphabet of the automaton $\mathcal{T}$ associated with the replacement $\mathcal{R}$.
Definition~\ref{def:plugginarefinement} condition~\ref{replacementRegular} specifies that the set of regular states of $\model\Join\mathcal{R}$ is the union of the set of the regular states of $\model$ and the set of the regular states of the automaton $\mathcal{T}$ associated with the replacement $\mathcal{R}$.
Definition~\ref{def:plugginarefinement} condition~\ref{replacementTransparent} specifies that the set of boxes of $\model\Join\mathcal{R}$ is the union of the set of the boxes of $\model$, with the exception of the box $b$ which is refined, and the set of the boxes of the automaton $\mathcal{T}$ associated with the replacement $\mathcal{R}$.
Definition~\ref{def:plugginarefinement} condition~\ref{replacementStates} specifies the set of the states of $\model\Join\mathcal{R}$ which corresponds to the union of its regular and box states. Note that the box $b$ is not contained into $Q_{\model\Join\mathcal{R}}$.
Definition~\ref{def:plugginarefinement} condition~\ref{transitionRelation} specifies the set of the transitions of $\model\Join\mathcal{R}$. The transitions include all the transitions of the original model $\Delta_\model$ with the exception of the transitions that reach and leave the box $b$, all the transitions $\Delta_{\mathcal{T}}$ of the automaton the automaton $\mathcal{T}$ associated with the replacement and its incoming and outgoing transitions $\Delta^{inR}$ and $\Delta^{outR}$.
The set $Q_{\model\Join\mathcal{R}}^0$ of the initial states of $\model\Join\mathcal{R}$ includes all the initial states $Q_\model^0$ of the \IBA\ and the initial states $Q^0_{\mathcal{T}}$ of the automaton $\mathcal{T}$ associated with the replacement.
 The intersection with the set $Q_{\model\Join\mathcal{R}}$ is computed to remove the box $b$ (if present).
The set $F_{\model\Join\mathcal{R}}$ of the accepting states of $\model\Join\mathcal{R}$ include all the accepting states $F_\model$ of the \IBA\ and the accepting states $F_{\mathcal{T}}$ of  the automaton $\mathcal{T}$ associated with  the replacement. 
As previously, the intersection with the set $Q_{\model\Join\mathcal{R}}$  removes the box $b$ (if present).

\begin{theorem}[Refinement Preservation]
\label{th:RefinementPreservation}  
Given a model $\model=\langle \Sigma_{\model}, R_{\model},$ $ B_{\model},Q_ {\model}, \Delta_{\model},$ $ Q_{\model}^0, F_{\model} \rangle$ and a replacement $\mathcal{R}=\langle \mathcal{T}, \Delta^{inR},$ $ \Delta^{outR} \rangle$ 
 which refers to one of its boxes $b$, $\model \preceq \model\Join\mathcal{R}$.
\end{theorem}

\begin{proof}
To prove that $\model \preceq \model\Join\mathcal{R}$ we must define a refinement relation $\Re$ which satisfies the conditions specified in Definition~\ref{def:refinement}. 

The set of initial states $ Q^0_{\model\Join\mathcal{R}}$ contains the initial states of $\model$ (with the exception of the refined box $b$) and the initial states of the automaton corresponding to replacement $\mathcal{R}$. 
It is possible to associate to each initial state of $\model$ (with the exception of the refined box $b$) the corresponding state of $\model$ and to each initial state of the replacement  $\mathcal{R}$ the box $b$. 
Note that a replacement $\mathcal{R}$ can contain an initial state only if $b$ is initial for $\model$. 
This construction guarantees that the relation $\Re$ satisfies the conditions~\ref{def:initianlNinitialM} and \ref{def:initialRegularMinitialRegularN} of the Definition~\ref{def:refinement}. 
Conditions~\ref{def:refinementRegularMRegularN},~\ref{def:stateNstateM},~\ref{def:boxNboxM},~\ref{def:acceptingRegularMacceptingRegularN},~\ref{def:initialRegularMinitialRegularN} and~\ref{def:acceptingRegularMacceptingRegularN} can be satisfied in a similar way, i.e., by associating the box/regular states of $\model \preceq \model\Join\mathcal{R}$ to the corresponding state of the model or to the states of the box $b$ that is refined. 
 Let us finally analyze conditions~\ref{def:refinement6} and ~\ref{def:refinement7}.
 Each transition $\Delta_\model$ whose destination is not a box can be associated with the corresponding transition of the model, which makes~\ref{def:refinement6} trivially satisfied. The transitions whose destinations are the box $b$ can be associated with the corresponding transitions in $\Delta^{inR}$. 
 Note that Definition~\ref{replacement} forces each incoming/outgoing transition of a box to have at least a corresponding incoming/outgoing transition inside the replacement. 
 Let us finally consider the outgoing transition of the box $b$ of $\model$. 
 Each outgoing transition can be associated with the corresponding outgoing transition in $\Delta^{outR}$. 
 The same procedure can be applied to satisfy the condition \ref{def:refinement7}. Note that, each transition in $\Delta_{\mathcal{T}}$ is associated with the box $b$. 
 By following this procedure the refinement relation $\Re$  satisfies the conditions specified in  Definition~\ref{def:refinement} by construction, therefore $\model \preceq \model\Join\mathcal{R}$ is satisfied. 
\end{proof}

\begin{mydef}[Replacement refinement]
\label{def:replacementrefinement}  Let $\wp_{\model}$ the set of all possible replacements. 
A replacement $\mathcal{R}_\refinement= \langle \mathcal{T}_\refinement, \Delta^{inR}_\refinement,$ $ \Delta^{outR}_\refinement \rangle$ is a refinement of a replacement $\mathcal{R}_\model=\langle \mathcal{T}_\model, \Delta^{inR}_\model,$ $ \Delta^{outR}_\model \rangle$, i.e., $\mathcal{R}_\model \preceq \mathcal{R}_\refinement$, iff:
\begin{enumerate}
\item \label{def:replacementrefinement1} $\mathcal{T}_\model \preceq  \mathcal{T}_\refinement$, through the relation $\Re$;
\item \label{def:replacementrefinement2} for all $(q_\model, a, q^\prime_\model) \in \Delta^{inR}_\model$ there exists $(q_\refinement, a, q^\prime_\refinement) \in \Delta^{inR}_\refinement$,  such that $(q^\prime_\model, q^\prime_\refinement) \in \Re$.
\item \label{def:replacementrefinement3} for all $(q_\model, a, q^\prime_\model) \in \Delta^{outR}_\model$ there exists $(q_\refinement, a, q^\prime_\refinement) \in \Delta^{outR}_\refinement$,  such that $(q_\model, q_\refinement) \in \Re$.
\item \label{def:replacementrefinement4} for all $(q_\refinement, a, q^\prime_\refinement) \in \Delta^{inR}_\refinement$ there exists a unique $(q_\model, a, q^\prime_\model) \in \Delta^{inR}_\model$,  such that $(q^\prime_\model, q^\prime_\refinement) \in \Re$.
\item \label{def:replacementrefinement5} for all $(q_\refinement, a, q^\prime_\refinement) \in \Delta^{outR}_\refinement$ there exists a unique $(q_\model, a, q^\prime_\model) \in \Delta^{outR}_\model$,  such that $(q_\model, q_\refinement) \in \Re$.
\end{enumerate}
\end{mydef}

\begin{theorem}[Plugging principle for refinement]
\label{th:pluggingPrinciple} If $\mathcal{R}\preceq \mathcal{R}^\prime$, then $ \model \Join \mathcal{R} \preceq \model \Join \mathcal{R}^\prime$.
\end{theorem}
\begin{proof} 
It is sufficient to construct a relation $\Re_{\model \Join \mathcal{R},  \model \Join \mathcal{R}^\prime}$ between the states of $\model \Join \mathcal{R}$ and the states of $ \model \Join \mathcal{R}^\prime$ that satisfies the conditions specified in Definition~\ref{def:refinement}.
Let's consider the relation $\Re_{ \model \Join \model}$ between the states of $\model$ constructed as specified in Lemma~\ref{lem:reflexive}, and 
the relation $\Re_{ \mathcal{R}  \Join  \mathcal{R}^\prime }$ between the states the automaton $ \mathcal{T}$ of the replacement $\mathcal{R}$ and the automaton $ \mathcal{T}^\prime$ of the automaton $ \mathcal{R}^\prime$ which must exist by Definition~\ref{def:replacementrefinement}.
We prove that the relation $\Re_{\model \Join \mathcal{R},  \model \Join \mathcal{R}^\prime}$, such that  $\Re_{\model \Join \mathcal{R},  \model \Join \mathcal{R}^\prime}=\Re_{ \model \Join \model} \cup \Re_{  \mathcal{T}  \Join  \mathcal{T}^{\prime} } \setminus \left \{ (b,b) \right \}$ satisfies  the conditions specified in Definition~\ref{def:refinement} implying that $ \model \Join \mathcal{R} \preceq \model \Join \mathcal{R}^\prime_b$.
The conditions~\ref{def:refinementRegularMRegularN},~\ref{def:stateNstateM},~\ref{def:initianlNinitialM},~\ref{def:boxNboxM},~\ref{def:acceptingNacceptingM},~\ref{def:initialRegularMinitialRegularN} and~\ref{def:acceptingRegularMacceptingRegularN} are satisfied by construction of the relations $\Re_{ \model \Join \model} $ and $\Re_{ \mathcal{T}  \Join  \mathcal{T}^{\prime} } $. 
We need to prove that also conditions~\ref{def:refinement6} and~\ref{def:refinement7} are satisfied.
Each transition $(q_{ \model}, a,  q^\prime_\model) \in \Delta_{\model}$, whose source or destination is not the box $b$,  satisfies the first statements of the conditions~\ref{def:refinement6} and~\ref{def:refinement7} since the transition is also present in refinement. 
Similarly, each transition $(q_{\mathcal{T}}, a,  q^\prime_{\mathcal{T}}) \in \Delta_{\mathcal{T}}$ satisfies the first statements of the conditions~\ref{def:refinement6} and~\ref{def:refinement7} since  $\mathcal{R} \preceq \mathcal{R}^\prime$ and as a consequence $ \mathcal{T} \preceq  \mathcal{T}^\prime$. 
Let's finally consider each transition $(q_{ \model}, a,  q^\prime_\model) \in \Delta^{inR}$,  conditions~\ref{def:replacementrefinement2} and~\ref{def:replacementrefinement4} of Definition~\ref{def:replacementrefinement}  imply the satisfaction of conditions~\ref{def:refinement6} and~\ref{def:refinement7} of  Definition~\ref{def:refinement}. 
Similarly for each transition $(q_{ \model}, a,  q^\prime_\model) \in \Delta^{outR}$ conditions~\ref{def:replacementrefinement3} and~\ref{def:replacementrefinement5} of Definition~\ref{def:replacementrefinement}  imply the satisfaction of conditions~\ref{def:refinement6} and~\ref{def:refinement7} of  Definition~\ref{def:refinement}. 
Thus, $ \model \Join \mathcal{R} \preceq \model \Join \mathcal{R}^\prime$.
\end{proof}

\begin{theorem}[Sequential composition preserves the refinement relation]
\label{th:seqRefinement}
Given two \IBAS\ $\model$ and $\mathcal{K}$ and two replacements $\mathcal{R}$ and $\mathcal{R}^\prime$ for the black box state $b$ of $\model$,
\begin{itemize}
\item if $\mathcal{K}  \preceq \model\Join\mathcal{R} $ and $\mathcal{R}\preceq \mathcal{R}^\prime$ then $\mathcal{K} \preceq \model\Join\mathcal{R}^\prime$
\end{itemize}
\end{theorem}
\begin{proof} 
It follows from the fact that the refinement relation is transitive and by the plugging principle for refinement. For the plugging principle (Theorem~\ref{th:pluggingPrinciple}) $\model\Join\mathcal{R}\preceq  \model\Join\mathcal{R}^\prime$. For the transitive relation Lemma~\ref{lem:transitivity}, if $\mathcal{K}  \preceq \model\Join\mathcal{R} $ and $\model\Join\mathcal{R}\preceq  \model\Join\mathcal{R}^\prime$ then $\mathcal{K} \preceq \model\Join\mathcal{R}^\prime $.
\end{proof}

\section{Modeling the claim}
\label{sec:modelingclaim}
When a system is incomplete a different semantic for the formulae of interest, such as a \emph{three-valued} semantic, can be considered. 
Given a formula $\phi$ (expressed in some logic)  and an \IBA\  $\mathcal{M}$ three truth values can be associated to the satisfaction of formula  $\phi$ in model $\mathcal{M}$: \emph{true}, \emph{false} and \emph{unknown} (maybe).  
Whenever a formula $\phi$ is true or false its satisfaction does not depend on the incomplete parts present in the model $\mathcal{M}$. 
We say that $\phi$ is definitely satisfied and not satisfied, respectively.
In the fist case, all the behaviors of the system (including the one that the system may exhibit) satisfy the formula $\phi$. 
In the second case, there exists a behavior of  $\mathcal{M}$, which does not depend on the incomplete parts which violates $\phi$. 
In the third case the satisfaction of $\phi$ depends on the incomplete parts, i.e., $\phi$ is possibly satisfied.
The three-valued semantic of Linear Time Temporal Logic (LTL) formulae specifies when LTL formulae are definitely satisfied, possibly satisfied or not satisfied by the \IBA\ $\mathcal{M}$.

\subsection{Three value Linear Time Temporal Logic semantic}
Given an LTL formula $\phi$ and an \IBA\ $\mathcal{M}$ the semantic function $\| \mathcal{M}^\phi \|$  associates to $\mathcal{M}$ and $\phi$ one of the true values true ($T$), false ($F$) and unknown ($?$). 
Whenever a formula is true, it is true in all the implementations of $\mathcal{M}$, i.e., it does not exists any replacement of the boxes that makes $\phi$ violated. 
If the formula is false, there exists a behavior of  $\mathcal{M}$, which does not depend on how the system is refined which violates the property of interest. 
Thus, all the implementations of $\mathcal{M}$ will make $\phi$ not satisfied. In the third case the satisfaction of $\phi$ depends on the replacements of the boxes of $\mathcal{M}$. 
This type of three value semantic is also known in literature as \emph{inductive semantic}~\cite{wei2009mixed} and is different from the thorough semantic defined in~\cite{generalized}.

\begin{mydef}[Three value LTL semantic over \IBA]
\label{def:threevalueltl}
Given an \IBA\ $\mathcal{M}$ and the LTL formula $\phi$:
\begin{enumerate}
\item \label{semantic:satisfaction}  $\| \mathcal{M}^\phi \|=  T$ if and only if for all $v \in (\mathcal{L}^\omega(\mathcal{M})  \cup \mathcal{L}_p^\omega(\mathcal{M}))$, $v \models \phi  $  
\item \label{semantic:notsatisfaction}   $\| \mathcal{M}^\phi \|=  F$ if and only if exists $v \in \mathcal{L}^\omega(\mathcal{M})$ such that $v \not \models \phi $
\item \label{semantic:possiblysatisfaction}   $ \| \mathcal{M}^\phi \|=\  ? $ if and only for all $v \in \mathcal{L}^\omega(\mathcal{M})$,  $v \models \phi$ and  there exists $ u \in \mathcal{L}_p^\omega(\mathcal{M})$ such that $u \not \models \phi     $
\end{enumerate}
\end{mydef}

A formula $\phi$ is true in the model $\mathcal{M}$ if and only if every word $v$ that is in the language definitely accepted or possibly accepted by the automaton satisfies the claim $\phi$ (Definition~\ref{def:threevalueltl}, condition~\ref{semantic:satisfaction}). 
A formula $\phi$ is false in the model $\mathcal{M}$ if and only if there exists word $v$ that is in the language definitely accepted by the \IBA\ that does not satisfy the claim $\phi$ (Definition~\ref{def:threevalueltl}, condition~\ref{semantic:notsatisfaction}). 
A formula $\phi$ is possibly satisfied in the model $\mathcal{M}$ if and only if there exists word $u$ that is in the language possibly accepted by the \IBA\ that does not satisfy the claim $\phi$, but all the words $v$ in the language definitely accepted by $\mathcal{M}$ satisfy the formula $\phi$ (Definition~\ref{def:threevalueltl}, condition~\ref{semantic:possiblysatisfaction}). 
For example, the property  $\phi=\LTLglobally( send \rightarrow \LTLfinally success)$  is possibly satisfied by the model described in Figure~\ref{Fig:IFSAExample} since there exists a word \{$start$\}.\{$send$\}.\{$fail$\}.\{$fail$\}.\{$abort$\}$^\omega$ in the possible accepted language which does not satisfy the formula and there are no words in the definitely accepted language. 

\begin{theorem} [Refinement preservation of LTL properties] Given an \IBA\ $\mathcal{M}$ and its refinement $\mathcal{N}$, such that $\mathcal{M} \preceq \mathcal{N}$, Then:
\begin{enumerate}
\item   \label{ltlpreservation1}  if $ \| \mathcal{M}^\phi \|=  T$ then $ \| \mathcal{N}^\phi \|=  T  $; 
\item \label{ltlpreservation2}    if $\| \mathcal{M}^\phi \|=  F$ then $ \| \mathcal{N}^\phi \|=  F  $.  
\end{enumerate}
\end{theorem}

\begin{proof}
Let us first consider condition~\ref{ltlpreservation1}. 
The proof is done by contradiction. 
Assume that $\| \mathcal{M}^\phi \|=  T $ and $ \| \mathcal{N}^\phi \| \not=  T$. If $ \| \mathcal{N}^\phi \|=  F$, by Definition~\ref{def:threevalueltl} $\exists v \in \mathcal{L}^\omega(\mathcal{N}),\ v \not \models \phi$. 
By Theorem~\ref{th:languagePreservation} $v \in \mathcal{L}_p^\omega(\mathcal{M})$, i.e., $v$ must be in the possible recognized language of $\mathcal{M}$, or $v \in \mathcal{L}^\omega(\mathcal{M})$. 
This condition makes the condition~\ref{semantic:satisfaction} of the Definition~\ref{def:threevalueltl}  not satisfied, i.e., it exists a word that definitely satisfies $\phi$ and is possibly recognized or recognized by $\mathcal{M}$, and the hypothesis $\| \mathcal{M}^\phi \|=  T $ contradicted. 
If $ \| \mathcal{N}^\phi \|=  ?$, by Definition~\ref{def:threevalueltl} $\exists v \in \mathcal{L}_p^\omega(\mathcal{N}),\ v \not \models \phi$. 
By Theorem~\ref{th:languagePreservation} $v \in \mathcal{L}_p^\omega(\mathcal{M})$, i.e., $v$ must be in the possibly recognized language of $\mathcal{M}$. 
Again, this condition makes the condition~\ref{semantic:satisfaction} of the Definition~\ref{def:threevalueltl}  not satisfied, i.e., it exists a word that satisfies $\phi$ and is possibly recognized by $\mathcal{M}$ and thus the hypothesis $\| \mathcal{M}^\phi \|=  T $ is contradicted. 

Let us now consider condition~\ref{ltlpreservation2}. 
Since $\| \mathcal{M}^\phi \|=  F $ from Definition~\ref{def:threevalueltl} condition~\ref{semantic:notsatisfaction} it must exists a word $v \in \mathcal{L}^\omega(\mathcal{M})$ that does not satisfy $\phi$. 
By Definition~\ref{th:languagePreservation} condition~\ref{satisfiedPreservation} $v \in \mathcal{L}^\omega(\mathcal{N})$. 
By Theorem~\ref{def:threevalueltl} condition~\ref{semantic:notsatisfaction} we can conclude that $\| \mathcal{N}^\phi \|=  F $.
\end{proof}

\subsection{Three value B\"uchi Automata semantic}
Given a BA $\mathcal{A}_\phi$ and an \IBA\ $\mathcal{M}$ which describes the model of the system, the semantic function $\| \mathcal{M}^{\mathcal{A}_\phi} \|$  associates to the model $\mathcal{M}$ and the property $\mathcal{A}_\phi$  one of true values true $T$, false $F$ and unknown $?$ depending on whether the model definitely satisfies, possibly satisfies or does not satisfy the claim specified by the BA $\mathcal{A}_\phi$.

\begin{mydef} [Three value BA semantic]
\label{def:threevalueBAsemantic}
Given and \IBASPartialExtended\ $\mathcal{M}$ and a BA $\mathcal{A}_\phi$ which specifies the definitely accepted behaviors of $\mathcal{M}$,
\begin{enumerate}
\item  \label{def:threevalueBAsemanticcond1} $\| \mathcal{M}^{\mathcal{A}_\phi} \|=  T$ iff $ \mathcal{L}^\omega(\mathcal{M}) \cup \mathcal{L}_p^\omega(\mathcal{M})  \subseteq \mathcal{L}(\mathcal{A}_\phi) $; 
\item \label{def:threevalueBAsemanticcond3} $\| \mathcal{M}^{\mathcal{A}_\phi} \|=  F$ iff $ \mathcal{L}^\omega(\mathcal{M}) \not \subseteq \mathcal{L}^\omega(\mathcal{A}_\phi)$
\item  \label{def:threevalueBAsemanticcond2} $\| \mathcal{M}^{\mathcal{A}_\phi} \|=\ ? $ iff $\mathcal{L}^\omega(\mathcal{M}) \subseteq \mathcal{L}^\omega(\mathcal{A}_\phi)$ and $\mathcal{L}_p^\omega(\mathcal{M}) \not \subseteq \mathcal{L}^\omega(\mathcal{A}_\phi)$  
\end{enumerate}
\end{mydef} 

Informally, a model $\mathcal{M}$  \emph{definitely satisfies} the claim expressed as a BA $\mathcal{A}_\phi$ if and only if Condition~\ref{def:threevalueBAsemanticcond1} is satisfied, i.e.,  all the behaviors of the model of the system, including possible behaviors,  are contained in the set of behaviors allowed by the property. 
Condition~\ref{def:threevalueBAsemanticcond3} specifies that a  model $\mathcal{M}$ \emph{does not satisfy} the claim expressed as a BA $\mathcal{A}_\phi$ if and only if there exists a behavior of the model which is not allowed by the property.
Finally, a  model $\mathcal{M}$ \emph{possibly satisfies} the claim expressed as a BA $\mathcal{A}_\phi$ if and only if the condition~\ref{def:threevalueBAsemanticcond2} is satisfied, i.e., all the definitely accepting behaviors of the model of the system are contained into the set of behaviors allowed by the property, but there exists a possible behavior which is not contained into the set of behaviors allowed by the property.

\begin{lemma}[Relation between BA and LTL Semantic]
\label{lem:BALTLRelation}  Given an LTL formula $\phi$ and the corresponding BA $\mathcal{A}_\phi$,
$\| \mathcal{M}^\phi \|$ if and only if  $\| \mathcal{M}^{\mathcal{A}_\phi} \|$.
\end{lemma}
\begin{proof}
The proof follows from the fact that the automaton $\mathcal{A}_\phi$ contains all the words that satisfy the claim $\phi$. Thus, asking for language containment as done in Definition~\ref{def:threevalueBAsemantic} corresponds with checking that all the words definitely accepted and possibly accepted by $\mathcal{M}$ satisfy the claim $\phi$ as done in Definition~\ref{def:threevalueltl}.
\end{proof}
Lemma~\ref{lem:BALTLRelation} allows to relate the satisfaction of LTL formulae with respect to BAs and it is necessary since the models and claims of interest must have compatible semantics~\cite{clarke1999model}.

\section{Checking \IBASExtended}
\label{sec:checkingIBAs}

\thispagestyle{empty}

The core of the envisaged development process is the \emph{development-analysis} cycle. 
During the \emph{development} designers refine an incomplete model $\model$ which describes the system up to some level of abstraction. 
At each development step, they produce a new replacement (increment) which describes the behavior of the system inside one of its black box states, leading to a new refined model $\refinement$, which may in turn contain incompleteness. 
When an increment is ready, developers analyze the properties of the refined model $\refinement$. 
If the model satisfies the designer's expectation, the development-analysis cycle is repeated, i.e., the development of the new increment is started. 

The verification of incomplete models offers three major benefits:
\begin{inparaenum}[\itshape a\upshape)]
\item instead of forcing the verification procedure to be performed at the end of the development process, it allows the system to be checked at the early stages of the design;
\item complex parts of the design can be encapsulated into unspecified (incomplete) parts (\emph{abstraction});
\item the location of design errors can be identified by sequentially narrowing portions of the system into incomplete parts.
\end{inparaenum}

Given an \IBA\ $\model$ and a LTL formula $\propertytext$, the incomplete model checking problem verifies whether the model definitely satisfies, possibly satisfies or does not satisfy property $\propertytext$, i.e., $\| \model^\propertytext \| $ is equal to true ($T$), false ($F$) or maybe ($?$). 
Given a LTL formula $\propertytext$, it is possible to transform the formula into a corresponding BA $\mathcal{A}_\propertytext$ and check  $\| \model^{\mathcal{A}_\propertytext} \| $. 
Since BAs are closed under intersection and complementation, it is possible to transform $\neg \propertytext$ into the corresponding automaton $\mathcal{A}_{\neg \propertytext}$ and to reformulate  the Conditions~\ref{def:threevalueBAsemanticcond1},~\ref{def:threevalueBAsemanticcond3} and \ref{def:threevalueBAsemanticcond2} of Definition~\ref{def:threevalueBAsemantic}  as: $(\mathcal{L}(\model) \cup \mathcal{L}_p(\model)) \cap \mathcal{L}(\mathcal{A}_{\neg \propertytext})=\emptyset$;   $\mathcal{L}(\model)  \cap \mathcal{L}(\mathcal{A}_{\neg \propertytext}) \not =\emptyset$; and $\mathcal{L}(\model)  \cap \mathcal{L}(\mathcal{A}_{\neg \propertytext})=\emptyset$ and $\mathcal{L}_p(\model)  \cap \mathcal{L}(\mathcal{A}_{\neg \propertytext}) \not =\emptyset$, respectively. 
However, to check these conditions,  it is necessary to redefine the behavior of the intersection operator ($\cap$) over an \IBA\ and a BA.

\subsection{The intersection automaton}
\label{sec:BAIBAIntersection}
This section describes how the intersection between an \IBA\ and a BA is computed. 
To exemplify the intersection between an \IBA\ and a BA we will consider the model $\model$ presented in Figure~\ref{fig:modelexample} and the automaton corresponding to the negation of the LTL claim $\propertytext=\LTLglobally (send \rightarrow \LTLfinally ( success))$ represented in Figure~\ref{fig:claimexample}.

\begin{figure*}[ht]
    \centering
    \subfloat[The \IBA\ that corresponds to model $\model$.]
     { \includegraphics[scale=0.7]{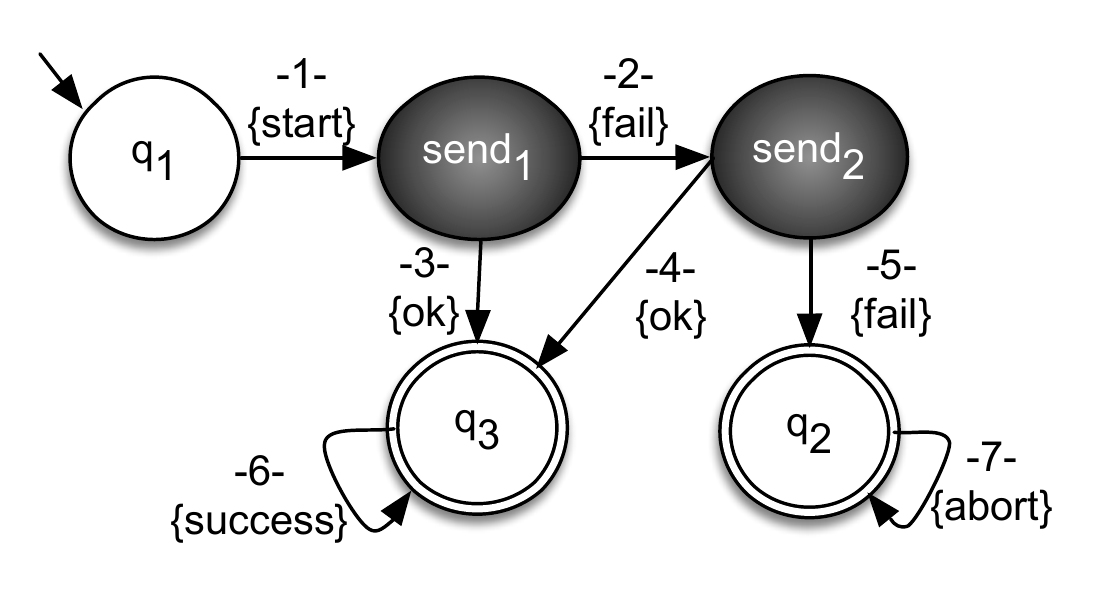}
    \label{fig:modelexample}}
    \subfloat[The BA $\mathcal{A}_{\neg \propertytext}$ of $\neg \propertytext$.]{
          \includegraphics[scale=0.7]{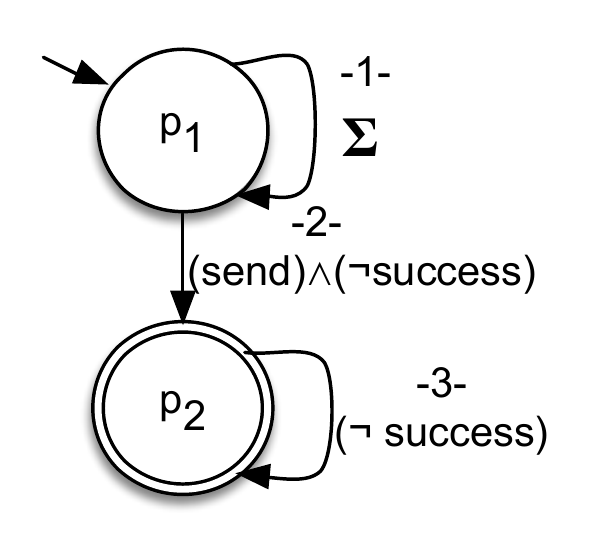}
              \label{fig:claimexample}}
    \caption{The \IBA\ and the BA used as examples in the description of the computation of the intersection automaton $\mathcal{I}$.}
\end{figure*}

\begin{mydef}[Intersection between an \IBA\ and a BA]
\label{def:intersection}
The intersection automaton $\mathcal{I}=\model \cap \mathcal{A}_{\neg \propertytext}$ between an \IBA\ $\model$ and a BA $\mathcal{A}_{\neg \propertytext}$ is the BA
$\mathcal{I}=\langle \Sigma_{\mathcal{I}}, Q_{\mathcal{I}},$ $ \Delta_{\mathcal{I}},  Q_{\mathcal{I}}^0,  F_{\mathcal{I}} \rangle$, such as:
\begin{itemize}
\item $\Sigma_{\mathcal{I}}=\Sigma_{\model} \cup \Sigma_{\mathcal{A}_{\neg \propertytext}}$ is the alphabet of $\mathcal{I}$; 
\item $Q_\mathcal{I}=((R_{\model} \times R_{\mathcal{A}_{\neg \propertytext}}) \cup (B_{\model} \times R_{\mathcal{A}_{\neg \propertytext}})) \times \left \{ 0,\ 1,\ 2\right \} $ is the set of states;
\item $\Delta_{\mathcal{I}}=\Delta_{\mathcal{I}}^c \cup \Delta_{\mathcal{I}}^p$ is the set of transitions of the intersection automaton.
 $\Delta_{\mathcal{I}}^c$ is the set of transitions $(\langle q_i, q^\prime_j, x \rangle, a,$ $ \langle q_m, q^\prime_n, y \rangle)$, such that $(q_i, a, q_m) \in \Delta_{\model}$ and  $(q_j^\prime, a, q_n^\prime) \in \Delta_{\mathcal{A}_{\neg \propertytext}}$.
 $\Delta_{\mathcal{I}}^p$ corresponds to the set of transitions $(\langle q_i, q^\prime_j, x \rangle, a,$ $ \langle q_m, q^\prime_n, y \rangle)$ where $q_i=q_m$ and $q_i \in B_{\model}$ and $(q^\prime_j, a, q^\prime_n) \in \Delta_{\mathcal{A}_{\neg \propertytext}}$.
Moreover,  each transition in $\Delta_\mathcal{I}$ must satisfy the following conditions:
\begin{itemize}
\item if $x=0$ and $q_m \in F_{\model}$, then $y=1$;
\item if $x=1$ and $q^\prime_n \in F_{\mathcal{A}_{\neg \propertytext}}$, then $y=2$;
\item if $x=2$, then $y=0$;
\item otherwise, $y=x$;
\end{itemize}
\item  $Q^0_\mathcal{I}=Q^0_{\model} \times Q^0_{\mathcal{A}_{\neg \propertytext}} \times \left \{ 0 \right \}$ is the set of initial states;
\item $F_\mathcal{I}=F_{\model} \times F_{\mathcal{A}_{\neg \propertytext}} \times \left \{ 2 \right \}$ is the set of accepting states.
\end{itemize}
\end{mydef}

The intersection $\mathcal{I}$ between the model $\model$, depicted in Figure~\ref{fig:modelexample}, and the BA $\mathcal{A}_{\neg \propertytext}$ of Figure~\ref{fig:claimexample}, that corresponds to the negation of the property $\propertytext$, is the BA described in Figure~\ref{fig:intersection}. 
The portions of the state space that contain mixed states associated with the black box states of the model $send_1$ and $send_2$ are surrounded by a dashed-dotted frame.

\begin{figure*}[htpb]
\centering
\includegraphics[scale=0.45]{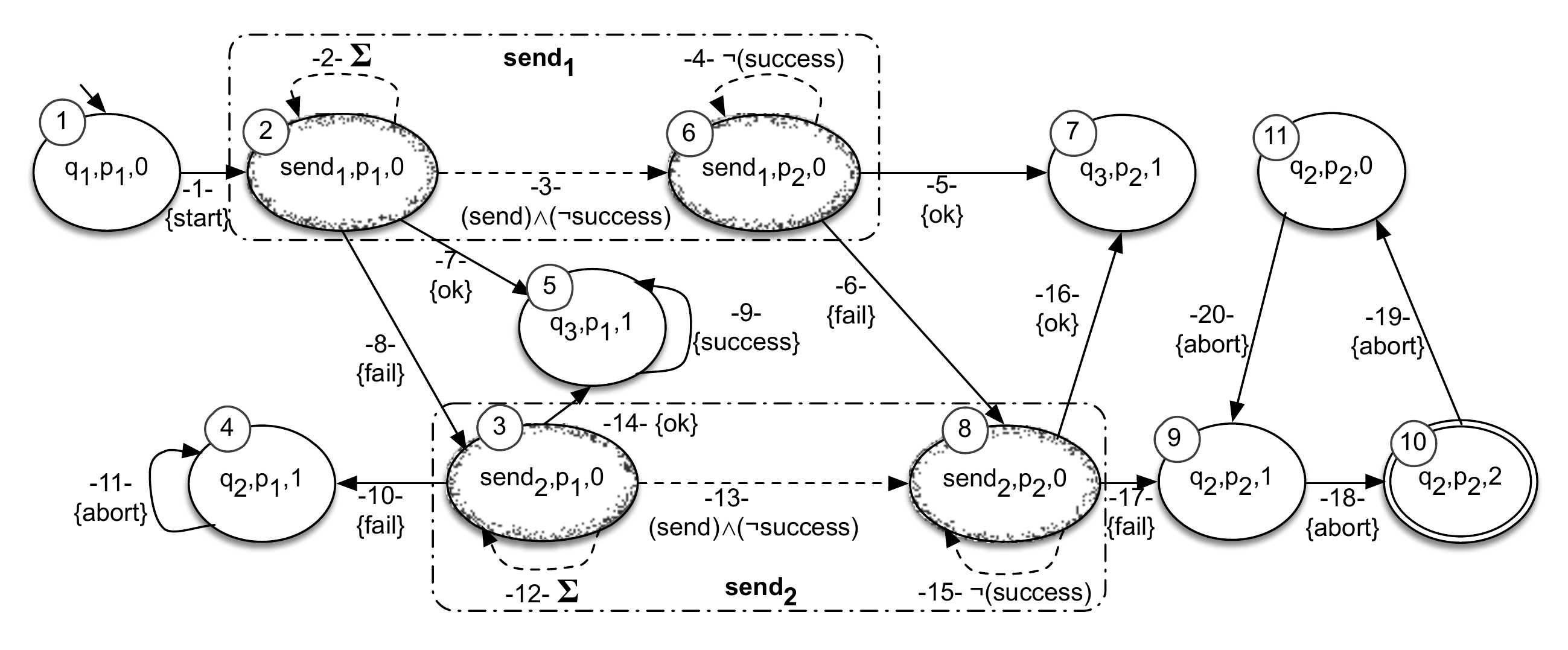}
\centering
\caption{The intersection automaton $\mathcal{I}$ between the \IBAPartialExtended\ $\model$ and the BA automaton $\mathcal{A}_{\neg \propertytext}$ which corresponds to the negation of the property $\propertytext$.}
\label{fig:intersection}
\end{figure*}

The alphabet $\Sigma_{\mathcal{I}}$ includes all the propositions of the alphabets of $\model$ and $\mathcal{A}_{\neg \propertytext}$.
The set $Q_\mathcal{I}$ is composed by the states obtained combining states of the automaton associated with the negation of the property  $\mathcal{A}_{\neg \propertytext}$ with regular states and boxes of the model $\model$. 
As in the classical intersection algorithm for  BAs~\cite{clarke1999model}, the labels $0$, $1$ and $2$ indicate that no accepting state is entered, at least one accepting state of $\model$ is entered, and  at least one accepting state of $\model$ and one accepting state of $\mathcal{A}_{\neg \propertytext}$ are entered, respectively. 
We define  $M_\mathcal{I}=B_{\model} \times R_{\mathcal{A}_{\neg \propertytext}}\times \left \{ 0,\ 1,\ 2\right \}$ as the set of \emph{mixed} states (graphically indicated in Figure~\ref{fig:intersection} with a stipple border), 
and $PR_\mathcal{I}=R_{\model} \times R_{\mathcal{A}_{\neg \propertytext}} \times \left \{ 0,\ 1,\ 2\right \}$ as the set of \emph{purely regular} states.
For example, state \raisebox{.5pt}{\textcircled{\raisebox{-.9pt} {1}}} is obtained by combining the state $q_1$ of $\model$ and $p_1$ of $\mathcal{A}_{\neg \propertytext}$. 
This state is initial and purely regular since both $q_1$ and $p_1$ are initials and regulars. Conversely, state \raisebox{.5pt}{\textcircled{\raisebox{-.9pt} {2}}} is mixed, since it is obtained by combining the box $send_1$ of $\model$ and the regular state $p_1$ of $\mathcal{A}_{\neg \propertytext}$.  

The transitions in $\Delta_{\mathcal{I}}^c$ are obtained by the synchronous execution of the transitions of $\model$ and the transitions of $\mathcal{A}_{\neg \propertytext}$. 
For example, the transition from \raisebox{.5pt}{\textcircled{\raisebox{-.9pt} {2}}} to \raisebox{.5pt}{\textcircled{\raisebox{-.9pt} {3}}} is obtained by combining the transition from $send_1$ to $send_2$ of $\model$ and the transition from $p_1$ to $p_1$ of $\mathcal{A}_{\neg \propertytext}$. 
The transitions in $\Delta_{\mathcal{I}}^p$ are, instead, obtained when a transition of $\mathcal{A}_{\neg \propertytext}$ synchronizes with a transition in the replacement of a box of $\model$. 
For example, the transition from \raisebox{.5pt}{\textcircled{\raisebox{-.9pt} {2}}} to \raisebox{.5pt}{\textcircled{\raisebox{-.9pt} {6}}} is performed when $\mathcal{A}_{\neg \propertytext}$ moves from $p_1$ to $p_2$ and the automaton $\model$ performs a transition in the replacement of the box $send_1$. 

The language recognized by $\mathcal{I}$ is the intersection of the language possibly recognized and recognized by $\model$ and the language recognized by $\mathcal{A}_{\neg \propertytext}$.
For example, the word \{$start$\}.\{$send$\}.\{$fail$\}.\{$fail$\}. \{$abort$\}$^\omega$ is possibly recognized by $\model$ and is recognized by $\mathcal{A}_{\neg \propertytext}$. Indeed, the system may fire the $start$ transition, then perform a transition in the refinement of the state $send_1$ which satisfies the condition $(send)\wedge(\neg success)$, then perform the two $fail$ transitions and finally perform the transition labeled with $abort$ an infinite number of times.

\begin{proposition}[Size of the intersection automaton]
\label{sec:sizeOfTheIntersectionAutomaton}
The intersection automaton $\mathcal{I}$ contains in the worst case $3 \cdot |Q_\model| \cdot  |Q_{\mathcal{A}_{\neg \propertytext}}|$ states and can be computed in time $\mathcal{O}(|\model|\cdot |\mathcal{A}_{\neg \propertytext}|)$.
\end{proposition}

\begin{lemma}[Intersection language]
\label{lem:intersection} 
 The intersection automaton $\mathcal{I}=\model \cap \mathcal{A}_{\neg \propertytext}$ between an \IBAPartialExtended\ $\model$ and a BA $\mathcal{A}_{\neg \propertytext}$ recognizes the language $(\mathcal{L}^\omega(\model) \cup \mathcal{L}^\omega_p(\model)) \cap \mathcal{L}^\omega(\mathcal{A}_{\neg \propertytext})$, i.e., $v \in \mathcal{L}(\mathcal{I}) \Leftrightarrow v \in (\mathcal{L}^\omega(\model) \cup \mathcal{L}^\omega_p(\model)) \cap \mathcal{L}^\omega(\mathcal{A}_{\neg \propertytext})$.
\end{lemma}
\begin{proof}
($\Rightarrow$) If $v \in \mathcal{L}(\mathcal{I})$, it must exists an accepting run $\rho^\omega$ in the intersection automaton which recognizes $v$. 
Since, by Definition~\ref{def:intersection}, $\mathcal{I}$ is a BA, for all $i>0$ $\rho^\omega(i)$ and $\rho^\omega(i+1)$ are states of $\rho^\omega$ if and only if $(\rho^\omega(i), v_i, \rho^\omega(i+1)) \in \Delta_{\mathcal{I}}$. Let us consider the two states of the model $\rho_{\model}^\omega(i)$ and $\rho_{\model}^\omega(i+1)$ associated with $\rho^\omega(i)$ and $\rho^\omega(i+1)$. 
Since there exists a transition $(\rho^\omega(i), v_i, \rho^\omega(i+1)) \in \Delta_{\mathcal{I}}$, by construction it must exists a transition  $(\rho_{\model}^\omega(i), v_i, \rho_{\model}^\omega(i+1)) \in \Delta_{\model}$ or $\rho_{\model}^\omega(i)=\rho_{\model}^\omega(i+1) \in B_{\model}$. 
In the first case, $v_i$ is recognized by a transition of the model, in the second case, it is recognized by a box. 
This condition must hold $\forall i \geq 0$.
Furthermore, since $\rho^\omega$ is accepting, an accepting state of the model must be entered infinitely often. 
It follows that $v \in (\mathcal{L}^\omega(\model) \cup \mathcal{L}^\omega_p(\model))$. 
The same idea can be applied with respect to the automaton $\mathcal{A}_{\neg \propertytext}$, which implies that $v \in \mathcal{L}(\mathcal{A}_{\neg \propertytext})$. Thus, $v \in (\mathcal{L}^\omega(\model) \cup \mathcal{L}^\omega_p(\model)) \cap \mathcal{L}^\omega(\mathcal{A}_{\neg \propertytext})$.

($\Leftarrow$) The proof is by contradiction. 
Imagine that there exists a word $v \not \in \mathcal{L}(\mathcal{I})$ which is in $(\mathcal{L}^\omega(\model) \cup \mathcal{L}^\omega_p(\model)) \cap \mathcal{L}^\omega(\mathcal{A}_{\neg \propertytext})$. 
Since $v \not \in \mathcal{L}(\mathcal{I})$, it is not recognized by $ \mathcal{L}(\mathcal{I})$, i.e., for every possible accepting run $\rho^\omega$ it must exists a character $v_i$ of $v$ such that $(\rho^\omega(i), v_i, \rho^\omega(i+1)) \not \in \Delta_{\mathcal{I}}$. 
Let us consider the corresponding states $\rho_{\model}^\omega(i)$, $\rho_{\model}^\omega(i+1)$ and $\rho_{\mathcal{A}_{\neg \propertytext}}^\omega(i)$, $\rho_{\mathcal{A}_{\neg \propertytext}}^\omega(i+1)$ of the model and of the claim, respectively. 
To make $(\rho^\omega(i), v_i, \rho^\omega(i+1)) \not \in \Delta_{\mathcal{I}}$ two cases are possible: $(\rho_{\mathcal{A}_{\neg \propertytext}}^\omega(i),  v_i,\rho_{\mathcal{A}_{\neg \propertytext}}^\omega(i+1)) \not \in \Delta_{\mathcal{A}_{\neg \propertytext}}$ or $(\rho_{\model}^\omega(i),  v_i,\rho_{\model}^\omega(i+1)) \not \in \Delta_{\model}$ and not $\rho^\omega(i)_{\model}=\rho^\omega(i+1)_{\model} \in B_{\model}$. 
However, in the first case, the condition implies that $v \not \in  \mathcal{L}^\omega(\mathcal{A}_{\neg \propertytext})$, while in the second,  $v \not \in (\mathcal{L}^\omega(\model) \cup \mathcal{L}^\omega_p(\model))$, thus $v \not \in (\mathcal{L}^\omega(\model) \cup \mathcal{L}^\omega_p(\model)) \cap \mathcal{L}^\omega(\mathcal{A}_{\neg \propertytext})$, contradicting the hypothesis.
\end{proof}

Given an infinite run $\rho^\omega$ of $\mathcal{I}$ associated with the infinite word $v \in \mathcal{L}(\mathcal{I})$, we may want to identify  the portions of the word $v$ recognized by each box $b \in B_{\model}$. Note that, these portions may include both finite words (i.e., finite portions of the words that are recognized inside the boxes) or infinite words (i.e., suffixes of words recognized inside \emph{accepting} boxes).

\begin{mydef}[Finite abstractions of a run]
\label{def:finiteRunAbstraction} 
Given an infinite run $\rho^\omega$ of $\mathcal{I}=\model \cap \mathcal{A}_{\neg \propertytext}$ associated with the infinite word $v \in \mathcal{L}^\omega(\mathcal{I})$ and a box $b \in B_{\model}$, $\alpha^\ast_b(v, \rho^\omega)$ is the set of \emph{finite} words $\nu_{init}.\nu.\nu_{out} \in \Sigma^*$ associated with the box $b$ and the run $\rho^\omega$ of the infinite word $v$.  
A word $\nu^\ast=\nu_{init}.\nu.\nu_{out}$ is in  $\alpha^\ast_b(v, \rho^\omega)$ if and only if given two indexes $i, j$ such that $0 \leq i < j < \infty $, for all $0 \leq k <j-i$ the following conditions must be satisfied:
\begin{enumerate}
\item  \label{def:portionOfTheWord} $\nu^\ast_k=v_{i+k}$;
\item \label{def:finiteRunAbstractionStateTransparent} $\rho^\omega (i+k) =  \langle b, p, x \rangle$;
\item  \label{def:finiteRunAbstractionTransparentTransition} $(\rho(k)^\omega, v_{i+k}, \rho^\omega(k+1))\in \Delta_i^p$;
\item \label{def:finiteRunAbstractionEnding} $(\rho(j)^\omega, v_{j}, \rho^\omega(j+1))\in \Delta_i^c$ and $\nu_{out}=v_{j} $;
\item  \label{def:finiteRunAbstractionSource} $(i>0 \Leftrightarrow (\rho^\omega(i-1), v_{i-1}, \rho^\omega(i))\in \Delta_i^c$ and $\nu_{init}=v_{i-1})$ or $(\nu_{init}=\epsilon$ and $\rho^\omega(0) \in Q^0_{\mathcal{I}})$.
\end{enumerate}
\end{mydef}
 Condition~\ref{def:portionOfTheWord} specifies that $\nu^\ast$ contains only the portion of the word of interest, i.e., the portion of the word recognized by the box. 
 Condition~\ref{def:finiteRunAbstractionStateTransparent} specifies that the state $\rho^\omega (i+k)$ of the run must corresponds to the tuple $\langle b, p, x \rangle$ where $b$ is the box of interest.
 Condition~\ref{def:finiteRunAbstractionTransparentTransition} specifies that the transition that recognizes the character $v_{i+k}$ must be in $\Delta_i^p$, i.e., it is obtained by firing a transition of the claim when the system is inside the box $b$. 
 Condition~\ref{def:finiteRunAbstractionEnding} forces the word to be of maximal length, i.e., the  transition $(\rho^\omega(j), v_{j}, \rho^\omega(j+1))$ of the run $\rho^\omega$ must be a transition that forces the model to leave the box $b$. The corresponding character $v_{j}$ is added as a suffix $\nu_{out}$ of the word $\nu$.
Similarly, condition~\ref{def:finiteRunAbstractionSource} forces the transition that precedes the set of transitions which recognize $\nu$ to be a transition that enters the box $b$ (excluding the case in which $i=0$, i.e., the initial state of the run $\rho^\omega(0)$ is mixed, in which $\nu_{init}=\epsilon$\footnote{The $\epsilon$ character denotes an empty string.}). 
The character $v_{i-1}$ which labels the transition must be used as a prefix $\nu_{init}$ of $\nu$.

Let us consider the intersection automaton $\mathcal{I}$ presented in Figure~\ref{fig:intersection}, the word $v=$\{$start$\}. \{$send$\}.\{$fai$\}$^\omega$, and the corresponding run $\rho^\omega=\raisebox{.5pt}{\textcircled{\raisebox{-.9pt} {1}}} \raisebox{.5pt}{\textcircled{\raisebox{-.9pt} {2}}} \raisebox{.5pt}{\textcircled{\raisebox{-.9pt} {6}}} \raisebox{.5pt}{\textcircled{\raisebox{-.9pt} {8}}}  (\raisebox{.5pt}{\textcircled{\raisebox{-.9pt} {9}}}  \raisebox{.5pt}{\textcircled{\raisebox{-.9pt} {10}}}  \raisebox{.5pt}{\textcircled{\raisebox{-.9pt} {11}}})^\omega$. 
The finite abstraction of the run associated with the state $send_1$ is the function $\alpha^\ast_{send_1}(v, \rho^\omega)=$\{$start$\}.\{$send$\}.\{$fail$\}.

\begin{mydef}[Finite abstraction of the intersection]
\label{def:finiteIntersectionAbstraction}  
Given the intersection automaton  $\mathcal{I}=\model \cap \mathcal{A}_{\neg \propertytext}$, and a box $b \in B_\model$, the set  $\alpha^\ast_b(\mathcal{I})$ of the finite abstraction of the intersection automaton is defined as $\alpha^\ast_b(\mathcal{I})=$ $\{$ $ \bigcup_{\forall v  \in \mathcal{L}^\omega(\mathcal{I})} \alpha^\ast_b(v, \rho^\omega),$ such that  $\rho^\omega$ is an accepting run of $v$ $\}$.
\end{mydef}  
The finite abstraction of the intersection automaton contains the finite abstraction of the runs associated to every possible word in $\mathcal{L}^\omega(\mathcal{I})$.
  
 \begin{mydef}[Infinite abstractions of a run]
 \label{def:ifiniteRunAbstraction}
Given an infinite run $\rho^\omega$ of $\mathcal{I}=\model \cap \mathcal{A}_{\neg \propertytext}$ associated with the infinite word $v \in \mathcal{L}^\omega(\mathcal{I})$ and a box $b \in B_{\model}$, $\alpha^\omega_b(v, \rho^\omega)$ is the set of \emph{infinite} words $\nu_{init}.\nu^\omega \in \Sigma^\omega$ associated with box $b$ and the run $\rho^\omega$ of the infinite word $v$. 
A word $\nu^\omega=\nu_{init}.\nu$ is in  $\alpha^\omega_b(v, \rho^\omega)$, if and only if given the index $i \geq 0$, for all $k \geq 0 $ the following conditions are be satisfied:
\begin{enumerate}
\item \label{def:ifiniteRunAbstractionPortion} $\nu^\omega_k=v_{i+k}$;
\item \label{def:ifiniteRunAbstractionStateTransparent} $\rho^\omega (i+k) =  \langle b, p, x \rangle$;
\item \label{def:ifiniteRunAbstractionTransparentTransition} $(\rho^\omega(k), v_{i+k}, \rho^\omega(k+1))\in \Delta_i^p$;
\item  \label{def:ifiniteRunAbstractionSource} $(i>0 \Leftrightarrow (\rho^\omega(i-1), v_{i-1}, \rho^\omega(i))\in \Delta_i^c$ and $\nu_{init}=v_{i-1})$ or $(\nu_{init}=\epsilon$ and $\rho^\omega(0) \in Q^0_{\mathcal{I}})$.
\end{enumerate}
\end{mydef}
Condition~\ref{def:ifiniteRunAbstractionPortion} specifies that $\nu^\omega$ contains only the portion of the word of interest, i.e., the portion of the word recognized by the box $b$. Condition~\ref{def:ifiniteRunAbstractionStateTransparent} specifies that the state $\rho^\omega (i+k)$ of the run must corresponds to the tuple $\langle b, p, x \rangle$ where $b$ is the box of interest.
Condition~\ref{def:ifiniteRunAbstractionTransparentTransition} specifies that the transition that recognizes the character $v_{i+k}$ must be in $\Delta_i^p$, i.e., it is obtained by firing a transition of the claim
 when the system is inside the box $b$. Condition~\ref{def:ifiniteRunAbstractionSource} forces the transition that precedes the set of transitions which recognize $\nu$ to be a transition that enters the box $b$ (excluding the case in which $i=0$, i.e., the initial state of the run $\rho^\omega(0)$ is mixed, in which $\nu_{init}=\epsilon$). The character $v_{i-1}$ which labels the transition must be used as a prefix $\nu_{init}$ of $\nu$.
 
\begin{mydef}[Infinite abstraction of the intersection]
\label{def:infiniteIntersectionAbstraction} Given the intersection automaton  $\mathcal{I}=\model \cap \mathcal{A}_{\neg \propertytext}$ and a box $b \in B_\model$,
the set  $\alpha^\omega_b(\mathcal{I})$ of the infinite abstractions of the intersection automaton is defined as $\alpha^\omega_b(\mathcal{I})=$ $\{$ $ \bigcup_{\forall v  \in \mathcal{L}^\omega(\mathcal{I})} \alpha^\omega_b(v, \rho^\omega) $ such that  $\rho^\omega$ is an accepting run of $v$ $\}$.
\end{mydef}  
Informally, the infinite abstraction of the intersection automaton contains the infinite abstractions of the runs associated to every possible word recognized by the automaton $\mathcal{I}$.

\subsection{The model checking procedure}
\label{sec:incompleteModelChecking}
The model checking procedure between an \IBA\ $\model$ and a LTL property $\propertytext$ is based on the intersection between an \IBA\ and a BA (Definition~\ref{def:intersection}) and the completion of an \IBA\ (Definition~\ref{def:IBACompletion}).

\begin{mydef}[Incomplete Model Checking]
\label{th:incompleteModelChecking}  Given an \IBA\ $\model$ and a LTL formula $\propertytext$ associated with the B\"uchi automaton $\mathcal{A}_\propertytext$,
\begin{enumerate}
\item  \label{th:modelchecking1} $\| \model^\propertytext \|=  F       \Leftrightarrow    \model_c \cap \mathcal{A}_{\neg \propertytext} \neq  \emptyset $;
\item \label{th:modelchecking2} $\| \model^\propertytext \|=\ ?        \Leftrightarrow     \| \model^\propertytext \| \neq  F \text{ and }  \model \cap \mathcal{A}_{\neg \propertytext}  \neq \emptyset$;
\item \label{th:modelchecking3} $\| \model^\propertytext \|=  T        \Leftrightarrow    \| \model^\propertytext \| \neq  F \text{ and }  \| \model^\propertytext \| \neq\ ?$.
\end{enumerate}
\end{mydef}

\begin{theorem}[Incomplete Model Checking correctness] 
The incomplete model checking technique is correct.
\end{theorem}

\begin{proof}
Let us consider condition~\ref{th:modelchecking1}. 
($\Leftarrow$) By Lemma~\ref{th:languageMc}, $\model_c$ recognizes the language $\mathcal{L}^\omega(\model)$. 
Since $ \mathcal{A}_{\neg \propertytext}$ describes all the possible infinite words that violate $\propertytext$, if $\model_c \cap \mathcal{A}_{\neg \propertytext}  \neq  \emptyset $ it must exists a word $v \in \mathcal{L}^\omega(\model)$ which violates $\propertytext$. 
Therefore, by Definition~\ref{def:threevalueltl} it follows that $ \| \model^\propertytext \|=  F$.
($\Rightarrow$) As specified in Lemma~\ref{def:threevalueltl} $\| \model^\propertytext \|=  F $ implies that there exist $v \in \mathcal{L}^\omega(\model),  v \not \models \propertytext $. 
Since $v$ does not model $\propertytext$, it is a violating behavior and is included in $\mathcal{L}^\omega(\mathcal{A}_{\neg \propertytext})$. 
Since $v \in \mathcal{L}^\omega(\mathcal{A}_{\neg \propertytext})$ and $v \in \mathcal{L}^\omega(\model)$ and  from Lemma~\ref{th:languageMc} $\mathcal{L}^\omega(\model)$ is  recognized by $\model_c$, $\model_c \cap \mathcal{A}_{\neg \propertytext}$ contains at least the word $v$.

Let us then consider condition~\ref{th:modelchecking2}.
($\Leftarrow$) Given that $ \| \model^\propertytext \| \neq  F$ by  Definition~\ref{def:threevalueltl} it follows that for all $v \in \mathcal{L}^\omega(\model)$,  the word $v$ satisfies $\propertytext$. 
Since $\model \cap \mathcal{A}_{\neg \propertytext} \neq \emptyset$, by Lemma~\ref{lem:intersection} it must exists a word $v$ such that $v \in (\mathcal{L}^\omega(\model) \cup \mathcal{L}^\omega_p(\model)) \cap \mathcal{L}^\omega(\mathcal{A}_{\neg \propertytext})$. 
Since for all $v \in \mathcal{L}^\omega(\model),  v \models \propertytext$, i.e.,  $\mathcal{L}^\omega(\model)\cap \mathcal{L}^\omega(\mathcal{A}_{\neg \propertytext}) = \emptyset$, it must be that $\mathcal{L}^\omega_p(\model) \cap \mathcal{L}^\omega(\mathcal{A}_{\neg \propertytext}) \neq \emptyset$. 
Thus, by Definition~\ref{def:threevalueBAsemantic} condition~\ref{def:threevalueBAsemanticcond2} it must exist $u \in \mathcal{L}_p^\omega(\model)$ such that $u \not \models \propertytext $. 
Thus, by Definition~\ref{def:threevalueltl}, condition~\ref{semantic:possiblysatisfaction} it must be that $\| \model^\propertytext \|=\ ?  $.
($\Rightarrow$)  The proof is by contradiction. Imagine that $ \| \model^\propertytext \| =  \ ? $ and $ \| \model^\propertytext \| =  F \text{ or }  \model \cap \mathcal{A}_{\neg \propertytext}  = \emptyset  $. 
Let us consider the case in which  $\| \model^\propertytext \| =  F$. 
By Definition~\ref{def:threevalueltl} condition~\ref{semantic:notsatisfaction} it must exists a word $v \in \mathcal{L}^\omega (\model)$ that does not satisfy $\propertytext$. 
This implies that condition~\ref{semantic:possiblysatisfaction} of Definition~\ref{def:threevalueltl}  is not satisfied and  $ \| \model^\propertytext \| \neq   ?$ making the hypothesis contradicted. 
Consider then the case in which $\model \cap \mathcal{A}_{\neg \propertytext}  = \emptyset$. 
Since the intersection automaton $\mathcal{I}=\model \cap \mathcal{A}_{\neg \propertytext}$ is empty, from Lemma~\ref{lem:intersection} it must not exists a word $v \in (\mathcal{L}^\omega(\model) \cup \mathcal{L}^\omega_p(\model)) \cap \mathcal{L}^\omega(\mathcal{A}_{\neg \propertytext})$. 
This implies that condition~\ref{semantic:possiblysatisfaction} of Definition~\ref{def:threevalueltl} is not satisfied and $ \| \model^\propertytext \| \neq   ?$ making the hypothesis contradicted.

Condition~\ref{th:modelchecking3} is a consequence of the previously described conditions and Definition~\ref{def:threevalueltl}.
\end{proof}

Theorem~\ref{th:incompleteModelChecking} suggests the model checking procedure presented in Algorithm~\ref{algorithm:modelchecking}.  The algorithm works in five different steps:
\begin{itemize}
\item \noindent \emph{Create the automaton $\mathcal{A}_{\neg \propertytext}$ (Line~\ref{alg:extractLTL}).} As in the classical model checking framework, the first step is to construct the BA that contains the set of behaviors forbidden by the property $\propertytext$. 
This step can be performed in time $\mathcal{O}(2^{(|\neg \propertytext|)})$;
\end{itemize}

\begin{algorithm}[t]
\caption{Checks if an  IBA satisfies a LTL property $\propertytext$}
\label{algorithm:modelchecking}
\begin{algorithmic}[1]
\Procedure{ModelChecking}{$\model$, $\propertytext$}
\State $\mathcal{A}_{\neg \propertytext}$ $\leftarrow$ LTL2BA($\neg \propertytext$); \label{alg:extractLTL} 
\State $\model_c$ $\leftarrow$ ExtractMc($\model$); \label{alg:extractMC}
\State $\mathcal{I}_c$ $\leftarrow$ $\model_c \cap \mathcal{A}_{\neg \propertytext}$; \label{alg:computeMCIntersection} 
\State empty $\leftarrow$ CheckEmptiness($\mathcal{I}_c$);\label{alg:checkMCEmptiness} 
\If{!empty}
 \State   \Return F;
 \Else \label{alg:endcheckMCEmptiness}
\State $\mathcal{I} \leftarrow \model \cap \mathcal{A}_{\neg \propertytext}$; \label{alg:computeIntersection} 
 \State empty $\leftarrow$ CheckEmptiness($\mathcal{I}$); \label{alg:checkEmptiness} 
\If{empty}
  \State \Return $T$;
\Else
 \State \Return $?$;
\EndIf
\label{alg:endcheckEmptiness}\EndIf
\EndProcedure
\end{algorithmic}
\end{algorithm}

For example, the automaton $\mathcal{A}_{\neg \propertytext}$, corresponding to the LTL property $\propertytext=G(send \rightarrow F (success))$, is represented in Figure~\ref{fig:claimexample}.
\begin{itemize}
\item \noindent \emph{Extract the automaton $\model_c$ and build the intersection automaton $\mathcal{I}_c = \model_c \cap \mathcal{A}_{\neg \propertytext}$ (Lines~\ref{alg:extractMC}-\ref{alg:computeMCIntersection}).} The automaton $\model_c$  contains all the accepting behaviors of the system. 
In this sense, $\model_c$  is a lower bound on the set of behaviors of the system, i.e., it contains all the behaviors the system is going to exhibit at run-time. 
Thus, the intersection automaton $\mathcal{I}_c$ contains the behaviors of $\model_c$ that violate the property. 
The computation of  $\model_c$ can be performed in time $\mathcal{O}(|Q_\model| + |\Delta_\model|)$, since it is sufficient to remove from the automaton the black box states and their incoming and outgoing transitions.
The intersection automaton $\mathcal{I}_c$ contains in the worst case $3 \cdot |R_\model| \cdot  |Q_{\mathcal{A}_{\neg \propertytext}}|  $ states.
\end{itemize}
For example, the automaton $\model_c$ associated with the the model $\model$ described in Figure~\ref{fig:modelexample}  contains only states $q_1$, $q_2$ and $q_3$ and the transitions marked with the labels $success$ and $abort$. 
The intersection between this automaton and the property $\mathcal{A}_{\neg \propertytext}$ described in Figure~\ref{fig:claimexample} contains all the behaviors of the sending message system that violate the property. 
Since the automaton is empty, there are no behaviors of $\model$ that violate $\propertytext$.
\begin{itemize}
\item \noindent \emph{Check the emptiness of the intersection automaton $\mathcal{I}_c$ (Lines~\ref{alg:checkMCEmptiness}-\ref{alg:endcheckMCEmptiness}).} 
If $\mathcal{I}_c$ is not empty, the condition $\mathcal{L}(\model)  \cap \mathcal{L}(\mathcal{A}_{\neg \propertytext}) \not =\emptyset$ is matched, i.e., the property is not satisfied and every infinite word in the intersection automaton is a counterexample. 
If, instead, $\mathcal{I}_c$ is empty, $\model$ possibly satisfies or definitely satisfies $\propertytext$ depending on the result of the next steps of the algorithm.
\end{itemize}
The intersection automaton $\mathcal{I}_c$ of the sending message example is empty since the automaton $\mathcal{I}_c$  does not contain any accepting state that can be entered infinitely often and is reachable from one of its initial states. 
Indeed, both $q_2$ and $q_3$  (the accepting states of  $\model_c$) are never reachable from $q_1$ in the intersection automaton. 
Thus, $\model$ definitely satisfies or possibly satisfies the property $\propertytext$ depending on the next steps of the algorithm.
\begin{itemize}
\item \noindent \emph{Compute the intersection  $\mathcal{I}= \model \cap \mathcal{A}_{\neg \propertytext}$ of the incomplete model $\model$ and the automaton $ \mathcal{A}_{\neg \propertytext}$ associated with the property $\propertytext$ (Line~\ref{alg:computeIntersection}).} 
To  check whether $\model$ definitely satisfies or possibly satisfies $\propertytext$ it is necessary to verify if $(\mathcal{L}(\model) \cup \mathcal{L}_p(\model)) \cap \mathcal{L}(\mathcal{A}_{\neg \propertytext})=\emptyset$.
Indeed$\model$ intrinsically specifies an upper bound on the behaviors of the system, i.e., it contains all the behaviors the system must and may exhibit.  
The intersection algorithm presented in Section~\ref{sec:BAIBAIntersection} is used to compute the intersection automaton $\mathcal{I} = \model \cap \mathcal{A}_{\neg \propertytext}$.
\end{itemize}
The intersection automaton $\mathcal{I}$ of the sending message example is depicted in Figure~\ref{fig:intersection}.
\begin{itemize}
\item \noindent \emph{Check the emptiness of the intersection automaton  $\mathcal{I}=\model \cap \mathcal{A}_{\neg \propertytext}$ (Lines \ref{alg:checkEmptiness}-\ref{alg:endcheckEmptiness}).} 
By checking the emptiness of the automaton $\mathcal{I}$ we verify whether the property $\propertytext$ is definitely satisfied or possibly satisfied by the model $\model$. 
Since we have already checked that $\mathcal{L}(\model) \subseteq \mathcal{L}(\mathcal{A}_{\neg \propertytext})$, two cases are possible: if $\mathcal{I}$ is empty,  $\mathcal{L}_p(\model) \subseteq \mathcal{L}(\mathcal{A}_{\neg \propertytext})$ and the property is definitely satisfied whatever refinement is proposed for the boxes of $\model$, otherwise, $\mathcal{L}_p(\model) \not \subseteq \mathcal{L}(\mathcal{A}_{\neg \propertytext})$, meaning that there exists some refinement of $\model$ that violates the property. 
\end{itemize}
For example, the word \{$start$\}.\{$send$\}.\{$fail$\}.\{$fail$\}.\{$abort$\}$^\omega$, which is a possibly accepted by $\model$, violates $\propertytext$ since it corresponds to a run where a $send$ is not followed by a $success$. 
This behavior can be generated by replacing to box $send_1$ a component containing runs where a message is $sent$ and no $success$ is obtained, and to box $send_2$ an empty component that neither tries to $send$ a message nor waits for a $success$.

\begin{theorem}[Incomplete Model Checking complexity]
\label{th:incompleteModelCheckingComplexity}  
Checking an \IBA\ $\model$ against the BA $\mathcal{A}_{\neg \propertytext}$ associated to the LTL formula $\neg \propertytext$ can be performed in time $\mathcal{O}(|\model|\cdot|\mathcal{A}_{\neg \propertytext}|)$, where $|\model|$ and $|\mathcal{A}_{\neg \propertytext}|$ are the sizes of the model and the automaton associated with the negation of the claim, respectively. 
\end{theorem}

\begin{proof}
Checking an IBA $\model$ against a property expressed as a BA $\mathcal{A}_{\neg \propertytext}$ requires to check the emptiness of two intersection automata $\mathcal{I}_c$ and $\mathcal{I}$ representing a lower bound and an upper bound on the behaviors of the system, respectively. 
These automata contain in the words case $3 \cdot |Q_\model| \cdot  |Q_{\mathcal{A}_{\neg \propertytext}}| $ states. 
The emptiness checking procedure can be performed in time $\mathcal{O}(|Q_{\mathcal{I}_c}|+|\Delta_{\mathcal{I}_c}|)$ and $\mathcal{O}(|Q_{\mathcal{I}}|+|\Delta_{\mathcal{I}}|)$, for the two intersection automata, respectively.
\end{proof}

%\begin{theorem}
%\label{th:incompleteModelCheckingComplexityLTL} (LTL Model Checking Complexity) Checking an IBA $\model$ against the LTL formula $\neg \propertytext$ has an exponential temporal complexity since the automaton corresponding to the formula $\propertytext$ $\mathcal{O}(|\model|\cdot 2^{\mathcal{O}(|\propertytext|)})$, where $|Q_{\mathcal{I}}|$ and $|\Delta_{\mathcal{I}}|$ is the cardinality of the set of the states and transitions of the intersection automaton, respectively. 
%\end{theorem}

\section{Constraint computation}
\label{sec:ComputingConstraints}

When a property $\phi$ is possibly satisfied, each word $v$, recognized by the intersection automaton $\mathcal{I}=\mathcal{M} \cap \mathcal{A}_{\neg \phi}$, corresponds to a behavior $\mathcal{B}$ the system \emph{may} exhibit that violates $\phi$. 
To make $\phi$  satisfied, the developer must design the replacements $\mathcal{R}_{b_1}, \mathcal{R}_{b_2} \ldots  \mathcal{R}_{b_n}$ of the black box states $b_1, b_2, \ldots b_n$ to forbid $\mathcal{B}$ from occurring. A replacement of a box is an automaton to be substituted to a box $b$. 
The goal of the constraint computation is to find  the set of sub-properties, one for each box, to be satisfied by the replacements of the boxes of the model $\mathcal{M}$ such that the refinement $\mathcal{N}$ does not violate $\phi$. 
Sub-properties are guidelines that help the developer in the replacement design and can be considered as a contract in a contract based design setting~\cite{meyer1992applying}. 
The constraint computation procedure is based on three subsequent steps:
\begin{inparaenum}[\itshape a\upshape)]
\item intersection cleaning;
\item sub-properties generation;
\item constraint identification.
\end{inparaenum}

\subsection{Intersection cleaning}
\label{sec:IntersectionCleaning}
The intersection cleaning phase removes from the intersection automaton $\mathcal{I}=\mathcal{M} \cap \mathcal{A}_{\neg \phi}$ the states that are not involved in any behavior $\mathcal{B}$ that possibly violates the property.
Indeed, these behaviors must not be included in any sub-property. 
Given the intersection automaton $\mathcal{I}=\mathcal{M} \cap \mathcal{A}_{\neg \phi}$, the cleaned intersection automaton $\mathcal{I}_{cl}$ is a version of the intersection automaton where the states from which it is not possible to reach an accepting state that can be entered infinitely many often are removed. 

\begin{mydef}[Intersection cleaning]
\label{th:cleanedIntersection}  
Given the intersection automaton $\mathcal{I}=\langle \Sigma_{\mathcal{I}}, Q_{\mathcal{I}},$ $ \Delta_{\mathcal{I}},  Q_{\mathcal{I}}^0,  F_{\mathcal{I}} \rangle$, the cleaned intersection automaton $\mathcal{I}_{cl}=\langle \Sigma_{\mathcal{I}_{cl}}, Q_{\mathcal{I}_{cl}},$ $ \Delta_{\mathcal{I}_{cl}},  Q_{\mathcal{I}_{cl}}^0,  F_{\mathcal{I}_{cl}} \rangle$ satisfies the following conditions:
\begin{itemize}
\item $\Sigma_{\mathcal{I}_{cl}}=\Sigma_{\mathcal{I}}$;
\item $Q_{\mathcal{I}_{cl}} = \left \{ q \in Q_{\mathcal{I}}  \right.$ such that there exits a possibly accepting run $\rho^\omega$ and an index $i \geq 0$ and $\left. \rho^\omega(i)=q  \right \}$;
\item $\Delta_{\mathcal{I}_{cl}}  = \left \{(q, a, q^\prime) \in \Delta_{\mathcal{I}} \text{ such that } q \in Q_{\mathcal{I}_{cl}} \text{ and } q^\prime \in Q_{\mathcal{I}_{cl}} \right \}$;
\item $Q_{\mathcal{I}_{cl}}^0=Q_{\mathcal{I}}^0 \cap Q_{\mathcal{I}_{cl}}$;
\item $F_{\mathcal{I}_{cl}}= F_{\mathcal{I}} \cap Q_{\mathcal{I}_{cl}}$.
\end{itemize}
\end{mydef}

The cleaned version of the intersection automaton can be obtained using Algorithm~\ref{algorithm:cleaning}. 
The algorithm first creates a copy of the intersection automaton $\mathcal{I}$ which will contain $\mathcal{I}_{cl}$ (Line~\ref{alg:cleaningcopy}). 
Then, the non trivial strongly connected components are computed and stored in the  $SCC$ set (Line~\ref{alg:cleaningNonTrivialSCC}).  
Then, the set $Nxt$ is defined (Line~\ref{algorithm:cleaningNextDefinition}). 
The set next is used to store the set of states from which it is possible to reach an accepting state (which can be entered infinitely often) whose predecessors still have to be analyzed.

\begin{algorithm}[t]
\caption{Removes the states that are not involved in a possibly accepting run.}
\label{algorithm:cleaning}
\begin{algorithmic}[1]
\Procedure{IntersectionCleaner}{$\mathcal{I}$}
\State $\mathcal{I}_{cl} \leftarrow$\Call{Clone}{$\mathcal{I}$}; \label{alg:cleaningcopy}
\State $SCC \leftarrow$\Call{getNonTrivialSCC}{$\mathcal{I}_{cl}$};  \label{alg:cleaningNonTrivialSCC}
\State $Nxt \leftarrow \left \{ \right \}$; \label{algorithm:cleaningNextDefinition}
\For{$scc \in SCC$} \label{algorithm:cleaningForEachscc}
\If{$scc \cap F_{\mathcal{I}} \neq \emptyset$}  \label{algorithm:cleaningAtLeastATransparent}
\State  $Nxt \leftarrow Nxt \cup scc$; \label{algorithm:cleaningAddingTheStates}
\EndIf
\EndFor \label{algorithm:cleaningForEachsccEnd}
\State $Vis \leftarrow \left \{ \right \}$;  \label{algorithm:cleaningVis}
\While{$Nxt \neq \emptyset$} \label{algorithm:cleaningEmpty}
\State $s \leftarrow$\Call{choose}{$Nxt$}; \label{algorithm:stateSelection}
\State $Vis \leftarrow Vis \cup \left \{ s\right \}$; \label{algorithm:cleaningAddingVisited}
\State $Nxt \leftarrow Nxt\setminus \left \{ s \right \}$; \label{algorithm:IntersectionCleanerRemovingNext}
\State $Nxt \leftarrow Nxt\cup \left \{ s^\prime \mid (s^\prime, a, s) \in \Delta_{\mathcal{I}}  \wedge s \not \in Vis \right \}$; \label{algorithm:cleaningAddingNext}
\EndWhile \label{algorithm:cleaningEmptyEnd}
\For{$s \in (Q_{\mathcal{I}} \setminus Vis) $} \label{algorithm:cleaningForVisited}
\State \Call{removeState}{$\mathcal{I}_{cl}$, $s$}; \label{algorithm:cleaningRemoveState}
\EndFor \label{algorithm:cleaningForVisitedEnd}
\EndProcedure
\end{algorithmic}
\end{algorithm}

For each set of states that form a strongly connected component $scc$ which is in the $SCC$ set (Line~\ref{algorithm:cleaningForEachscc}), if at least an accepting state is present (Line~\ref{algorithm:cleaningAtLeastATransparent}), the states are added to the set $Nxt$ of the states to be visited next (Line~\ref{algorithm:cleaningAddingTheStates}). 
Then, the set $Vis$ is defined (Line~\ref{algorithm:cleaningVis}). 
The set $Vis$ is used to store the set of the already visited states. 
The goal of this set is to guarantee that a state is not visited twice during the state space exploration. 
By exploring the state space of $\mathcal{I}$ the states from which it is possible to reach a state in the set  $Nxt$ are identified. 

The state space exploration is an iterative process that ends when the set $Nxt$ is empty (Line~\ref{algorithm:cleaningEmpty}). 
At each exploration step, a state $s \in Nxt$ is selected (Line~\ref{algorithm:stateSelection}), added to the set of visited states (Line~\ref{algorithm:cleaningAddingVisited}), and removed to the set $Nxt$ (Line~\ref{algorithm:IntersectionCleanerRemovingNext}). 
Then, all the states $s^\prime$, which have not  already been visited and are predecessors of the state $s$, are added to the set  $Nxt$ of the states to be analyzed next (Line~\ref{algorithm:cleaningAddingNext}). 
Finally, each state that has not been visited (Line~\ref{algorithm:cleaningForVisited})  and its incoming and outgoing transitions are removed from the automaton (Line~\ref{algorithm:cleaningRemoveState}). 

The intersection cleaning algorithm applied to the intersection automaton described in Figure~\ref{fig:intersection} removes the states $\raisebox{.5pt}{\textcircled{\raisebox{-.9pt} {4}}}$, $\raisebox{.5pt}{\textcircled{\raisebox{-.9pt} {5}}}$ and $\raisebox{.5pt}{\textcircled{\raisebox{-.9pt} {7}}}$ since they are not involved in any possibly accepting run.

\begin{theorem}[Correctness]
\label{th:cleaningCorrect}  The cleaning procedure is correct. 
\end{theorem}

\begin{proof}
To prove the correctness of the procedure, it is necessary to prove that the automaton $\mathcal{I}_{cl}$ obtained using Algorithm~\ref{algorithm:cleaning} satisfies the conditions specified in Definition~\ref{th:cleanedIntersection}. 
Note that, the intersection cleaning algorithm is executed on the intersection automaton only when the property is possibly satisfied, i.e., it does not exists any definitely accepting run, but the intersection only contains  possibly accepting ones. 
The proof is by contradiction. 
Assume that there exists an automaton $\mathcal{I}_{cl}$ obtained using Algorithm~\ref{algorithm:cleaning} which does not satisfy Definition~\ref{th:cleanedIntersection}. 
Then, it must exist a state $q \in Q_{\mathcal{I}_{cl}}$ which is not involved in any possibly accepting run $\rho^\omega$. 
Imagine that such a state exists. 
To not be removed at the end of the cleaning algorithm it must be contained into the set $Vis$ of the visited states. 
To be inserted in $Vis$ it is necessary that $q$ was inserted in the set $Nxt$ before. Two cases are possible:
\begin{inparaenum}[\itshape a\upshape)]
\item the state was included in the set $Nxt$ before the \texttt{while} cycle. In this case, the state was a part of a strongly connected component that contains at least an accepting state.
Thus, the hypothesis is contradicted, since from that state there was at least a possible way to reach an accepting state that is entered infinitely often;
\item the state is included in the set $Vis$ inside the \texttt{while} cycle. 
Then, the state is a predecessor of a state $q^\prime$ from which it is possible to reach an accepting state that can be entered infinitely often. 
Thus, the hypothesis is again contradicted.
\end{inparaenum}
\end{proof}

\begin{theorem}[Complexity]
\label{th:ComplexityCleaningProcedure}
The intersection cleaning procedure can be performed in time $\mathcal{O}(|Q_{\mathcal{I}}|+|\Delta_{\mathcal{I}}|)$.
\end{theorem}

\begin{proof}
Cloning the intersection automaton $\mathcal{I}$ (Line~\ref{alg:cleaningcopy}) can be be done in time $\mathcal{O}(|Q_{\mathcal{I}}|+|\Delta_{\mathcal{I}}|)$. 
Indeed, it is sufficient to traverse $\mathcal{I}$ and clone its states and transitions. 
The same time complexity is required for finding the no trivial strongly connected components (Line~\ref{alg:cleaningNonTrivialSCC}). 
For example, it is possible to use the well known Tarjan's Algorithm~\cite{tarjan1972depth}. 
Checking whether a strongly connected component contains an accepting state (Lines~\ref{algorithm:cleaningNextDefinition}-\ref{algorithm:cleaningForEachsccEnd}) can be done in time $\mathcal{O}(|Q_{\mathcal{I}}|)$. 
Finally, the state space exploration (Lines~\ref{algorithm:cleaningEmpty}-\ref{algorithm:cleaningEmptyEnd}) has $\mathcal{O}(|Q_{\mathcal{I}}|+|\Delta_{\mathcal{I}}|)$ time complexity, since each state and transition is visited at most once. 
The same time complexity is required to remove the not reachable states (Lines~\ref{algorithm:cleaningForVisited}-\ref{algorithm:cleaningForVisitedEnd}). 
Thus, the final complexity of the algorithm is $\mathcal{O}(|Q_{\mathcal{I}}|+|\Delta_{\mathcal{I}}|)$.
\end{proof}

\subsection{Generation of constraints}
\label{Sec:ModelingSubProperty}
The constraint computation algorithm identifies for each box the corresponding constraint. 
Each constraint $C=\langle S, S_p \rangle$ encodes the set of behaviors the replacement of the box $b$ cannot or should not exhibit through two sub-properties $S$ and $S_p$. The behaviors encoded in $S$ and $S_p$ if present in the replacement would lead to a  \emph{violation} or a \emph{possible violation} of the property of interest.  
For example, assume that the box $send_1$  presented in Figure~\ref{fig:modelexample} is accepting. If it is replaced by the automaton presented in Figure~\ref{fig:exampleReplacementStateSend1}, the claim $\LTLglobally (send \rightarrow \LTLfinally (success))$ is violated, since the sending activity is not followed by a success. Instead, if  $send_1$ is replaced by the automaton presented in Figure~\ref{fig:exampleReplacementStateSend2}, it possibly violates the claim since it needs the ``cooperation" of the replacement of the box $send_2$ to generate a violating behavior. Indeed, checking the refinement model after the first replacement is proposed yields to a claim violation, while in the second case the claim is possibly satisfied.
\begin{figure*}[ht]
    \centering
    \subfloat[A replacement for the box $send_1$.]
     { \includegraphics[scale=0.44]{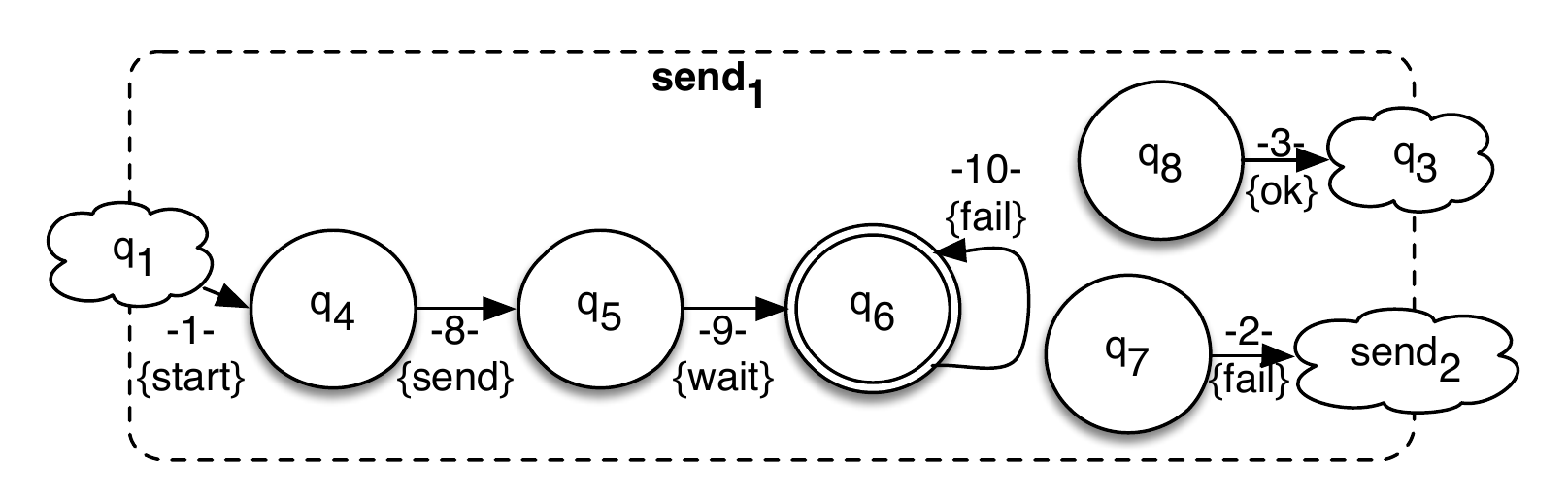}
    \label{fig:exampleReplacementStateSend1}}
    \subfloat[A replacement for the box $send_1$.]{
          \includegraphics[scale=0.44]{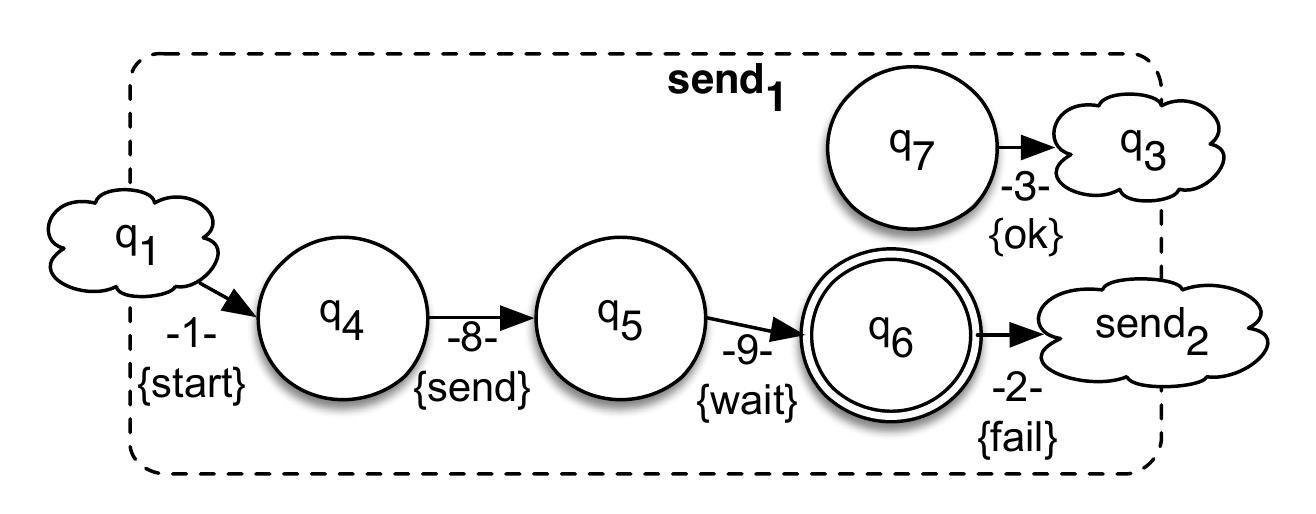}
              \label{fig:exampleReplacementStateSend2}}
    \caption{Two examples of replacements for the box $send_1$.}
\end{figure*}

To encode the behaviors of the replacement of the box $b$ that would lead to a violation of the claim $\phi$, it is necessary to describe the internal violating behaviors and how these behaviors garnish the model of the system. 
For this reason, each of the two sub-properties associated with the box $b$ provides 
\begin{inparaenum}[\itshape a\upshape)]
\item a BA $\mathcal{P}$, which describes the internal violating or possibly violating behaviors of the replacement, 
\item the incoming $\Delta^{in\mathcal{P}}$ and outgoing $\Delta^{out\mathcal{P}}$ transitions,  which describe how the replacement must (must not) be connected to the original model,
\item a subset $G$ of the source states of $\Delta^{in\mathcal{P}}$. These states are reachable by the system without making assumption on the behavior of the box $b$.
\item  a subset $R$ of the source states of $\Delta^{in\mathcal{P}}$. These states are states from which an accepting state that can be entered infinitely often can be reached, i.e., a violating behavior is exhibited by the system.
\item  a relation $K$ between the target states of $\Delta^{out\mathcal{P}}$ and the source states of $\Delta^{in\mathcal{P}}$. This relation specifies whether allowing the system to reach a target state of a transition in  $\Delta^{out\mathcal{P}}$ allows reaching a source state of $\Delta^{in\mathcal{P}}$.
\end{inparaenum} Formally:

\begin{mydef}[Sub-property]
\label{def:subproperty} 
Given a model $\mathcal{M}$ defined over the set of states $Q_{\mathcal{M}}$, a sub-property $\mathcal{S}$ associated with the box $b \in B_{\mathcal{M}}$ is a tuple \subproperty, where:
\begin{itemize}
\item $\mathcal{P}=\langle \Sigma_{\mathcal{P}}, Q_{\mathcal{P}}, \Delta_{\mathcal{P}},$ $ Q_{\mathcal{P}}^0, F_{\mathcal{P}} \rangle$ is a BA. 
$\mathcal{P}$ must satisfy the following conditions: if $b \not \in Q^0_{\mathcal{M}}$ then $Q^0_{\mathcal{P}}=\emptyset$,  if $b \not \in F_{\mathcal{M}}$ then $F_{\mathcal{P}}=\emptyset$; 
\item the sets $\Delta^{in\mathcal{S}} \subseteq \left \{ (q^\prime, a, q) \mid (q^\prime, a, b) \in  \Delta_{\mathcal{M}} \text{ and } q \in  Q_{\mathcal{P}} \text{ and } q^\prime \in \right.$ $\left. Q_{\mathcal{M}} \right \}$ and 
$ \Delta^{out\mathcal{S}} \subseteq \left \{ (q, a, \right.$ $\left. q^\prime) \mid  (b, a, q^\prime) \in  \Delta_{\mathcal{M}}  \text{ and } q \in Q_{\mathcal{P}}\text{ and }  q^\prime \right.$ $\left. \in  Q_{\mathcal{M}}  \right \}$ are the incoming and outgoing transitions  of the sub-property $\mathcal{S}$;
\item the set $G \subseteq \Delta^{in\mathcal{S}}$ of the incoming transitions of the sub-property reachable without making any assumption of the behavior of the system in the box $b$;
\item the set $R \subseteq \Delta^{out\mathcal{S}} $ of the outgoing transition from which an accepting state that can be entered infinitely often is reachable without crossing state $b$;
\item $K \subseteq \Delta^{in\mathcal{S}} \times  \Delta^{out\mathcal{S}}$ specifies if from an outgoing transition of the sub-property it is possible to reach one of its incoming transitions.
\item $\Gamma_{\mathcal{M}}, \Gamma_{\mathcal{A}_{\neg \phi}}: K \rightarrow \{ T, F\}$: the function $\Gamma_{\mathcal{M}}$ ($\Gamma_{\mathcal{A}_{\neg \phi}}$) specifies if to reach an incoming transition from an outgoing one, an accepting state of the model (claim) is traversed.
\end{itemize}
\end{mydef}

$\mathcal{P}$ is the BA that encodes the condition the developer must satisfy in the design of the replacement $\mathcal{R}$ to be substituted to the black box state $b$.
If the box $b$ is not initial/accepting, the automaton $\mathcal{P}$ can not contain initial/accepting states. 
The incoming $\Delta^{in\mathcal{S}}$ and outgoing $\Delta^{out\mathcal{S}}$ transitions specify how the behaviors encoded in the automaton $\mathcal{P}$ are related to the original model $\mathcal{M}$. 
The sets $G$ and $R$ contain a subset of the incoming and outgoing transitions, respectively.
An incoming transition $\delta$ is in $G$ if it is not necessary to traverse the replacement of the box $b$ for $\delta$ to be reached.
Note that depending on whether sub-property $S_p$ or $S$ is considered, the replacement of other black box states can be crossed for $\delta$ to be reached.
An outgoing transition $\delta$ is in $R$ if it is not necessary to traverse the replacement of the box $b$ for reaching an accepting state that can be entered infinitely many often.
Note that depending on whether sub-property $S_p$ or $S$ is considered, the replacement of other black box states can be crossed or not for the accepting state to be reached.
The reachability relations $K$  describe whether allowing the system to reach an outgoing transition of $S$ permits the system also to reach one of its incoming transition (it can be that the transition was not directly reachable before, i.e., it was necessary to cross $b$).
 A tuple $(\delta^o, \delta^i)$ is in $K$ if and only if from the outgoing transition $\delta^o$ it is possible to reach, in the intersection automaton, the incoming transition $\delta^i$  without traversing a state of the intersection automaton generated from $b$.
 Depending on whether  sub-property $S_p$ or $S$ is considered, the replacement of other boxes can be crossed or not in the path that connects the outgoing transition to the incoming one.
 $\Gamma_{\mathcal{M}}$ and $\Gamma_{\mathcal{A}_{\neg \phi}}$  associate to each reachability entry $(\delta_o, \delta_i)$ in $K$   a $T$ or a $F$ value. 
$G$, $R$,  $K$, $\Gamma_{\mathcal{M}}$ and $\Gamma_{\mathcal{A}_{\neg \phi}}$are described more in detail in the following of this section.

 \begin{figure*}[t]
    \centering
    \subfloat[The sub-property $\mathcal{S}_p$ that correspond to the box $send_1$.]
     { \includegraphics[scale=0.5]{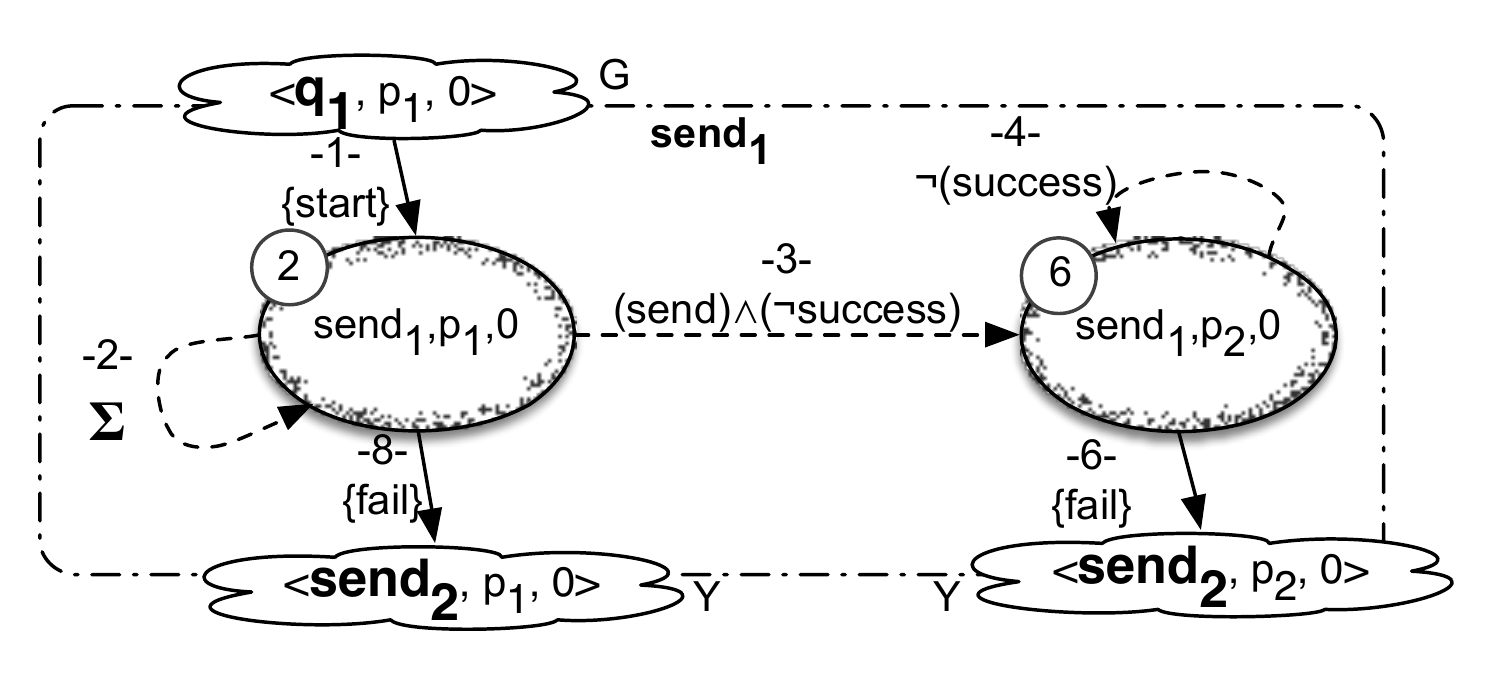}
   \label{fig:subpropertysend1}}\\
    \subfloat[The sub-property $\mathcal{S}_p$ that correspond to the box $send_2$.]{
          \includegraphics[scale=0.5]{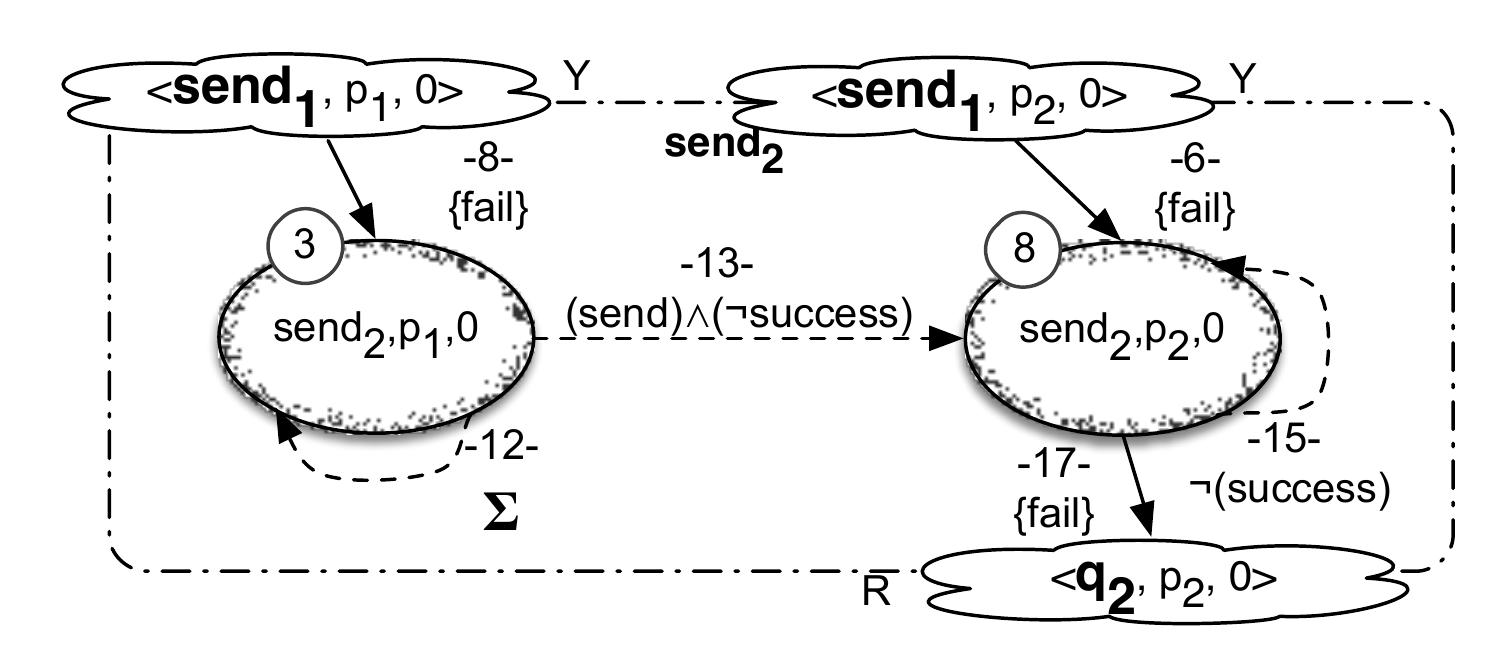}
              \label{fig:subpropertysend2}}
    \caption{The sub-properties associated with the boxes $send_1$ and $send_2$.}
\end{figure*}

The sub-properties $\mathcal{S}_p$ and $\mathcal{S}_p$ associated with the boxes $send_1$ and $send_2$ of model $\mathcal{M}$ described in Figure~\ref{fig:modelexample} and the claim $\phi$, associated with the automaton $\mathcal{A}_\phi$ described in Figure~\ref{fig:claimexample}, are presented in Figures~\ref{fig:subpropertysend1} and~\ref{fig:subpropertysend2}. 
The sub-properties are surrounded by a dotted border, which contains the name of the box the sub-property refers to. 
 The automaton $\mathcal{P}$, which corresponds to the sub-property $\mathcal{S}_p$ associated with the box $send_1$ contains the two states $\raisebox{.5pt}{\textcircled{\raisebox{-.9pt} {2}}}$ and $\raisebox{.5pt}{\textcircled{\raisebox{-.9pt} {6}}}$ and the internal transitions $2, 3$ and $4$. 
 The sub-property $\mathcal{S}_p$  associated with the box $send_1$ has the incoming 
transition $1$ and the outgoing transitions $6$ and $8$.
The automaton $\mathcal{P}$, which corresponds to the sub-property $\mathcal{S}_p$ associated with the box $send_2$, contains the two states $\raisebox{.5pt}{\textcircled{\raisebox{-.9pt} {3}}}$ and $\raisebox{.5pt}{\textcircled{\raisebox{-.9pt} {8}}}$ and the internal transitions $12, 13$ and $15$. 
  The sub-property $\mathcal{S}_p$ associated with the property $send_2$ has incoming transitions $6$, $8$ and the outgoing transition $17$. 
  The incoming transition $1$ of the sub-property $\mathcal{S}_p$ of associated with box $send_1$ is in the set $G$ since this transition is reachable from an initial state of the system through a path which does not involve the replacement associated with the box $b$ (graphically, the source/destination of an incoming/outgoing transition is depicted with a cloud which is decorated with a $G$/$R$ label to indicate that the transition is in the set $G$/$R$). 
The outgoing transitions of the sub-property $\mathcal{S}_p$ associated with the box $send_1$ are in the set $R$ since from these transitions it is possible to reach an accepting state through a run which involves the refinement of the box $send_2$. 
Note that since we are considering sub-property $\mathcal{S}_p$ other boxes can be crossed by the run, i.e., of the box $send_2$.  
Differently, when sub-property $\mathcal{S}$ only state of the intersection obtained from regular states of the model can be crossed.
The outgoing transition $17$ of the sub-property $\mathcal{S}_p$ of $send_2$  is in the set $R$ since by executing this transition it is possible to reach an accepting state that can be entered infinitely often. 
The incoming transitions $6$ and $8$ of the sub-property $\mathcal{S}_p$ of the box $send_2$ are in the set $G$ since these transitions are reachable in the intersection automaton without crossing box $send_2$.
Note that since we are considering sub-property $\mathcal{S}_p$ other boxes can be crossed by the run, i.e., of the box $send_1$. 
Differently, when sub-property $\mathcal{S}$ only state of the intersection obtained from regular states of the model can be crossed.
Finally, the reachability relation $K$ of $\mathcal{S}_p$ of $send_1$ ($send_2$) is empty since it is not possible to reach from its outgoing transitions its incoming transitions.

As a replacement, the sub-property $\mathcal{S}_p$  is associated with the four types of accepting runs identified in Section~\ref{Sec:ModelingReplacements}: finite internal, infinite internal, finite external and infinite external accepting. 

Given the cleaned intersection automaton $\mathcal{I}_{cl}$ (which contains the behaviors that possibly violate the claim), obtained from the intersection automaton $\mathcal{I}=\mathcal{M} \cap \mathcal{A}_{\neg \phi}$ between the model $\mathcal{M}$ and the automaton $\mathcal{A}_{\neg \phi}$, the sub-property identification problem concerns the identification of the sub-properties the developer must satisfy in the refinement activity. 
The sub-property identification procedure works through a set of subsequent steps: automata extraction, computation of the sets $G$ and $R$ and of the reachability relation $K$.

First, the \emph{automata extraction} procedure computes the automata associated with the sub-properties and their incoming and outgoing transitions.

\begin{mydef}[Automata extraction]
\label{def:subidentification}  
Given a property $\phi$ which is possibly satisfied by the model $\mathcal{M}$,  the cleaned intersection automaton $\mathcal{I}_{cl}$ obtained from the intersection automaton $\mathcal{I}=\mathcal{M} \cap \mathcal{A}_{\neg \phi}$, and the constraint $C=\langle S, S_p \rangle$ associated with the box $b$, the automaton $\mathcal{P}$ associated with the sub-properties $S$ and $S_p$ is a BA $\langle \Sigma_{\mathcal{P}}, Q_{\mathcal{P}}, \Delta_{\mathcal{P}},$ $ Q_{\mathcal{P}}^0, F_{\mathcal{P}} \rangle$, such that:
\begin{itemize}
\item $\Sigma_{\mathcal{P}}=\Sigma_{\mathcal{I}_{cl}}$;
\item $Q_{\mathcal{P}}= \{ \langle q_{\mathcal{M}}, p_{\mathcal{A}_{\neg \phi}}, x \rangle \in Q_{\mathcal{I}_{cl}} $ such that $ q_{\mathcal{M}} =b  \}$;
\item $\Delta_{\mathcal{P}}=  \{ (q, a, q^\prime) $ such that $ q, q^\prime \in Q_{\mathcal{P}} \text{ and } (q, a, q^\prime) \in \Delta^p_{\mathcal{I}_{cl}}  \}$;
\item $Q^0_{\mathcal{P}}= Q^0_{\mathcal{I}_{cl}} \cap Q_{\mathcal{P}}$;
\item $F_{\mathcal{P}}=  F_{\mathcal{I}_{cl}} \cap Q_{\mathcal{P}}$.
\end{itemize}
The incoming $\Delta^{in\mathcal{S}}$ and outgoing $\Delta^{out\mathcal{S}}$ transitions associated with the sub-properties $\mathcal{S}$ and  $\mathcal{S}_p$ are defined as:
\begin{itemize}
\item $\Delta^{in\mathcal{S}}= \{ (q_{\mathcal{M}}, a, \langle b, p_{\mathcal{A}_{\neg \phi}}^\prime, y \rangle) $ such that $  
(
\langle q_{\mathcal{M}}, p_{\mathcal{A}_{\neg \phi}}, x \rangle, a, 
\langle b, p_{\mathcal{A}_{\neg \phi}}^\prime, y \rangle)$ $\in$ $(\Delta_{\mathcal{I}}^c$ $\cap \Delta_{\mathcal{I}_{cl}})
 \}$;
\item $\Delta^{out\mathcal{S}} = \{ (\langle b, p_{ \mathcal{A}_{\neg\phi}}, x \rangle, a, q_{\mathcal{M}}^\prime)  $ such that $ 
(\langle b_i, p_{\mathcal{A}_{\neg \phi}}, x \rangle, 
a, 
\langle q_{\mathcal{M}}^\prime, p_{\mathcal{A}_{\neg \phi}}^\prime, y \rangle) $  $\in$ $(\Delta_{\mathcal{I}}^c$ $\cap \Delta_{\mathcal{I}_{cl}}) \}$.
\end{itemize} 
\end{mydef}

Algorithm~\ref{algorithm:subpropertyIdentification} is used to compute the the incoming and outgoing transitions of the  sub-properties associated to the cleaned intersection automaton $\mathcal{I}_{cl}$. 
Note that since multiple boxes may be present in the model, we indicate as $\mathcal{P}_b$ the automaton associated with the box $b$ and with $\Delta^{in\mathcal{S}_b}$ and $\Delta^{out\mathcal{S}_b}$ its incoming and outgoing transitions.
First, the algorithm considers the states $s_{\mathcal{I}_{cl}}$ of the intersection automaton (Line~\ref{algorithm:subpropertyIdentificationIntersectionState}).
 If the corresponding state of the model $q$ is a box (Line~\ref{algorithm:subpropertyTransparent}), then $s_{\mathcal{I}_{cl}}$  is added to the set $Q_{\mathcal{P}_{q}}$ of the states  of the sub-property $\mathcal{S}_{q}$ (Line~\ref{algorithm:subpropertyaddingTransparent}). 
 If $s_{\mathcal{I}_{cl}}$  is  initial (Line~\ref{algorithm:subpropertyaddingInitial}) or accepting (Line~\ref{algorithm:subpropertyaddingAccepting}), it is also added to the initial (Line~\ref{algorithm:subpropertyaddingInitialState}) and accepting (Line~\ref{algorithm:subpropertyaddingAcceptingState})  states of $\mathcal{P}_{b}$.
 
\begin{algorithm}[t]
\caption{Identifies the automata associated with the sub-properties and their incoming  and outgoing transitions.}
\label{algorithm:subpropertyIdentification}
\begin{algorithmic}[1]
\Procedure{SubPropertyIdentification}{$\mathcal{I}_{cl}$, $\mathcal{M}$}
\For{$s_{\mathcal{I}_{cl}} \in Q_{\mathcal{I}_{cl}}$} \label{algorithm:subpropertyIdentificationIntersectionState}
\If{$s_{\mathcal{I}_{cl}}= \langle q, p, x\rangle \wedge q \in B_\mathcal{M} $}  \label{algorithm:subpropertyTransparent}
\State $Q_{\mathcal{P}_{q}} \leftarrow Q_{\mathcal{P}_{q}} \cup \left \{ s_{\mathcal{I}_{cl}} \right \}$; \label{algorithm:subpropertyaddingTransparent}
\If{$s_{\mathcal{I}_{cl}} \in  Q^0_{\mathcal{I}_{cl}}$} \label{algorithm:subpropertyaddingInitial}
\State $Q^0_{\mathcal{P}_{q}} \leftarrow Q^0_{\mathcal{P}_{q}}  \cup \left \{ s_{\mathcal{I}_{cl}} \right \}$; \label{algorithm:subpropertyaddingInitialState}
\EndIf
\If{$s_{\mathcal{I}_{cl}} \in  F_{\mathcal{I}_{cl}}$} \label{algorithm:subpropertyaddingAccepting}
\State $F_{\mathcal{P}_{q}} \leftarrow F_{\mathcal{P}_{q}}  \cup \left \{ s_{\mathcal{I}_{cl}} \right \}$; \label{algorithm:subpropertyaddingAcceptingState}
\EndIf
\EndIf
\EndFor \label{algorithm:subpropertyIdentificationIntersectionStateEnd}
\For{$(\langle b, p , x\rangle, a, \langle b, p^\prime , y\rangle) \in \Delta^p_{\mathcal{I}_{cl}}$ and $b \in B_\mathcal{M}$} \label{algorithm:subpropertyInternalTransitions}
\State $\Delta_{\mathcal{P}_{b}} = \Delta_{\mathcal{P}_{b}} \cup ( \langle b, p , x\rangle, a, \langle b, p^\prime, y\rangle ) $; \label{algorithm:subpropertyAddingInternalTransitions}
\EndFor
\For{$(\langle q, p , x\rangle, a, \langle q^\prime, p^\prime , y\rangle) \in \Delta^c_{\mathcal{I}_{cl}}$} \label{algorithm:subpropertyAddingExternalTransitions}
\If{$q \in B_\mathcal{M}$} \label{algorithm:subpropertyIdentificationSourceTransparent}
\State $\Delta^{in\mathcal{S}_{q}} \leftarrow \Delta^{in\mathcal{S}_{q}} \cup (\langle q, p , x\rangle, a, \langle q^\prime, p^\prime, y\rangle) $; \label{algorithm:subpropertyIdentificationOutgoing}
\EndIf
\If{$q^\prime \in B_\mathcal{M}$} \label{algorithm:subpropertyIdentificationDestinationTransparent}
\State $\Delta^{out\mathcal{S}_{q}} \leftarrow \Delta^{out\mathcal{S}_{q}} \cup  (\langle q, p , x\rangle, a, \langle q^\prime, p^\prime, y\rangle)$; \label{algorithm:subpropertyIdentificationIncoming}
\EndIf
\EndFor \label{algorithm:subpropertyInternalTransitionsEnd}
\EndProcedure
\end{algorithmic}
\end{algorithm}

Then, each transition $(\langle b, p, x\rangle, a, \langle b, p^\prime , y\rangle) $ of the intersection automaton which is possibly executed inside the box $b$ is considered (Line~\ref{algorithm:subpropertyInternalTransitions}) and added to the corresponding automaton (Line~\ref{algorithm:subpropertyAddingInternalTransitions}). Finally, each transition $(\langle q, p , x\rangle, a, \langle q^\prime, p^\prime, y\rangle)$ which is obtained by combining transitions of the model and of the claim is analyzed (Line~\ref{algorithm:subpropertyAddingExternalTransitions}). If the state of the model associated with the source state (Line~\ref{algorithm:subpropertyIdentificationSourceTransparent}) or the destination state (Line~\ref{algorithm:subpropertyIdentificationDestinationTransparent})  of the transition is a box, the transition is added to the set of the outgoing (Line~\ref{algorithm:subpropertyIdentificationOutgoing}) or incoming (Line~\ref{algorithm:subpropertyIdentificationIncoming}) transitions of the sub-property.

\begin{theorem}[Automaton extraction complexity]
\label{th:automatonExtractionComplexity}  The automaton extraction process can be performed in time $\mathcal{O}(|Q_{\mathcal{I}_{cl}}|+|\Delta_{\mathcal{I}_{cl}}|)$.
\end{theorem}
\begin{proof} Lines~\ref{algorithm:subpropertyIdentificationIntersectionState}-\ref{algorithm:subpropertyIdentificationIntersectionStateEnd} of Algorithm~\ref{algorithm:subpropertyIdentification} visit each state of the intersection automaton at most once, while Lines~\ref{algorithm:subpropertyInternalTransitions}-\ref{algorithm:subpropertyInternalTransitionsEnd} visit each transition of the intersection automaton exactly once. In both the cases, at each step, a constant number of operations is executed inducing a $\mathcal{O}(|Q_{\mathcal{I}_{cl}}|+|\Delta_{\mathcal{I}_{cl}}|)$ time complexity .
\end{proof}

It is important to notice that every word $v$ that is in the finite and infinite abstraction of the intersection automaton is a word associated with a sub-property and vice versa. 
\begin{theorem}[The language of the  sub-property corresponds to the abstraction of the intersection automaton]
\label{th:subpropertyCorrecteness}
Given the model $\mathcal{M}$ which possibly satisfies the claim $\phi$, the intersection automaton $\mathcal{I}=\mathcal{M} \cap \mathcal{A}_{\neg \phi}$ and the set of sub-properties $\zeta$ obtained as specified in Definition~\ref{def:subidentification}, for every box $b$:
\begin{enumerate}
\item  \label{th:subpropertyfiniteCorrecteness} $v \in \alpha^\ast_b(\mathcal{I}) \Leftrightarrow v \in (\mathcal{L}^{e\ast}(\mathcal{S}_b) \cup \mathcal{L}^{i\ast}(\mathcal{S}_b))$;
\item  \label{th:subpropertyinfiniteCorrecteness} $v \in \alpha^\omega_b(\mathcal{I}) \Leftrightarrow v \in (\mathcal{L}^{e\omega}(\mathcal{S}_b) \cup \mathcal{L}^{i\omega}(\mathcal{S}_b))$.
\end{enumerate}
\end{theorem}
\noindent Theorem~\ref{th:subpropertyCorrecteness} specifies that the words of the finite abstractions of the intersection automaton associated with the box $b$ correspond to the union of the external and internal finite words associated with the sub-property $\mathcal{S}_b$ (condition~\ref{th:subpropertyfiniteCorrecteness}). Furthermore,  the words of the infinite abstractions of the intersection automaton associated with the box $b$ correspond to the union of the external and internal infinite words associated with the sub-property $\mathcal{S}_b$ (condition~\ref{th:subpropertyinfiniteCorrecteness}). 
\begin{proof} Let us first consider the statement~\ref{th:subpropertyfiniteCorrecteness}.

($\Leftarrow$) Let us first consider the case in which $v$ is in the set of finite externally accepted words associated with the sub-property $\mathcal{S}_b$, i.e., $v \in \mathcal{L}^{e\ast}(\mathcal{S}_b)$.  We want to construct a run $\rho_{\mathcal{I}}(i)$ whose abstraction corresponds to the word $v$ in the intersection automaton. If such run exists we can conclude that $v \in \alpha^\ast_b(\mathcal{I})$. 
Thus, we write the word $v$ as $\nu_{init} \nu^\ast \nu_{out}$. By definition, $\nu_{init}$ must be associated to an incoming transition of the sub-property. 
Thus, it is possible to associate to $\rho_{\mathcal{I}}(0)$ and $\rho_{\mathcal{I}}(1)$ the  states of the  intersection associated to the source and the destination of the incoming transition of the sub-property. 
Then, each state $\rho_{\mathcal{S}_b}(i)$ of the run that makes the word $\nu^\ast$ externally accepted is associated with the corresponding state $\rho_{\mathcal{I}}(i)$ of the intersection automaton (which must exist from construction). This implies that the run $\rho_{\mathcal{I}}(1)$ satisfies the conditions~\ref{def:portionOfTheWord} and~\ref{def:finiteRunAbstractionStateTransparent} of Definition~\ref{def:finiteRunAbstraction}.
Since $(\rho_{\mathcal{S}_b}(i), v_i, \rho_{\mathcal{S}_b}(i+1)) \in \Delta_{\mathcal{S}_b}$ and $(\rho_{\mathcal{I}}(i), v_i, \rho_{\mathcal{I}}(i+1)) \in \Delta^p_{\mathcal{I}}$ by construction, it implies that the run $\rho_{\mathcal{S}_b}$ is also executable on the automaton $\mathcal{I}$, i.e., condition~\ref{def:finiteRunAbstractionTransparentTransition} is satisfied. 
Furthermore, since by construction and Definition~\ref{def:finiteExternalRun}, $\rho_{\mathcal{S}_b}(0)\in Q^{0in_b}$ (i.e., the source state of the run must be a source of an incoming transition) and $\rho_{\mathcal{S}_b}(|v|)\in F^{out_b}$ (i.e., the last state of the run must be the destination of an outgoing transition), conditions~\ref{def:finiteRunAbstractionEnding} and \ref{def:finiteRunAbstractionSource} are satisfied. Since we have found a finite abstracted run $\rho_{\mathcal{I}}(i)$ associated to $v$ and the box $b$, we conclude that $v \in \alpha^\ast_b(\mathcal{I}) $.\\
The same approach applies to the case in which $v$ is in the set of finite internally accepted words associated with the sub-property $\mathcal{S}_b$, i.e., $v \in \mathcal{L}^{i\ast}(\mathcal{S}_b)$. The only difference concerns the initial state $\rho_{\mathcal{S}_b}(0)$, which by Definition~\ref{def:finiteExternalRun} must be an initial state of the sub-property, i.e., $\rho_{\mathcal{S}_b}(0)$ is in $ Q^0_{\mathcal{P}_{b}}$ that by construction implies that $\rho_{\mathcal{I}}(0) \in Q^0_{\mathcal{I}}$. This implies that $\rho_{\mathcal{S}_b}(0)$ can be considered as the initial state of the run corresponding to the word, making $i=0$ in Definition~\ref{def:finiteRunAbstraction} condition~\ref{def:finiteRunAbstractionSource}.

($\Rightarrow$) The proof is by contradiction. 
Assume that $v \in (\mathcal{L}^{e\ast}(\mathcal{S}_b) \cup \mathcal{L}^{i\ast}(\mathcal{S}_b))$ is false and $v^\ast \in \alpha^\ast_b(\mathcal{I})$ is true. 
Since $v \in \alpha^\ast_b(\mathcal{I})$, there exists an infinite possibly accepting run in $\mathcal{I}$ and $v$ is an abstraction of the corresponding word which is associated to the box $b$. 
Since $\mathcal{S}_b$ is obtained from the intersection automaton by aggregating the portions of the state space that refer to the box  $b$, it follows that $v$ must be associated with a run that traverses $\mathcal{S}_b$ (either starting from one of the initial states of  $\mathcal{S}_b$ and reaching one of its outgoing transitions, or starting from an incoming transition of  $\mathcal{S}_b$ and reaching one of its outgoing transitions). 
This implies that $v^\ast \in (\mathcal{L}^{e\ast}(\mathcal{S}_b) \cup \mathcal{L}^{i\ast}(\mathcal{S}_b))$ which makes the hypothesis contradicted.

The proof of the statement~\ref{th:subpropertyinfiniteCorrecteness}  corresponds to the proof~\ref{th:subpropertyfiniteCorrecteness} with the exception that it considers infinite words.
\end{proof}

The second step of the constraint computation identifies for the sub-properties of each box $b$ the subsets $G$ and $R$ of the incoming and outgoing transitions of its sub-properties.
The sets $G$ and $R$ contains different incoming and outgoing transitions depending on whether property $S$ or $S_p$ is considered.
In the case of $S$, the set of incoming  transitions contained in $G$ contains the incoming transitions reachable in the intersection automaton without crossing  states of the intersection obtained from a black box state of the model.
In the case of $S_p$, the set of the incoming transitions contained in $G$ contains the incoming transitions reachable in the intersection automaton without crossing states of the intersection obtained from the ``only" the box $b$.
Similarly, in the case of $S$  the set of outgoing  transitions contained in $R$ contains the outgoing transitions from which it is possible to reach in the intersection automaton an accepting state that can be entered infinitely often without crossing a  state of the intersection obtained from a black box state of the model.
When $S_p$ is considered,   the set of outgoing  transitions contained in $R$ contains the outgoing transitions from which  it is possible to reach in the intersection automaton an accepting state that can be entered infinitely often without crossing a  state of the intersection obtained from the box $b$.

\begin{mydef}[The sets $G$ and $R$]
\label{def:coloring}  
Given a sub-property $\mathcal{S}$,  
the sets $\Delta^{in\mathcal{S}}$ and $\Delta^{out\mathcal{S}}$ of its incoming and outgoing transitions and the intersection automaton $\mathcal{I}$.
A transition $\delta \in \Delta^{in\mathcal{S}} \cup \Delta^{out\mathcal{S}}$ obtained from the transition $\delta^\prime=(s,a,s^\prime)$ of the intersection automaton:
\begin{itemize}
 \item $\delta \in G       \Leftrightarrow   $ it exists $\rho_{\mathcal{I}}^\omega$ and  $i>0$, such that $\rho^\omega(i)=s$ and  $(\rho^\omega(i), a,$ $\rho^\omega(i+1)) =\delta^\prime$ and for all $ 0 \leq j \leq i, \rho^\omega(j) \in PR_{\mathcal{I}}$;
 \item 
 $ \delta \in R \Leftrightarrow $  it exists $\rho_{\mathcal{I}}^\omega$ and  an $i> 0$, such that $\rho^\omega(i)=s^\prime$ and $(\rho^\omega(i-1), a, \rho^\omega(i)) =\delta^\prime $ and for all $ j \geq i$, $\rho^\omega(j) \in PR_{\mathcal{I}}$.
\end{itemize}
Given a sub-property $\mathcal{S}_p$,  
the sets $\Delta^{in\mathcal{S}_p}$ and $\Delta^{out\mathcal{S}_p}$ of its incoming and outgoing transitions and the intersection automaton $\mathcal{I}$.
A transition $\delta \in \Delta^{in\mathcal{S}_p} \cup \Delta^{out\mathcal{S}_p}$ obtained from the transition $\delta^\prime=(s,a,s^\prime)$ of the intersection automaton:
\begin{itemize}
 \item $\delta \in G       \Leftrightarrow   $ it exists $\rho_{\mathcal{I}}^\omega$ and  $i>0$, such that $\rho^\omega(i)=s$ and  $(\rho^\omega(i), a,$ $\rho^\omega(i+1)) =\delta^\prime$ and for all $ 0 \leq j \leq i, \rho^\omega(j) \not \in Q_{\mathcal{P}}$;
 \item 
 $ \delta \in R \Leftrightarrow $  it exists $\rho_{\mathcal{I}}^\omega$ and  an $i> 0$, such that $\rho^\omega(i)=s^\prime$ and $(\rho^\omega(i-1), a, \rho^\omega(i)) =\delta^\prime $ and for all $ j \geq i$, $\rho^\omega(j) \not \in Q_{\mathcal{P}}$.
\end{itemize}
\end{mydef}

Let us first consider the sub-property $\mathcal{S}$. 
The incoming transitions in the set $G$ are the ones which are reachable from the initial state of the intersection automaton $\mathcal{I}$ without passing through mixed states, i.e., the incoming transitions whose reachability does not depend on the replacements of other boxes. 
The outgoing transitions marked as $R$ are the transitions from which an accepting state that can be entered infinitely often of the intersection automaton $\mathcal{I}$ is reachable without passing through mixed states, i.e., the outgoing transitions from which it is possible to reach a state that makes the property violated independently on the replacements of the other boxes. 
Let us now consider  sub-property $\mathcal{S}_p$.
The incoming transitions in the set $G$ also include transitions reachable from the initial  state of the intersection automaton $\mathcal{I}$ crossing mixed states which are not obtained from the box $b$. 
Indeed, these transitions are reachable by assuming a particular behavior of other boxes present in the model of the system. 
Similarly, the outgoing transitions in the set $R$ also include transitions from which an accepting state that can be entered infinitely often of the intersection automaton $\mathcal{I}$ by crossing states obtained from other black box states of the model.
Given the sub-property $\mathcal{S}$, the sets $G$ and $R$  specify that the runs violate the property $\phi$ of interest. 
When the developer designs the replacement of a box $b$ associated to a sub-property $\mathcal{S}$ he/she \emph{must not} design a component that allows  $\mathcal{S}$ to reach a outgoing port marked as $R$ from an incoming port marked as $G$. 
Indeed, in this case, it is providing the system a way to reach from an initial state an accepting state of the intersection automaton which can be entered infinitely often. 
Similarly, given the sub-property $\mathcal{S}$, the sets $G$ and $R$  specify that the runs that possibly violate the property $\phi$ of interest. 
When the developer designs the replacement of a box $b$ associated to a sub-property $\mathcal{S}$ he/she \emph{should not} design a component that allows  $\mathcal{S}$ to reach a outgoing port marked as $R$ from an incoming port marked as $G$. 
Indeed, in this case, he/she is not directly providing the  system a way to reach from an initial state an accepting state of the intersection automaton which can be entered infinitely often.
The presence of this run depends also on the replacements associated with  other boxes.

 \begin{sloppypar}
Given a box $b$, the sets $G$ and $R$ for the sub-properties $S$ and $S_p$ can be computed using the procedure described in Algorithm~\ref{algorithm:coloring}. The algorithm works in two steps which identify the incoming (Lines~\ref{alg:coloringFor}-\ref{alg:coloringEndFor}) and outgoing (Lines~\ref{alg:redOutgoingSearch}-\ref{alg:yellowOutgoingSearch}) transitions to be inserted into the sets $G$ and $R$ for each of the sub-properties $S$ and $S_p$. 
The incoming transitions marked as $G$ (Lines~\ref{alg:greenIncomingSearch}-\ref{alg:yellowIncomingSearch}) are computed through the function \textsc{Forward$\Pi$Identifier} invoked over two different sets of states.  
When the incoming transitions of the sub-property $S$ to be marked as $G$ are searched (Line~\ref{alg:greenIncomingSearch}), only purely regular states $PR_{\mathcal{I}}$ of the automaton $\mathcal{I}_{cl}$ are traversed.
When the incoming transitions of the sub-property $S_p$ are considered (Line~\ref{alg:yellowIncomingSearch}), all the states of the intersection automaton $\mathcal{I}_{cl}$ which are not states $Q_{\mathcal{P}}$ of the  automaton $\mathcal{P}$ associated with the sub-property $\mathcal{S}_p$ can be explored.
The outgoing transitions in the sets $R$ are computed through the function  \textsc{Backward$\Pi$Identifier}.
When the outgoing transitions of sub-property $S$ are considered (Line~\ref{alg:redOutgoingSearch}), only the purely regular states $PR_{\mathcal{I}_{cl}}$ of the automaton $\mathcal{I}_{cl}$ can be traversed.
When the outgoing transitions of sub-property $S_p$ are analyzed (Line~\ref{alg:yellowOutgoingSearch}), all the states of the intersection automaton $\mathcal{I}_{cl}$ which are not states of the set $Q_{\mathcal{P}}$ of the  automaton $\mathcal{P}$ associated with the sub-property $\mathcal{S}_p$ are considered.
\end{sloppypar}

\begin{algorithm}[t]
\caption{Computation of the function $\Pi$.}
\label{algorithm:coloring}
\begin{algorithmic}[1]
\Procedure{$\Pi$Identifier}{$\mathcal{I}_{cl}$}
\For{$q_0 \in Q^0_{\mathcal{I}} \cap PR_{\mathcal{I}_{cl}}$} \label{alg:coloringFor}
\State \Call{Forward$\Pi$Identifier}{$q_0$, $\mathcal{I}_{cl}$, $PR_{\mathcal{I}_{cl}}$, $S$}; \label{alg:greenIncomingSearch}
\State \Call{Forward$\Pi$Identifier}{$q_0$, $\mathcal{I}_{cl}$, $Q_{\mathcal{I}_{cl}} \setminus Q_{\mathcal{P}}$, $S_p$}; \label{alg:yellowIncomingSearch}
\EndFor\label{alg:coloringEndFor}
\State \Call{Backward$\Pi$Identifier}{$\mathcal{I}_{cl}$, $PR_{\mathcal{I}_{cl}}$, $S$}; \label{alg:redOutgoingSearch}
\State \Call{Backward$\Pi$Identifier}{$\mathcal{I}_{cl}$, $Q_{\mathcal{I}_{cl}} \setminus Q_{\mathcal{P}}$, $S_p$}; \label{alg:yellowOutgoingSearch}
\EndProcedure
\end{algorithmic}
\end{algorithm}

\begin{algorithm} 
\caption{The procedure to find the incoming transitions to be marked as $G$ and $Y$.}
\label{algorithm:greenIncomingSearch}
\begin{algorithmic}[1]
\Procedure{Forward$\Pi$Identifier}{$s$, $\mathcal{I}_{cl}$, $Q$, $S$}
\State hash(s); \label{alg:GreenIncomingSearch}
\For{$(s, a, s^\prime) \in \Delta_{\mathcal{I}_{cl}}$} ; \label{alg:GreenIncomingOutgoing}
\If{$s^\prime \in Q$} \label{alg:GreenIncomingPurelyRegular}
\If{$s^\prime$ not hashed} \label{alg:notHashed}
\State \Call{Forward$\Pi$Identifier}{$s^\prime$, $\mathcal{I}_{cl}$, $Q$, $C$}; \label{alg:continuingTheGreenSearch}
\EndIf
\Else
\State $R=R \cup \{(s,a,s^\prime)\}$;  \label{alg:continuingTheGreenMarking}
\EndIf
\EndFor
\EndProcedure
\end{algorithmic}
\end{algorithm}

The \textsc{Forward$\Pi$Identifier} procedure described in Algorithm~\ref{algorithm:greenIncomingSearch} starts from a state $s$ and searches for runs that involve only states in the set $Q$ (passed as parameter) until an incoming transition of a sub-property $S$ ($S_p$) is reached. 
Whenever a state $s$ is visited by the algorithm, it is hashed (Line~\ref{alg:GreenIncomingSearch}), then, each outgoing transition $(s, a, s^\prime) $ of the state $s$ (Line~\ref{alg:GreenIncomingOutgoing}) is analyzed. If the destination state $s^\prime$ is in $Q$ (Line~\ref{alg:GreenIncomingPurelyRegular}) and it has not already been hashed (Line~\ref{alg:notHashed}), then the \textsc{Forward$\Pi$Identifier} is continued (Line~\ref{alg:continuingTheGreenSearch}). If this is not the case, it means that the state $s^\prime$ is the destination of the incoming port $(s, a, s^\prime)$, thus, $(s, a, s^\prime)$ is added into the set $R$  (Line~\ref{alg:continuingTheGreenMarking}) of the sub-property $S$ or $S_p$ (depending on the sub-property that is currently analyzed).

\begin{theorem} [\textsc{Forward$\Pi$Identifier} complexity] The procedure described in  Algorithm~\ref{algorithm:greenIncomingSearch} can be performed in time $\mathcal{O}(|Q_{\mathcal{I}_{cl}}|+|\Delta_{\mathcal{I}_{cl}}|)$.
\end{theorem}
\begin{proof}
It is easy to prove that each state and transition of the cleaned intersection automaton $\mathcal{I}_{cl}$ is visited at most once, since it is visited if and only if it has not been hashed before. At each step, a finite and constant number of operations is performed  leading the $\mathcal{O}(|Q_{\mathcal{I}_{cl}}|+|\Delta_{\mathcal{I}_{cl}}|)$ time complexity.
\end{proof}

\begin{algorithm} 
\caption{The procedure to find the outgoing transitions to be marked as $R$ and $Y$.}
\label{algorithm:redOutgoingSearch}
\begin{algorithmic}[1]
\Procedure{Backward$\Pi$Identifier}{$\mathcal{I}_{cl}$, $Q$, $S$}
\State $\mathcal{I}_{Q} \leftarrow$\Call{abstract}{$\mathcal{I}_{cl}$, $Q$}; \label{alg:redOutgoingSearchPR}
\State $SCC \leftarrow$\Call{getNonTrivialSCC}{$\mathcal{I}_{Q}$}; \label{alg:redOutgoingSearchTarjan}
\State $next \leftarrow \left \{ \right \}$;  \label{alg:redOutgoingSearchNextDefinition}
\For{$scc \in SCC$} \label{alg:redOutgoingSearchNextDefinitionIsolating}
\If{$scc \cap F_{\mathcal{I}} \neq \emptyset$}
\State $next \leftarrow next \cup scc$; 
\EndIf \label{alg:redOutgoingSearchNextDefinitionIsolatingEnd}
\EndFor  \label{alg:redOutgoingSearchNextDefinitionEndFor}
\State $visited \leftarrow \left \{ \right \}$; \label{alg:redOutgoingSearchVisitedDefinition}
\While{$next \neq \emptyset$} \label{alg:redOutgoingSearchIteratively}
\State $s \leftarrow $ \Call{choose}{$next$};  \label{alg:redOutgoingSearchIterativelyChoose}
\State $visited \leftarrow visited \cup  \left \{ s \right \} $; \label{alg:redOutgoingSearchIterativelyAddingVisited}
\State $next \leftarrow next \setminus \left \{ s \right \}$;  \label{alg:redOutgoingSearchIterativelyRemoveNext}
\For{$(s^\prime, a, s) \in \Delta_{\mathcal{I}}$} \label{alg:redOutgoingSearchIterativelyIncoming}
\If{$s^\prime \in Q_{\mathcal{P}}$}  \label{alg:redOutgoingSearchMixed}
\State $R=R\cup \{(s^\prime, a, s) \} $; \label{alg:redOutgoingSearchMarkingRED}
\Else
\If{$s^\prime \not \in visited$} \label{alg:redOutgoingNotVisited}
\State $next \leftarrow next \cup \left \{ s^\prime \right \}$; \label{alg:visitedNext}
\EndIf
\EndIf
\EndFor
\EndWhile \label{alg:redOutgoingSearchIterativelyEnd}
\EndProcedure
\end{algorithmic}
\end{algorithm}

The \textsc{Backward$\Pi$Identifier} procedure is described in algorithm~\ref{algorithm:redOutgoingSearch}.
The  algorithm first looks for the non trivial strongly connected components that involve only states which are in the set $Q$ (passed as parameter). 
To this purpose the algorithm constructs a version $\mathcal{I}_{Q}$ of $\mathcal{I}_{cl}$  that contains only the states of $\mathcal{I}_{cl}$ that belongs to $Q$ (Line~\ref{alg:redOutgoingSearchPR}). 
Note that, depending on whether the sub-property $S$ or $S_p$ is considered, the states mixed states of the intersection automaton or the states of the automaton not in $Q_\mathcal{P}$ are contained in $Q$.
Then, the non trivial strongly connected components of  $\mathcal{I}_{Q}$ are identified (Line~\ref{alg:redOutgoingSearchTarjan}). The set $next$ is initialized to contain all the strongly connected components that contains at least a state which is accepting (Lines~\ref{alg:redOutgoingSearchNextDefinition}-\ref{alg:redOutgoingSearchNextDefinitionEndFor}).

Then, the state space of $\mathcal{I}$ is explored to compute the outgoing transitions from which it is possible to reach one of the states in the set $next$. The set $visited$ (Line~\ref{alg:redOutgoingSearchVisitedDefinition}) is used to keep track of the already visited states of $\mathcal{I}$. The algorithm iteratively chooses a state $s$ in the set $next$ (Lines~\ref{alg:redOutgoingSearchIteratively},\ref{alg:redOutgoingSearchIterativelyChoose}), which is removed from $next$ (Line~\ref{alg:redOutgoingSearchIterativelyRemoveNext}) and added to the set of visited states (Line~\ref{alg:redOutgoingSearchIterativelyAddingVisited}). 
For each incoming transition $(s^\prime, a, s) \in \Delta_{\mathcal{I}}$ of $s$ (Line~\ref{alg:redOutgoingSearchIterativelyIncoming}), if the state $s^\prime$ is a state of the sub-property (Line~\ref{alg:redOutgoingSearchMixed}) the corresponding transition is added into the set $R$ of the sub-property passed as parameter (Line~\ref{alg:redOutgoingSearchMarkingRED}).

 Otherwise, if the purely regular state $s^\prime$ has not already been visited (Line~\ref{alg:redOutgoingNotVisited}) it is added to the set $next$ of states to be considered next (Line~\ref{alg:visitedNext}).

\begin{theorem} [\textsc{Backward$\Pi$Identifier}  complexity] The procedure described in  Algorithm~\ref{algorithm:redOutgoingSearch} can be performed in time $\mathcal{O}(|Q_{\mathcal{I}_{cl}}|+|\Delta_{\mathcal{I}_{cl}}|)$.
\end{theorem}

\begin{proof}
The algorithm first extracts a version $\mathcal{I}_Q$ of the intersection automaton which contains only the  states in $Q$ (Line~\ref{alg:redOutgoingSearchPR}). This can be done by exploring the state space of the intersection automaton $\mathcal{I}_{cl}$ and isolating the portions of the state space of interest. Then (Line~\ref{alg:redOutgoingSearchTarjan}), the list of non trivial strongly connected components $SCC$ is isolated. This can be performed in time $\mathcal{O}(|Q_{\mathcal{I}_{cl}}|+|\Delta_{\mathcal{I}_{cl}}|)$, for example by using the well known Tarjan's Algorithm~\cite{tarjan1972depth}. The only states added to the set $next$ are the states that belong to a strongly connected component which contains at least one accepting state (Lines~\ref{alg:redOutgoingSearchNextDefinitionIsolating}-\ref{alg:redOutgoingSearchNextDefinitionIsolatingEnd}). Starting from these states, the backward search is performed (Lines~\ref{alg:redOutgoingSearchIteratively}-\ref{alg:redOutgoingSearchIterativelyEnd}). This search visits each state and transition at most once. Thus, the \textsc{Backward$\Pi$Identifier} algorithm cab be performed in time $\mathcal{O}(|Q_{\mathcal{I}_{cl}}|+|\Delta_{\mathcal{I}_{cl}}|)$. 
\end{proof}

%------------------------------------------------------------------------------------------------------
% Reachability Relation
%------------------------------------------------------------------------------------------------------
The last step of the sub-property identification procedure concerns the \emph{computation of the reachability relation} $K$. 
The reachability relation specifies how the presence of a run that traverses a sub-property influences the reachability of another run that traverses the sub-property itself. Imagine for example that the high level model of the system is the one presented in Figure~\ref{fig:modelingExample2}. 
Differently from the model presented in Figure~\ref{Fig:IFSAExample}, in this case, the box $send_2$ is not contained in the IBA and $send_1$ can be left also through a transition that moves the system into the state $q_4$. Whenever the system reaches the state $q_4$, a timer is started. The transition $7$ is fired whenever the $timer\_ack$ proposition is true, i.e., the system is notified that the time has been elapsed. The developer can now choose to propose a replacement for the box $send_1$ that behaves as follows. Whenever the replacement is entered through the transition $1$, a sending activity is performed. If the sending activity succeeds, the transition $3$ is fired, otherwise, the transition $6$ which activates the timer is performed. If, instead, the replacement is entered through the transition $7$   and the $send$ activity fails, the transition $2$ is fired. 
Imagine that one of the properties of the system specifies that only one sending message activity must be performed by the system. 
The developer must know that if a $send$ activity is performed on a run that connects $1$ to $6$, this activity cannot be replicated in the run that connects  $7$ to $2$ and vice versa. This is exactly the purpose of the reachability relation which specifies how the internal runs of a sub-property influence each others.

\begin{figure*}[htpb]
\centering
\includegraphics[scale=0.5]{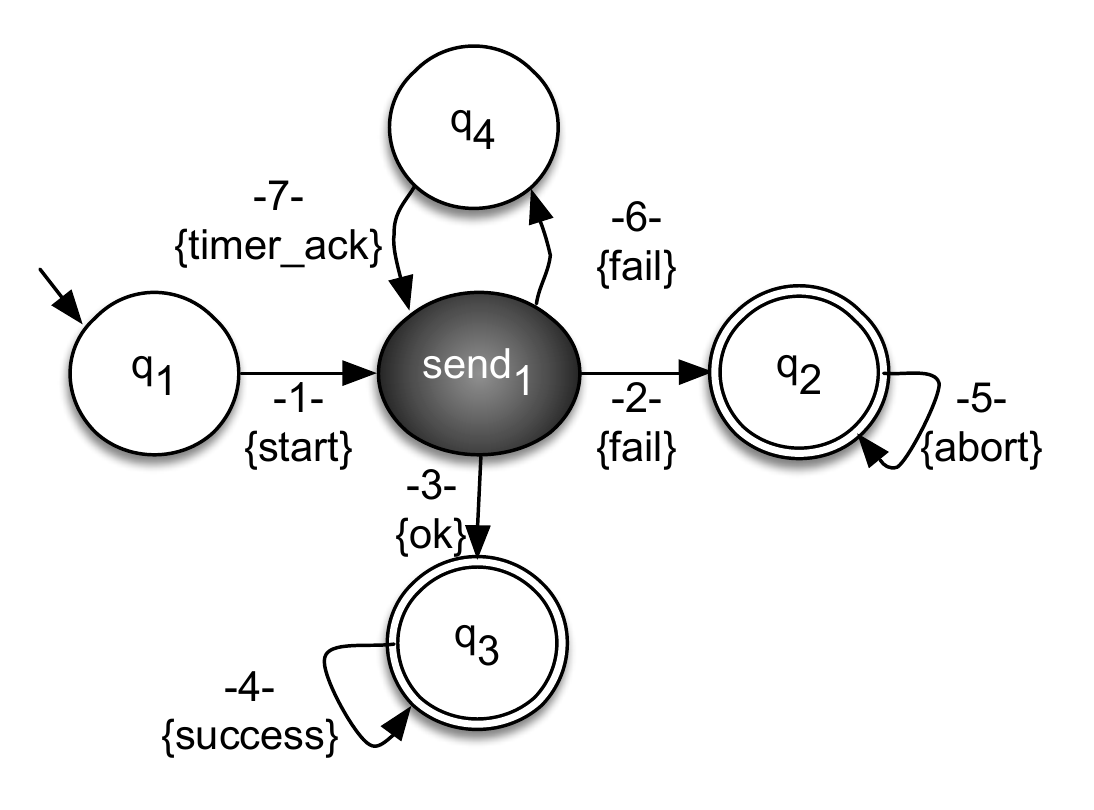}
\centering
\caption{A modeling alternative for the sending message protocol example.}
\label{fig:modelingExample2}
\end{figure*}

The reachability relation is computed by \emph{abstracting} the portion of the state space that connects the two runs of the sub-property $S$ ($S_p$) in the intersection automaton. 
It specifies if from an outgoing transition $\delta_{out} \in \Delta^{out\mathcal{S}}$ of the sub-property it is possible to reach its incoming transition $\delta_{in} \in \Delta^{in\mathcal{S}}$. 
Two versions of the reachability relation are computed depending on whether sub-property $S$ or $S_p$ is considered. 
In the first case, states of the intersection automaton obtained from other boxes can not be crossed, i.e., the replacements of other boxes do not allow reaching their outgoing transitions.
In the second case, all the states of the intersection automaton which are not obtained from the box $b$ can be traversed. 
The reachability relation will be graphically specified through directed dotted edges which connects the outgoing transitions of the sub-property with the corresponding incoming transitions.

\begin{mydef}[Reachability relation for the sub-property $\mathcal{S}$]
\label{def:portsReachabilityGraph}  Given a sub-property $\mathcal{S}$  associated with the intersection automaton $\mathcal{I}$ and the set of its incoming $\Delta^{in\mathcal{S}}$ and outgoing transitions  $\Delta^{out\mathcal{S}}$, the lower reachability relation $K=\Delta^{out\mathcal{S}} \times \Delta^{in\mathcal{S}}$ is a relation, such that given an outgoing transition $\delta_o$ obtained from the transition $\delta_o^{\prime}=(s,a,s^\prime)$ of  $\mathcal{I}$  and an incoming transition $\delta_i$ obtained from the incoming transition $\delta_i^{\prime}=(s^{\prime\prime},a,s^{\prime\prime\prime})$, $(\delta_o, \delta_i) \in \aleph_c$ if and only if one of the following conditions is satisfied:
\begin{enumerate}
\item \label{def:portsReachabilityGraphStatement1}   $\delta_o^\prime= \delta_i^\prime$; 
\item  \label{def:portsReachabilityGraphStatement2} there exist an accepting run $\rho^\omega$, and two indexes $i, j \geq 0$ such that  $s^\prime =\rho^\omega(i)$ and $s^{\prime\prime}=\rho^\omega(j)$  and for all $k$, such that $  j \geq k \geq i, \rho^\omega(k) \in PR_{\mathcal{I}_{cl}}$.
\end{enumerate}
\end{mydef}

In other words, the reachability relation specifies how outgoing and incoming transitions of the sub-property are connected in the intersection automaton. 
\begin{proposition}
\label{prop:lowerReachabilityDimensition}
In the worst case, the reachability relation $K$ for the sub-property $S$ contains $|\Delta^{out\mathcal{S}}| \cdot |\Delta^{in\mathcal{S}}|$ elements.
\end{proposition}

%------------------------------------------------------------------------------------------------------
% UPPER Reachability Relation
%------------------------------------------------------------------------------------------------------
When the sub-property $S_p$ which refer to a black box state $b$ is considered, an incoming and outgoing transition are in $K$ if and only if there exists a run between them that potentially involves also mixed states of the intersection automaton which are not obtained from $b$, i.e., any state in $Q_\mathcal{I} \setminus Q_{\mathcal{P}}$.
\begin{mydef}[Reachability relation for the sub-property $\mathcal{S}_p$]
\label{def:portsReachabilityGraphUpperBound}  Given a sub-property $\mathcal{S}_p$  associated to the intersection automaton $\mathcal{I}$ and the set of its incoming $\Delta^{in\mathcal{S}}$ and outgoing transitions  $\Delta^{out\mathcal{S}}$, the  reachability relation $K=\Delta^{out\mathcal{S}} \times \Delta^{in\mathcal{S}}$ is a relation such that, given an outgoing transition $\delta_o$ obtained from the transition $\delta_o^{\prime}=(s,a,s^\prime)$ of  $\mathcal{I}$  and an incoming transition $\delta_i$ obtained from the incoming transition $\delta_i^{\prime}=(s^{\prime\prime},a,s^{\prime\prime\prime})$, $(\delta_o, \delta_i) \in \aleph$ if and only if one of the following is satisfied:
\begin{enumerate}
\item $\delta_o^\prime= \delta_i^\prime$;
\item there exist an accepting run $\rho^\omega$ and two indexes $i, j \geq 0$, such that $s^\prime =\rho^\omega(i)$ and $s^{\prime\prime}=\rho^\omega(j)$ and for all $k$ such that $j \geq k \geq i, \rho^\omega(k) \not \in Q_{\mathcal{P}}$.
\end{enumerate}
\end{mydef}
\begin{proposition}
\label{prop:upperReachabilityDimensition}
In the worst case, the reachability relation $K$ for the sub-property $S_p$ contains $|\Delta^{out\mathcal{S}}| \cdot |\Delta^{in\mathcal{S}}|$ elements.
\end{proposition}

The procedure used to compute the reachability relation for $S$ and $S_p$ is described in Algorithm~\ref{algorithm:rearchabilityRelation}, where $\Delta^{\mathcal{S}}$ and $\Delta^{in\mathcal{S}}$
 are the outgoing and incoming transitions of the sub-property $\mathcal{S}$ ($\mathcal{S}_p$) to be considered, $\mathcal{I}_{cl}$ is the intersection automaton, $Q$ is the set of the states to be considered in the computation of the reachability relation and $K$ contains the reachability relation. 
 When sub-property $\mathcal{S}$ is considered,  the set $Q$ contains all the purely regular  states of the intersection automaton.
Instead, when sub-property $\mathcal{S}_p$ is analyzed, the set $Q$ contains all the states of the intersection automaton  with the exception of the states in $Q_{\mathcal{P}}$.

The procedure described in Algorithm~\ref{algorithm:rearchabilityRelation} first computes an abstraction of the state space which only contains the states in the set $Q$ (Line~\ref{alg:ReachabilityRelationIdentifierAbstract}). Then (Line~\ref{alg:ReachabilityRelationIdentifierFloydWs}), for every pair of states $(s, s^\prime)$, the Floyd-Warshall algorithm~\cite{gross2004handbook} is used to compute if $s^\prime$ is reachable from $s$. 
Then, each transition $\delta_o^\prime=(s, a, s^\prime)$ (Line~\ref{alg:ReachabilityRelationIdentifierIncoming}) and $\delta_i^\prime=(s^{\prime\prime}, a, s^{\prime\prime\prime})$  (Line~\ref{alg:ReachabilityRelationIdentifierOutgoing}) of the intersection automaton from which an outgoing $\delta_o$  and an incoming transition $\delta_i$ are analyzed. 
If it is possible to reach $s^{\prime\prime}$ from $s^{\prime}$ (Line~\ref{alg:ReachabilityRelationIdentifierReacCheck}), it means that there exists a run which contains only states in $Q$ which allows to reach $(s^{\prime\prime}, a, s^{\prime\prime\prime})$ from $(s, a, s^\prime)$. Thus, the pair $\langle \delta_o, \delta_i \rangle$ is added to the reachability relation $K$ (Line~\ref{alg:ReachabilityRelationIdentifierReach}).

\begin{algorithm} 
\caption{The procedure to compute the reachability relation.}
\label{algorithm:rearchabilityRelation}
\begin{algorithmic}[1]
\Procedure{ReachabilityRelationIdentifier}{$\Delta^{out\mathcal{S}}$, $\Delta^{in\mathcal{S}}$, $\mathcal{I}_{cl}$, $Q$, $K$}
\State $\mathcal{I}_{Q} \leftarrow$\Call{abstract}{$\mathcal{I}_{cl}$, $Q$}; \label{alg:ReachabilityRelationIdentifierAbstract}
\State $Rec \leftarrow$ \Call{FloydWarshall}{$\mathcal{I}_{Q}$}; \label{alg:ReachabilityRelationIdentifierFloydWs}
\For{$\delta_o^\prime=(s, a, s^\prime) \wedge \delta_o \in \Delta^{out\mathcal{S}}$} \label{alg:ReachabilityRelationIdentifierIncoming}
\For{$\delta_i^\prime=(s^{\prime\prime}, a, s^{\prime\prime\prime}) \wedge \delta_i \in \Delta^{in\mathcal{S}}$}  \label{alg:ReachabilityRelationIdentifierOutgoing}
\If{$((s^\prime, s^{\prime\prime}) \in Rec)$ or $((s, a, s^\prime)=(s^{\prime\prime}, a, s^{\prime\prime\prime}))$ } \label{alg:ReachabilityRelationIdentifierReacCheck}
\State $K=K \cup \langle \delta_o, \delta_i \rangle$; \label{alg:ReachabilityRelationIdentifierReach}
\EndIf
\EndFor
\EndFor  \label{alg:ReachabilityRelationIdentifierIncomingEnd}
\EndProcedure
\end{algorithmic}
\end{algorithm}

\begin{theorem} [\textsc{ReachabilityRelationIdentifier} correctness] The procedure described in Algorithm~\ref{algorithm:rearchabilityRelation} is correct.
\end{theorem}

\begin{proof} Let us first consider the case in which the sub-property $S$ is considered. 
It is necessary to prove that $(\delta_o, \delta_i) \in K$ if and only if one of the conditions~\ref{def:portsReachabilityGraphStatement1} or \ref{def:portsReachabilityGraphStatement2} of Definition~\ref{def:portsReachabilityGraph} is satisfied. 

($\Leftarrow$) If $\delta_o^\prime$ is equal to $\delta_i^\prime$ (condition~\ref{def:portsReachabilityGraphStatement1}), $(\delta_o, \delta_i)$ is added in the reachability relation $K$ in Line~\ref{alg:ReachabilityRelationIdentifierReach} since the condition in Line~\ref{alg:ReachabilityRelationIdentifierReacCheck} is satisfied. If instead there exists a run,  which contains only purely regular states, that connects the state $s^\prime$ to the state $s^{\prime\prime}$ (condition~\ref{def:portsReachabilityGraphStatement2}), the tuple $(s^\prime, s^{\prime\prime})$ is added to the relation $Rec$ returned by the Floyd-Warshall algorithm, which makes  the condition in Line~\ref{alg:ReachabilityRelationIdentifierReacCheck} satisfied and implies that $(\delta_o, \delta_i)$ is added to  the reachability relation $K$.

($\Rightarrow$) If a tuple $(\delta_o, \delta_i)$ belongs to  $K$, the procedure described in Algorithm~\ref{algorithm:rearchabilityRelation} has added it to the lower reachability relation. To be added to this relation two cases are possible:
\begin{inparaenum}[\itshape a\upshape)]
\item the first clause of  condition specified in Line~\ref{alg:ReachabilityRelationIdentifierReacCheck} is triggered. Since $(s^\prime, s^{\prime\prime})$ is in $Rec$ it must exists a run made by states contained in the set $Q$ which allow to reach $s^{\prime\prime}$ from $s^\prime$ which implies that the condition~\ref{def:portsReachabilityGraphStatement1} is satisfied;
\item the second clause of the condition specified in Line~\ref{alg:ReachabilityRelationIdentifierReacCheck} is satisfied. In this case $\delta_o^\prime$ is equal to $\delta_i^\prime$ which makes  the condition~\ref{def:portsReachabilityGraphStatement2} satisfied.
\end{inparaenum}

The same approach can be used to demonstrate that the procedure is correct when the sub-property $S_p$ is considered.
\end{proof}

\begin{theorem}[Constraint computation complexity]
\label{th:reachabilityRelationComputationComplexity}  Given the sub-property $\mathcal{S}$ ($\mathcal{S}_p$), associated with the box $b$, the procedure described in Algorithm~\ref{algorithm:rearchabilityRelation} can be executed in time $\mathcal{O}(|Q^3|+|\Delta^{out\mathcal{S}}| \cdot  |\Delta^{\mathcal{S}}|)$.
\end{theorem}

\begin{proof}  As previously mentioned the abstraction procedure (Line~\ref{alg:ReachabilityRelationIdentifierAbstract}) can be executed in time $\mathcal{O}(|Q_{\mathcal{I}_{cl}}|+|\Delta_{\mathcal{I}_{cl}}|)$. 
The Floyd Warshall algorithm (Line~\ref{alg:ReachabilityRelationIdentifierFloydWs}) can be performed in time  $\mathcal{O}(Q^3)$, while the steps described in Lines~\ref{alg:ReachabilityRelationIdentifierIncoming}-\ref{alg:ReachabilityRelationIdentifierIncomingEnd} can be performed in time $|\Delta^{out\mathcal{S}}| \cdot  |\Delta^{in\mathcal{S}}|$ complexity. 
Thus,  Algorithm~\ref{algorithm:rearchabilityRelation} can be executed in time  $\mathcal{O}(|Q^3|+|\Delta^{out\mathcal{S}}| \cdot  |\Delta^{in\mathcal{S}}|)$.
\end{proof}

Together with the reachability relation the functions $\Gamma_{\mathcal{M}}$ and $\Gamma_{\mathcal{A}_{\neg \phi}}$ are computed both in the case of sub-property $S$ and $S_p$. For each tuple $(\delta_i, \delta_o) \in K$  these functions specify whether there exists a run that connect the outgoing and the incoming transitions that contains an intersection state made by an accepting state of the model ($\Gamma_\mathcal{M}$) and an intersection state made by an accepting state of the property ($\Gamma_{\mathcal{A}_{\neg \phi}}$), respectively. 
These functions specify the developer whether the presence of an accepting state in the replacement may lead to a violating run in the cases in which fairness conditions are considered.

A \emph{constraint} $C$ for a black box state $b$ is a tuple $\langle S, S_p \rangle$ which contains the sub-properties obtained as previously described.
Furthermore, a variable $\flag$ is associated with the value $T$ if there exists in the intersection automaton $\mathcal{I}_{\Upsilon}$  a  possible violating run which does not involve any state of the intersection automaton generated by the box $b \in B_\mathcal{M}$, $F$ otherwise.

\begin{mydef}[Constraint]
\label{def:constraint} Given the cleaned intersection automaton $\mathcal{I}_{\Upsilon}$ of the intersection automaton $\mathcal{I}=\mathcal{M} \cap \mathcal{A}_{\neg \phi}$ obtained from the IBA $\mathcal{M}$ and the BA $\mathcal{A}_{\neg \phi}$, the constraint $\mathcal{C}$ is made by a tuple  $\langle S, S_p \rangle$ and a value $\mathcal{U}$, such that
\begin{itemize}
\item $\flag=T \Leftrightarrow$  there exists an accepting run $\rho^\omega$ in $\mathcal{I}_{cl}$ such that for all $i \geq 0, \rho^\omega(i) \not \in Q_{\mathcal{P}}$
\end{itemize}

\end{mydef}

 The value of the function $\flag$  can be computed by running the emptiness checking procedure on the automaton $\mathcal{I}_{cl}$ by removing the portion of the state space containing states obtained from the box $b$ of the model.

\section{Replacement checking}
\label{sec:replacementChecking}
At each refinement round $i \in \mathcal{R}\mathcal{R}$, the developer produces a replacement $\mathcal{R}_b$ for a box  $b$ and wants to check if $\phi$ is satisfied by the new design. 
%The verification aims at checking whether the automaton $\mathcal{N}$, obtained by plugging the replacement $\mathcal{R}_b$ of the box $b$ into $$\model$$, satisfies, does not satisfy or possibly satisfies $\phi$.
%Whenever a replacement $\mathcal{R}_b$ for a box $b$ is proposed,
Two procedures can be employed. 
The \emph{refinement checking} procedure generates the refinement $\mathcal{N}$ as specified in Definition~\ref{def:plugginarefinement} and  checks  $\mathcal{N}$ against $\phi$. Roughly speaking, this means that the system would be verified from scratch at each refinement round.
The \emph{replacement checking} considers $\mathcal{R}_b$ against the previously generated constraint $\mathcal{C}$. 
In this case, the replacement is verified autonomously.

The refinement checking problem can  be  formulated as follows:
\begin{mydef} [Refinement Checking] Given a refinement round  $i \in \mathcal{R}\mathcal{R}$, where the developer refines the box $b$ of $\mathcal{M}$ through the replacement $\mathcal{R}_b$, the refinement checking problem is to compute whether the refined automaton $\mathcal{N}$, obtained by composing the replacement $\mathcal{R}_b$ of the box $b$ and the model $\model$, definitely satisfies, does not satisfy or possibly satisfies property $\phi$.
\end{mydef}

For example, assume that the box $send_1$ of the model $\model$, presented in Figure~\ref{Fig:IFSAExample}, is refined using the replacement $\mathcal{R}$ described in Figure~\ref{Fig:send1Refinement}. The replacement  $\mathcal{R}$, after it is entered through the incoming transition labeled with $start$, reaches the state $q_{14}$. Then,  the message is sent and the system moves from $q_{14}$ to $q_{15}$. After the sending activity, the system waits for a notification by moving to the state $q_{16}$. The state $q_{16}$ is a box, meaning that it still has to be refined. 
\begin{figure*}[ht]
    \centering
    \subfloat[The replacement $\mathcal{R}$ of the box $send_1$.]
     { \includegraphics[scale=0.5]{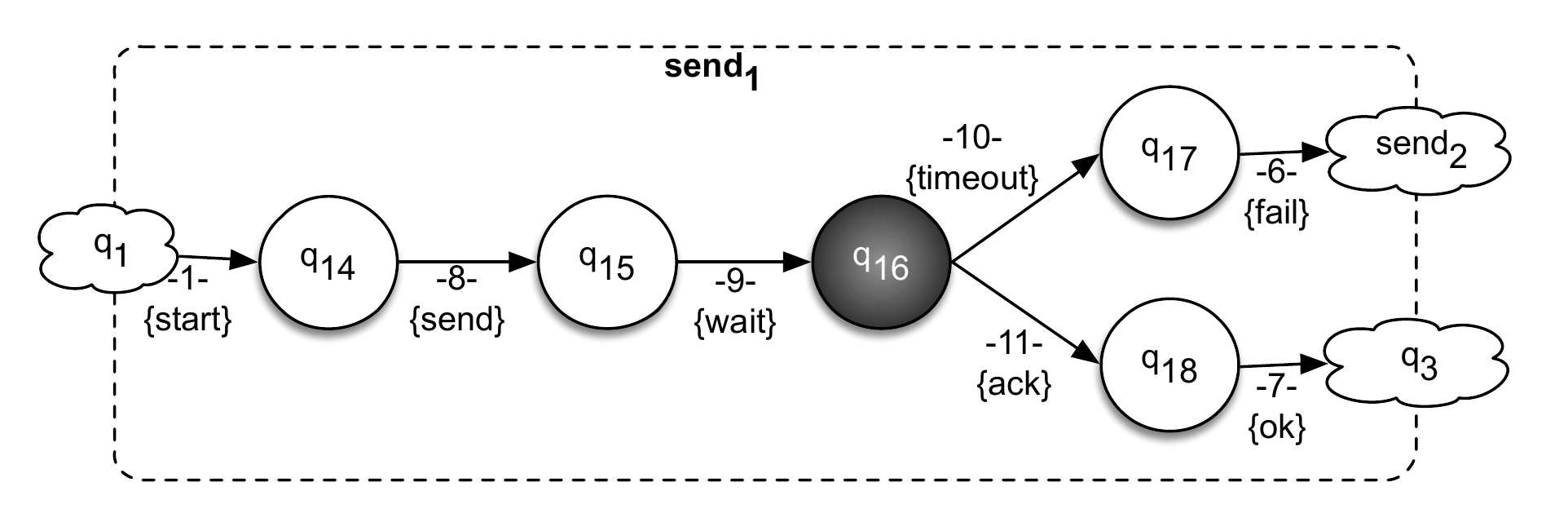}
   \label{Fig:send1Refinement}}\\
   \vspace{0.5pt}
    \subfloat[The refinement $\mathcal{N}$ obtained by plugging the replacement $\mathcal{R}$ into the model $\model$.]{
          \includegraphics[scale=0.5]{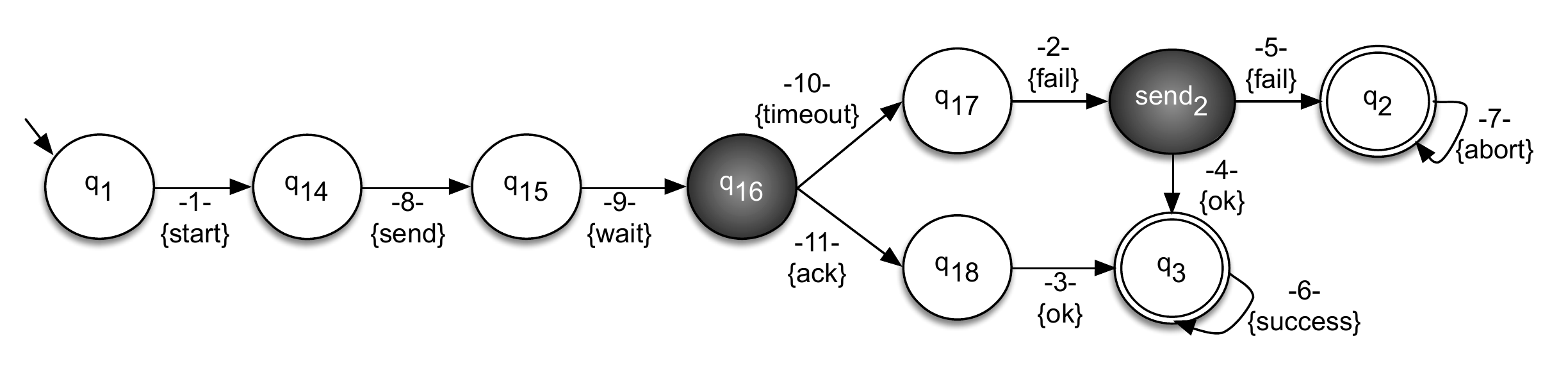}
              \label{fig:refinementModelSendMessage}}
    \caption{The replacement  $\mathcal{R}$ of the box $send_1$ of the model $\model$ presented in Figure~\ref{Fig:IFSAExample} and the refinement  $\mathcal{N}$ obtained by plugging the replacement $\mathcal{R}$ into the model $\model$.}
\end{figure*}
When the replacement of $q_{16}$ is left, two cases are possible: a $timeout$ event occurs and the system moves to the state $q_{17}$ or an $ack$ message is received which makes the system moving to $q_{18}$. In the first case, the replacement is left through the transition labeled with $fail$ which moves the system to the state $send_2$. In the second, the replacement is left through the transition labeled with $ok$ which moves the system to the state $q_3$. The replacement checking problem is the problem of verifying whether the refined model $\mathcal{N}$ (the initial model $\model$ plus the replacement of the box) satisfies the original property $\phi$. The refined model $\mathcal{N}$ is described in Figure~\ref{fig:refinementModelSendMessage}.

Whenever a replacement for a black box state $b$ is proposed, the idea is to not consider all the model from scratch, by generating its refinement $\mathcal{N}$, but to only consider the replacement $\mathcal{R}_b$ against the previously generated constraint $\mathcal{C}$. 
For example, the replacement of the box $send_1$ can be considered in relation to the sub-property $\mathcal{S}$ specified in Figure~\ref{fig:subpropertysend1}.
The sub-property specifies that any \emph{finite} path that crosses the box $send_1$, entering the box by means of the incoming transition $1$ (which arrives from the state $q_1$) and leaving the replacement through a transition marked with $fail$  (which reaches the state $send_2$) is a  possibly violating run. 
The run is possibly violating since if we do not satisfy the constraint we cannot claim that the property is not satisfied, since the violation depends on the replacement proposed for the other boxes.
Similarly, the possibly violating runs also include any \emph{finite} run entering from the transition which arrives from the state $q_1$ in which a $send$ is not followed by a $success$ before leaving the replacement through a transition marked with $fail$ and with destination the state $send_2$.

Checking whether a replacement satisfies a sub-property can be reduced to two emptiness checking problems. The first emptiness checking procedure considers an automaton which encodes the set of behaviors the system is going to exhibit at run-time (an \emph{under} approximation), and checks whether the property $\phi$ is violated. The second analyzes an automaton which also contains the behaviors the system may exhibit (an \emph{over} approximation). The under and the over approximation automaton are generated starting from a common intersection automaton. Section~\ref{sec:SubpropertyReplacementChecking} describes the intersection between a replacement and the corresponding sub-property and how the under and the over approximation are obtained from this intersection. Section~\ref{sec:replacementCheckingProcedure} presents the replacement checking procedure.

\subsection{Intersection between a sub-property and replacement}
\label{sec:SubpropertyReplacementChecking}
The basic version of the intersection automaton between a replacement $\mathcal{R}$  and the sub-property $\mathcal{S}$ of  the box $b$, from which the under and over approximation are computed, is described in Definition~\ref{def:subpropReplIntersection}.

\begin{mydef}[Intersection between a sub-property and a replacement]
\label{def:subpropReplIntersection} 
Given the sub-property $\mathcal{S}=\langle  \mathcal{P}, \Delta^{in\mathcal{S}},$ $ \Delta^{out\mathcal{S}}, G, R, K \rangle$, associated with the box $b$, and the replacement $\mathcal{R}=\langle \mathcal{T}, \Delta^{inR},$ $ \Delta^{outR} \rangle$, the intersection $\mathcal{I}=\mathcal{S} \cap R$ is a tuple  $\langle \mathcal{E}, \Delta^{in\mathcal{E}},$ $ \Delta^{out\mathcal{E}}, R, G \rangle$  such as:
\begin{itemize}
\item $\mathcal{E}$ is the intersection automaton. $\mathcal{E}$ is obtained as specified in Definition~\ref{def:intersection}, considering $ \mathcal{T}$ and $\mathcal{P}$ as model and claim, respectively;
\item $\Delta^{{in\mathcal{E}}} =\{ (q, a, \langle q^\prime, p, x \rangle ) \mid$  $(q, a, q^\prime) \in \Delta^{inR}$, $ (q, a, p) \in \Delta^{in\mathcal{S}}$ and $x \in \left \{ 0,1,2 \right \}    \}   $;
\item $\Delta^{out\mathcal{E}} = \{ ( \langle q^\prime, p, x \rangle, a, q) \mid$  $(q^\prime, a, q) \in \Delta^{outR}$,  $(p , a, q) \in \Delta^{out\mathcal{S}} $ and $ x \in \left \{ 0,1,2 \right \}    \}   $;
\item $G \subseteq \Delta^{{in\mathcal{E}}}$ contains the transitions of $\Delta^{{in\mathcal{E}}}$ obtained from a $G$ transition of $\mathcal{S}$;
\item $R \subseteq \Delta^{{out\mathcal{E}}}$ contains the transitions of $\Delta^{{out\mathcal{E}}}$ obtained from a $R$ transition of $\mathcal{S}$.
\end{itemize}
\end{mydef}

Informally, the intersection between a replacement $\mathcal{R}$ and the  sub-property $\mathcal{S}$ is an automaton which is obtained by the intersection of the automata associated with $\mathcal{R}$ and $\mathcal{S}$ and a set of incoming and outgoing transitions that corresponds to the synchronous execution of the incoming/outgoing transitions of $\mathcal{R}$ and $\mathcal{S}$. 
For example, the intersection between the replacement $\mathcal{R}$ described in Figure~\ref{Fig:send1Refinement} and the corresponding sub-property presented in Figure~\ref{fig:subpropertysend1} is presented in Figure~\ref{Fig:IntersectionSubPropertyReplacement}\footnote{Note that Figure~\ref{Fig:IntersectionSubPropertyReplacement} only contains the portion of the state space where $x=0$.}.
The set $G$ ($R$) contains the incoming (outgoing) transitions of the intersection obtained from a $G$ ($R$) incoming (outgoing) transition of the sub-property $\mathcal{S}$.

\begin{figure}[t]
\centering
\includegraphics[scale=0.49]{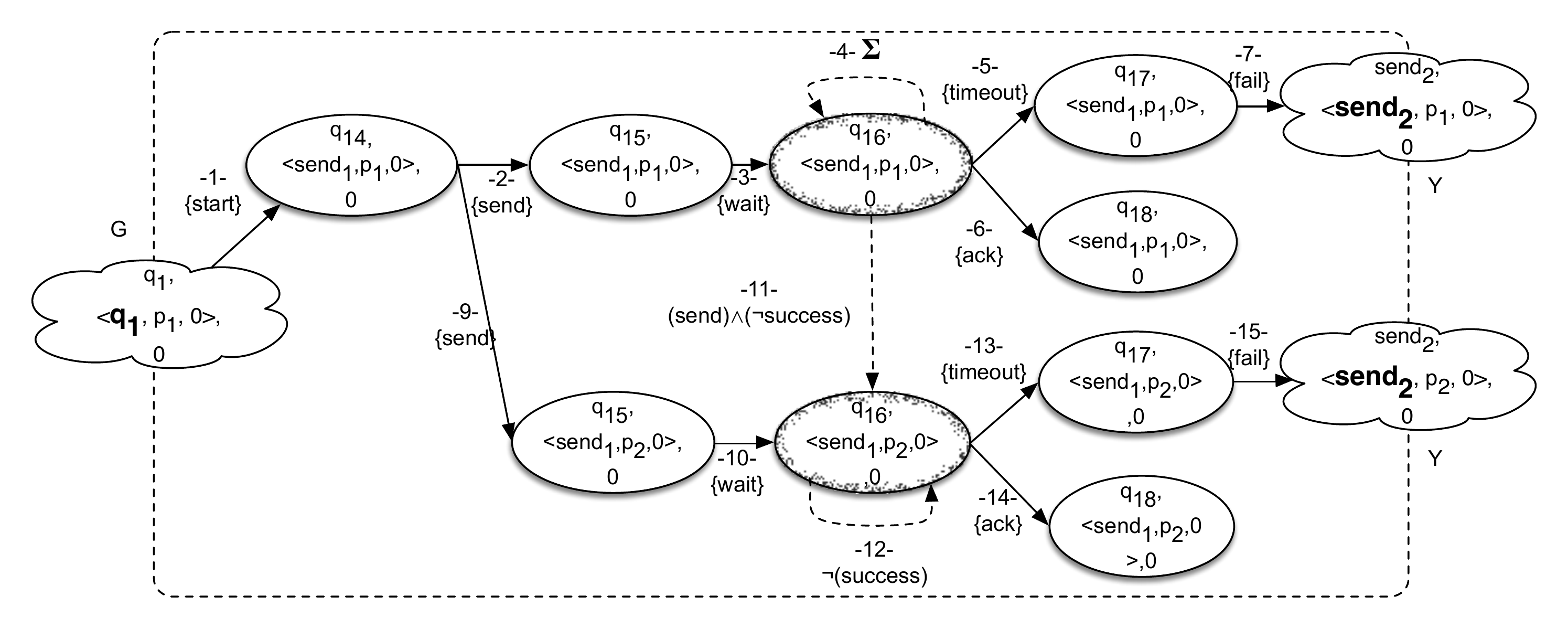}
\caption{Intersection between the replacement $\mathcal{R}$ described in Figure~\ref{Fig:send1Refinement} and the sub-property $\mathcal{S}_p$ presented in Figure~\ref{fig:subpropertysend1}.}
\label{Fig:IntersectionSubPropertyReplacement}
\end{figure}

We define as $\mathcal{E}_c$ the intersection obtained considering the completeness $\mathcal{T}_{c}$ of automaton associated with the replacement $\mathcal{R}$.
The intersection between  the sub-property  $\mathcal{S}=\langle  \mathcal{P}, \Delta^{in\mathcal{S}},$ $ \Delta^{out\mathcal{S}}, G, R, K \rangle$ associated to the box $b$ and the replacement $\mathcal{R}=\langle \mathcal{T}, \Delta^{inR},$ $ \Delta^{outR} \rangle$ has the same structure of a replacement (Defined in~\ref{replacement}), i.e., it contains an automata and a set of incoming and outgoing transitions, and it can be associated with finite internal, infinite internal, finite external and infinite external accepting runs as defined in Section~\ref{Sec:ModelingReplacements}. 

\begin{lemma}[Finite internal intersection language]
\label{replacement:finiteIntersectionLanguage}
 The intersection automaton $\mathcal{I}=\mathcal{R} \cap \mathcal{S}$ between the replacement $\mathcal{R}$ and the sub-property $\mathcal{S}$  recognizes the finite internal language $\mathcal{L}^{i\ast}(\mathcal{I})=(\mathcal{L}^{i\ast}(\mathcal{R}) \cup \mathcal{L}_p^{i\ast}(\mathcal{R})) \cap \mathcal{L}^\ast( \mathcal{S})$, i.e., $v \in \mathcal{L}^{i\ast}(\mathcal{I}) \Leftrightarrow v \in ((\mathcal{L}^{i\ast}(\mathcal{R}) \cup \mathcal{L}_p^{i\ast}(\mathcal{R})) \cap \mathcal{L}^{i\ast}( \mathcal{S}))$. 
\end{lemma}
\begin{proof} ($\Rightarrow$) If a finite word $v \in \mathcal{L}^{i\ast}(\mathcal{I})$, it must exists a finite internal run $\rho$ in the intersection automaton $\mathcal{I}$ where the initial state is an initial state of the intersection automaton and the final state is the destination of the outgoing transitions of $\mathcal{I}$. Since $\rho(0)$ must be an initial state of the intersection automaton, it must be obtained from an initial state of the automaton $\mathcal{R}$ associated with the replacement of the black box $b$. Let us identify with $\rho_{\mathcal{R}}(0)$ the state of the replacement $\mathcal{R}$ from which $\rho(0)$ is obtained. 
For each $i$, such that $0 \leq i < |\rho|$, for each transition ($\rho(i), a, \rho(i+1)$) that moves the system from the state $\rho(i)$ to the state $\rho(i+1)$, it must exist a transition ($\rho_{\mathcal{R}}(i), a, \rho_{\mathcal{R}}(i+1)$) in the automaton $\mathcal{R}$, which corresponds to the replacement of the box $b$, or $\rho_{\mathcal{R}}(i)=\rho_{\mathcal{R}}(i+1)$ and $\rho_{\mathcal{R}}(i)$ is a box of the replacement of $b$\footnote{This follows from the definition of the intersection (see Definition~\ref{def:intersection}).}. 
Since the last state of the run  $\rho(|v|)$ is the destination of an outgoing transition of $\mathcal{I}$, and the outgoing transitions of the intersection are obtained by synchronously executing transitions of the replacement and the sub-property, the  transition must also be outgoing for the replacement.
This implies that $\rho_{\mathcal{R}}(i)$ is a finite internal run (accepting or possibly accepting depending on the presence of boxes) for the replacement $\mathcal{R}$.
The same reasoning can be applied to demonstrate that $v$ is contained in the language $\mathcal{L}^\ast( \mathcal{S})$.\\
($\Leftarrow$) The proof is by contradiction. Assume that there exists a word $v \not \in \mathcal{L}^{i\ast}(\mathcal{I})$ which is in $((\mathcal{L}^{i\ast}(\mathcal{R}) \cup \mathcal{L}_p^{i\ast}(\mathcal{R})) \cap \mathcal{L}^\ast( \mathcal{S}))$. 
Since $v  \in ((\mathcal{L}^{i\ast}(\mathcal{R}) \cup \mathcal{L}_p^{i\ast}(\mathcal{R})) \cap \mathcal{L}^\ast( \mathcal{S}))$, it must exist a finite run $\rho_{\mathcal{R}}$ in the replacement and in the sub-property $\rho_{\mathcal{S}}$ associated with $v$. 
Given the  initial states $\rho_{\mathcal{R}}(0)$ and $\rho_{\mathcal{S}}(0)$ of the model and the claim, respectively, from which the initial state of the run is obtained, it must exist by construction a state $s$ in the intersection automaton $\mathcal{I}$ which is obtained by combining these two states. 
This state by construction is also initial for the intersection automaton. 
Let us identify with $\rho_{\mathcal{I}}$ the run that starts from this state. For each $0<i<|v|-1$, $\rho_{\mathcal{I}}(i+1)$ is associated to the state of the intersection automaton obtained by combining $\rho_{\mathcal{R}}(i+1)$ and $\rho_{\mathcal{S}}(i+1)$. 
Note that if $\rho_{\mathcal{R}}(i)$ and $\rho_{\mathcal{S}}(i)$ are connected to $\rho_{\mathcal{R}}(i+1)$ and $\rho_{\mathcal{S}}(i+1)$ with a transition labeled with $v_i$, or $\rho_{\mathcal{S}}(i)$ is connected to $\rho_{\mathcal{S}}(i+1)$ and $\rho_{\mathcal{R}}(i)$ is a box, then  $\rho_{\mathcal{I}}(i)$ and  $\rho_{\mathcal{I}}(i+1)$  are connected by a transition labeled with  $v_i$ by construction. 
Finally, the states $\rho_{\mathcal{R}}(|v|)$ and $\rho_{\overline{\mathcal{S}}}(|v|)$ are the destinations  of the outgoing transitions of $\mathcal{R}$ and $\mathcal{S}$ by construction. 
Indeed, it must exist an outgoing transition of the intersection automaton that corresponds to the synchronous execution of the outgoing transitions of  $\mathcal{R}$ and $\mathcal{S}$. 
Thus, the run $\rho_{\mathcal{I}}$ is a finite accepting run for the intersection automaton and $v \in \mathcal{L}^{i\ast}(\mathcal{I})$ that contradicts the hypothesis.
\end{proof}

\begin{lemma}[Finite external intersection language]
\label{lem:RepIntersectionFELanguage}
 The intersection automaton $\mathcal{I}=\mathcal{R} \cap \mathcal{S}$ between the replacement $\mathcal{R}$ and the sub-property $\mathcal{S}$  recognizes the finite external language $\mathcal{L}^{e\ast}(\mathcal{I})=(\mathcal{L}^{e\ast}(\mathcal{R}) \cup \mathcal{L}_p^{e\ast}(\mathcal{R})) \cap \mathcal{L}^{e\ast}( \mathcal{S})$, i.e., $v \in \mathcal{L}^{e\ast}(\mathcal{I}) \Leftrightarrow v \in ((\mathcal{L}^{e\ast}(\mathcal{R}) \cup \mathcal{L}_p^{e\ast}(\mathcal{R})) \cap \mathcal{L}^{e\ast}( \mathcal{S}))$. 
\end{lemma}

\begin{lemma}[Infinite internal intersection language]
\label{lem:ReIntersectionIILanguage}
 The intersection automaton $\mathcal{I}=\mathcal{R} \cap \mathcal{S}$ between the replacement $\mathcal{R}$ and the sub-property $\mathcal{S}$  recognizes the infinite internal language $\mathcal{L}^{i\omega}(\mathcal{I})=(\mathcal{L}^{i\omega}(\mathcal{R}) \cup \mathcal{L}_p^{i\omega}(\mathcal{R})) \cap \mathcal{L}^{i\omega}( \mathcal{S})$, i.e., $v^\omega \in \mathcal{L}^{i\omega}(\mathcal{I}) \Leftrightarrow v \in ((\mathcal{L}^{i\omega}(\mathcal{R}) \cup \mathcal{L}_p^{i\omega}(\mathcal{R})) \cap \mathcal{L}^{i\omega}( \mathcal{S}))$. 
\end{lemma}

\begin{lemma}[Infinite external intersection language]
\label{lem:RepIntersectionIELanguage}
 The intersection automaton $\mathcal{I}=\mathcal{R} \cap \mathcal{S}$ between the replacement $\mathcal{R}$ and the sub-property $\mathcal{S}$  recognizes the infinite external language $\mathcal{L}^{e\omega}(\mathcal{I})=(\mathcal{L}^{e\omega}(\mathcal{R}) \cup \mathcal{L}_p^{e\omega}(\mathcal{R})) $ $\cap \mathcal{L}^{e\omega}( \mathcal{S})$, i.e., $v^\omega \in \mathcal{L}^{e\omega}(\mathcal{I}) \Leftrightarrow v \in ((\mathcal{L}^{e\omega}(\mathcal{R}) \cup \mathcal{L}_p^{e\omega}(\mathcal{R})) \cap \mathcal{L}^{e\omega}( \mathcal{S}))$. 
\end{lemma}
\begin{proof}
The proofs of Lemmas~\ref{lem:RepIntersectionFELanguage},~\ref{lem:ReIntersectionIILanguage} and~\ref{lem:RepIntersectionIELanguage} can be easily derived from the proof of Lemma~\ref{replacement:finiteIntersectionLanguage}.
\end{proof}

%----------------------------------------------------------------------------------------------------------------------------------------------------------------
% UNDER APPROXIMATION
%----------------------------------------------------------------------------------------------------------------------------------------------------------------
The \emph{under approximation} automaton is used by an emptiness checking procedure to verify whether the claim \emph{is not} satisfied, i.e., it encodes the behaviors that violate the property of interest. The automaton is computed exploiting the information contained in the sub-property $\mathcal{S}$.

\begin{mydef}[Under approximation automaton]
\label{def:extendedsubpropReplIntersection} 
Given the sub-property \subproperty associated to the box $b$, the replacement \replacement, two additional automata states $g$ and $r$, the under approximation automaton $\mathcal{U}$ is the automaton obtained from $\mathcal{E}_c$ as follows:
\begin{itemize}
\item $\Sigma_{\mathcal{U}}=\Sigma_{\mathcal{E}_c}$;
\item $Q_{\mathcal{U}}=Q_{\mathcal{E}_c} \cup \left \{ g, r\right \} $;
\item $\Delta_{\mathcal{U}}=\Delta_{\underapproximation} \cup \Delta^{in}_{\underapproximation} \cup \Delta^{out}_{\underapproximation} \cup \Delta^{K}_{\underapproximation} \cup \Delta^{stut}_{\underapproximation}$, where 
\begin{itemize}
\item $\Delta^{stut}_{\underapproximation} = \{ (r, stut, r)   \} $;
\item $ \Delta^{in}_{\underapproximation}= \{ (g, a, s^\prime)  \mid (s, a, s^\prime) \in G  \} $;
\item $ \Delta^{out}_{\underapproximation}= \{ (s, a, r)  \mid (s, a, s^\prime) \in R  \} $;
\item $\Delta^{K}_{\underapproximation}= \{ (\langle q, p, x \rangle, \epsilon, \langle q^{\prime}, p^{\prime}, y \rangle)  \mid   ((q, a, q^{\prime\prime}), (q^{\prime\prime\prime}, b,  q^{\prime})) \in K   \}$. Moreover, the values $x$ and $y$ associated with $\delta_o=(q, a, q^{\prime\prime}),$ and $\delta_i=(q^{\prime\prime\prime}, b, q^{\prime})$  must satisfy the following conditions:
\begin{itemize}
\item if $K_{\mathcal{M}}(\delta_o, \delta_i)=T$ and $K_{\mathcal{A}_{\neg \phi}}(\delta_o, \delta_i)=T$ then $y=2$;
\item else if $x=1$ and  $K_{\mathcal{A}_{\neg \phi}}(\delta_o, \delta_i)=T$ or $p^{\prime} \in F_{\mathcal{P}} $  then $y=2$;
\item else if $x=0$ and  $K_\mathcal{M}(\delta_o, \delta_i)=T$ and $p^{\prime} \in F_{\mathcal{P}} $  then $y=2$;
\item else if $x=0$ and  $K_\mathcal{M}(\delta_o, \delta_i)=T$ or $q^{\prime} \in  F_{\mathcal{M}} $  then $y=1$;
\item else if $x=2$ then $y=0$;
\item else $y=x$.
\end{itemize}
\end{itemize} 
\item $Q^0_{\mathcal{U}}=Q^0_{\underapproximation} \cup \left \{ g \right \}$;
\item $F_{\mathcal{U}}=F_{\underapproximation} \cup \left \{ r \right \}$.
\end{itemize}
\end{mydef}
The completion of the extended intersection automaton contains all the behaviors of the intersection automaton that violate the claim $\phi$ plus additional transitions which specify how these behaviors are related to each others. 
The state  $g$ is used as a placeholder to represent the initial states of the system and the transitions in $\Delta^{in}_{\underapproximation} $ specify how the states of the intersection are reachable from the initial states. 
Similarly, the  state $r$ and the transition in $\Delta^{stut}_{\underapproximation}$ are used as  placeholders for a suffix of a run that does not involve the replacement of boxes and violates the claim. 
The transitions in $\Delta^{out}_{\underapproximation}$ specify how it is possible to reach these violating runs from the intersection between the sub-property and the refinement.
Finally, the transitions in $\Delta^{K}_{\underapproximation}$ specify how the violating behaviors of the intersection automaton between the replacement and the sub-property (which are portions of the intersection automaton between the refinement and the property) influence each other. 
Note that, as done in the computation of the classical intersection automaton, it is necessary to compute the  value of $y$ in the intersection state  $\langle q^{\prime}, p^{\prime}, y \rangle$. 
The value of $y$ depends on the presence of accepting states of the refinement and the property over the runs made by  purely regular states that connect the outgoing to the incoming transitions of the replacement, i.e., on the functions $\Gamma_{\mathcal{M}}$, $\Gamma_{\mathcal{A}_{\neg\phi}}$. 
More precisely, $y$ is identified as follows: 
whenever there exists both an accepting state of the model and of the claim in the original intersection automaton in a run made by purely regular states that connects the outgoing transition $\delta_o$ and the incoming transition $\delta_i$, the value of $y$ is $2$ to force the presence of an accepting state in the run.
 Similarly, if the value of $x$ is equal to $1$ and there exists a run  in the intersection automaton that connects $\delta_o$ to $\delta_i$ which traverses an accepting state of \claimautomata, the value of $y$ is set to $2$ to force the presence of an accepting state in the run. Finally, $y=2$ also if the value of $x$ is equal to $0$ there exists a run in the intersection automaton that connects $\delta_o$ to $\delta_i$ which traverses an accepting state of the model $\model$ and the destination state is accepting for the sub-property. Otherwise, if the value of $x$ is equal to $0$ and there exists a run in the intersection automaton that connects $\delta_o$ to $\delta_i$ which traverses an accepting state of the model $\model$ or the destination state is an accepting state of the replacement \replacementtext\ then $y=1$. If $x=2$, then $y=0$. In the other cases $y=x$.

\begin{figure}[ht]
\centering
\includegraphics[scale=0.5]{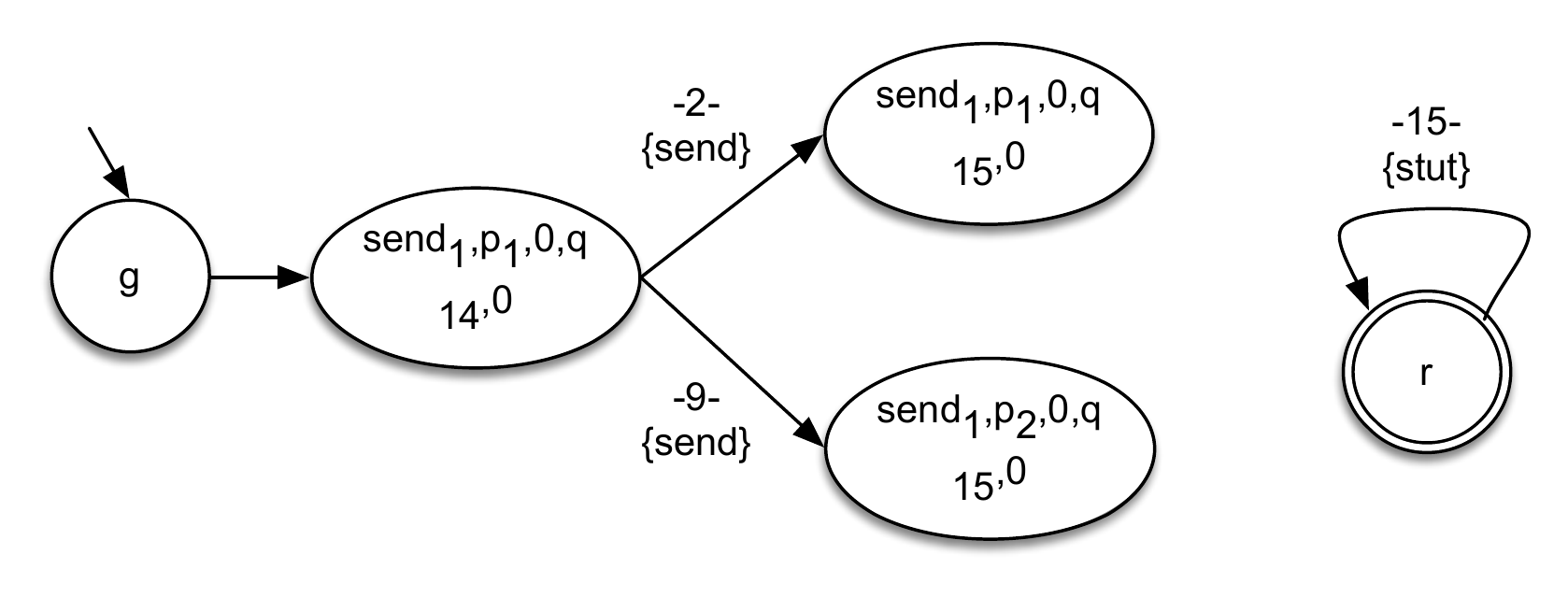}
\caption{The under approximation of the intersection described in Figure~\ref{Fig:IntersectionSubPropertyReplacement}.}
\label{Fig:IntersectionSubPropertyReplacementUnderApproximation}
\end{figure}

For example, the under approximation of the intersection automaton described in Figure~\ref{Fig:IntersectionSubPropertyReplacement} is presented in Figure~\ref{Fig:IntersectionSubPropertyReplacementUnderApproximation}. 
Note that, the reachability relation does not cause the injection of any transition in the intersection automaton. 
Furthermore, the accepting state marked with $r$ is not reachable. 

The \emph{over} approximation of the intersection automaton is similar to the under approximation, but the sub-property \subpropertypossiblytext\ is considered in its computation. This because the over approximation is the  automaton to be used by the emptiness checking procedure to verify the existence of possibly accepting behaviors.

\begin{mydef}[Over approximation automaton]
\label{def:extendedsubpropReplIntersectionNotCompletion} 
The over approximation automaton \overapproximation\ is obtained as the under approximation considering  sub-property \subpropertypossiblytext\ instead of \subpropertytext\
\end{mydef}

\begin{figure}[ht]
\centering
\includegraphics[scale=0.48]{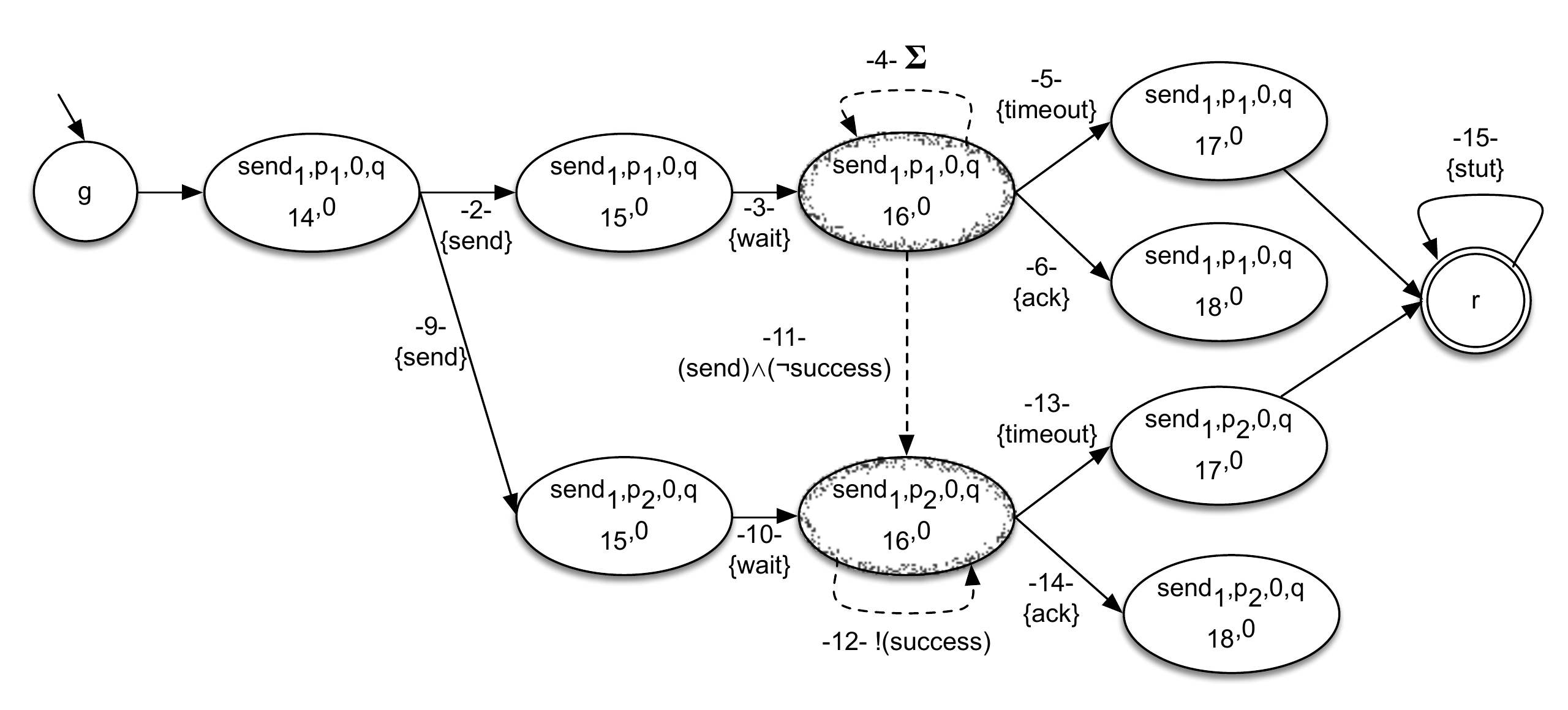}
\caption{The over approximation of the intersection automaton described in Figure~\ref{Fig:IntersectionSubPropertyReplacement}.}
\label{Fig:IntersectionSubPropertyReplacementOverApproximation}
\end{figure}
The over approximation automaton contains all the behaviors of the intersection automaton that violate and \emph{possibly} violate the claim $\phi$, and
 includes all the runs that connect  an incoming transition marked as $G$ with an outgoing marked as $R$. 
These runs may include transitions of the intersection automaton (which can be also generated by other boxes of the model), or transitions of the reachability  graph, which are used to abstract runs of the intersection between $\model$ and \claimautomata\ made by only purely regular and mixed states.

A portion of the over approximation of the intersection automaton described in Figure~\ref{Fig:IntersectionSubPropertyReplacement} is presented in Figure~\ref{Fig:IntersectionSubPropertyReplacementOverApproximation}. Note that, the state $r$ describes the presence of a suffix of a possibly violating behavior reachable in the intersection automaton.

\subsection{The model checking procedure}
\label{sec:replacementCheckingProcedure}
The replacement checking procedure, given the constraint \constrainttext\ and a replacement \replacementtext, checks whether \replacementtext\ definitely satisfies, possibly satisfies or does not satisfy \constrainttext. The replacement \replacementtext\ \emph{does not satisfy} the constraint \constrainttext\ if and only if the under approximation of the intersection automaton obtained considering the sub-property \subpropertytext\  associated to the box $b$ and the replacement \replacementtext\ is not empty. 
The replacement \replacementtext\ \emph{possibly satisfies} the constraint \constrainttext\ if and only if the over approximation of the intersection automaton is not empty or $\flag=T$, otherwise, the constraint is satisfied. Formally,

\begin{mydef} [Replacement checking]
\label{def:constraintsatisfaction}
Given the constraint \constraintf\ associated to the box $b$, and the replacement \replacement,
\begin{enumerate}
\item  \label{def:constraintsatisfactionFalse} $\| \mathcal{R}^{\mathcal{C}} \|=  F     \Leftrightarrow      \mathcal{L}(\underapproximation) \neq \emptyset$;   
\item  \label{def:subpropertysatisfactionTrue}   $\| \mathcal{R}^{\mathcal{C}}  \|=  T     \Leftrightarrow   \mathcal{L}($\overapproximation$)=\emptyset $ and $\flag=F$;
\item  $\| \mathcal{R}^{\mathcal{C}} \|=  \bot     \Leftrightarrow       \|  \mathcal{R}^{\mathcal{C}}  \| \not =  F \wedge  \|  \mathcal{R}^{\mathcal{C}} \| \not =  T$.
\end{enumerate}
where $\underapproximation$ and \overapproximation\ are the automata obtained as specified in Definitions~\ref{def:extendedsubpropReplIntersection} and \ref{def:extendedsubpropReplIntersectionNotCompletion}.
\end{mydef}

The idea behind the model checking procedure proposed in this section is to reduce the model checking problem to a cycle detection problem. A similar idea has been used, for example,  in the model checking of Hierarchical Kripke Structures~\cite{Alur:2001:MCH:503502.503503}.
To demonstrate the correctness of our definition, we prove that checking a replacement \replacementtext\ versus its constraint \constrainttext\ corresponds to checking the refined automaton $\refinement$ against the property $\propertytext$.

\begin{theorem}[Replacement checking correctness]
\label{th:constraintSatisfactionIsCorrect}
Given a model $\model$, a property $\propertytext$, a replacement \replacementtext\ for a box $b$ and the constraint \constrainttext\ obtained as previously described:
\begin{enumerate}
\item \label{def:correctenesssubpropertysatisfactionFalse} $\| \mathcal{R}^{\mathcal{C}} \|=  F      \Leftrightarrow      \| \mathcal{N}^{\mathcal{A}_{\neg \phi}} \|=  F   $;
\item  \label{def:correctenesssubpropertysatisfactionTrue} $\| \mathcal{R}^{\mathcal{C}}  \|=  T    \Leftrightarrow        \| \mathcal{N}^{\mathcal{A}_{\neg \phi}} \|=  T  $; 
\item \label{def:correctenesssubpropertysatisfactionMaybe} $ \| \mathcal{R}^{\mathcal{C}} \|=  \bot       \Leftrightarrow    \| \mathcal{N}^{\mathcal{A}_{\neg \phi}} \|=  \bot$.
\end{enumerate}
\end{theorem}

\begin{proof} Let us starts by proving condition~\ref{def:correctenesssubpropertysatisfactionFalse}. 

($\Rightarrow$) 
If $ \| \mathcal{R}^{\mathcal{C}} \|=  F$ by Definition~\ref{def:constraintsatisfaction}, Condition~\ref{def:constraintsatisfactionFalse}, it must exist a word $v$ accepted by the automaton $\underapproximation$.  
Let us consider the run $\rho_{\underapproximation}^\omega$ associated with $v$. 
We want to generate a run $\rho^\omega_{\mathcal{I}}$ in the intersection automaton  $\mathcal{I}$ between the refinement $\refinement$ and the claim \claimautomata\ which corresponds to $\rho_{\underapproximation}^\omega$. 
Let us consider the initial state $\rho_{\underapproximation}^\omega(0)$ of the run $\rho^\omega_{\underapproximation}$. 
Two cases are possible: 
\begin{inparaenum}[\itshape a\upshape)]
\item $\rho_{\underapproximation}^\omega(0)$ is obtained by combining an initial state $p$ of the automaton $\mathcal{P}$ of the sub-property $\mathcal{S}$ and an initial state $q$ of the replacement $\mathcal{R}$.  
Note that the initial state $p$ of the sub-property was obtained by combining an initial state $p^\prime$ of the property \claimautomata\ with a state $q^\prime$ of $\model$ which must be both initials. 
Furthemore, $q^\prime$ must be a box from construction (see Definitions~\ref{def:intersection} and \ref{def:subidentification}).  
Since the refinement  $\mathcal{N}$ contains all the states of the replacement $\mathcal{R}$, and an initial state of $\mathcal{R}$ is also initial for $\mathcal{N}$, it is possible to associate to  $\rho^\omega_{\mathcal{I}}(0)$ the state $\langle  q^\prime, p, 0\rangle$ of $\mathcal{I}$.
\item $\rho_{\underapproximation}^\omega(0)$ corresponds to the state $g$.  
Consider a transition $\delta \in  \Delta^{in}_{\underapproximation}$ that starts from the $g$ state. 
The transition $\delta$ is obtained by combining a transition $\delta^{inR} \in \Delta^{{inR}}$ with a transition $ \delta^{in\mathcal{S}} \in \Delta^{in\mathcal{S}}$, which is in turn obtained by combining a transition $\delta_{\mathcal{M}} \in \Delta_{\mathcal{M}}$ and a transition $\delta_{\mathcal{A}_{\neg \phi}} \in \Delta_{\mathcal{A}_{\neg \phi}}$. 
Let us consider the source states $q_{\mathcal{M}}$ and $p_{\mathcal{M}}$ of the transitions $\delta_{\mathcal{M}}$ and $\delta_{\mathcal{A}_{\neg \phi}}$ it is possible to replicate the run that reaches these states in the intersection automaton $\mathcal{I}$ since by definition plugging a replacement (Definition~\ref{def:plugginarefinement}) does not modify behaviors in which  only regular states are involved. 
Furthermore, the transition obtained from $\delta^{inR}$ and $\delta^{in\mathcal{S}}$ can be associated with the transition of the intersection automaton obtained combining $\delta^{inR}$ and $\delta_{\mathcal{A}_{\neg \phi}}$.
\end{inparaenum}
Let us now consider the other states of the run. Each state of the under approximation automaton can be rewritten as $\langle q_{\mathcal{R}}, \langle b, p, x \rangle, y\rangle$ since it is obtained by combining a state of the sub-property, which has the form $\langle b, p, x\rangle$, with a state $q_{\mathcal{R}}$ of the replacement. 
Each of these states can be associated with the state $\langle q_{\mathcal{R}}, p, y\rangle$ of the intersection between the replacement and the sub-property. 
Similarly, each transition $\delta_{\mathcal{I}} \in \Delta_{\mathcal{I}}$ of the intersection between the replacement and the sub-property is obtained by firing a transition of the replacement and an internal transition of the sub-property which corresponds to a transition of the original property, i.e., in $\Delta_{\mathcal{A}_{\neg \phi}}$, and a transition of the replacement, i.e.,   in $\Delta_{\mathcal{T}}$. 
Thus, the same transition can be identified in the intersection automaton obtained considering $\mathcal{N}$ and $\mathcal{A}_{\neg \phi}$. 
Let us finally consider a transition $(\rho_{\underapproximation}(i), a, \rho_{\underapproximation}(i+1)) \in (\Delta^{K}_{\underapproximation} \cup \Delta^{stut}_{\underapproximation})$ from construction (see Definitions~\ref{def:portsReachabilityGraph}) it must exists a sequence of transitions in the automaton $\mathcal{I}$ obtained from $\mathcal{N}$ and $\mathcal{A}_{\neg \phi}$ that connects only purely regular states and reach an accepting state that can be entered infinitely often and corresponds to this transition. 

($\Leftarrow$) 
The proof is by contradiction. 
Let us assume that  $ \| \mathcal{R}^{\mathcal{C}} \| \not =  F$ and $\| \mathcal{N}^{\mathcal{A}_{\neg \phi}} \|=  F $. 
Since $\| \mathcal{N}^{\mathcal{A}_{\neg \phi}} \|=  F $, it must exist a word $v$ accepted by the automaton $\mathcal{I}$ obtained from $\mathcal{N}$ and $\mathcal{A}_{\neg \phi}$. 
Since this run must be accepting, it must involve only purely regular states of $\mathcal{I}$. 
However, since $v$ was not accepted by the intersection obtained from $\model$ and $\mathcal{A}_{\neg \phi}$, some of these states must obviously be obtained by combining states of $\mathcal{R}$ and of $\mathcal{A}_{\neg \phi}$. 
This implies the presence of an accepting run in the automaton $\underapproximation$, which may connect the state ``$g$" with the state ``$r$" or another accepting state of $\underapproximation$ that can be entered infinitely often. 
Thus, $ \| \mathcal{R}^{\mathcal{C}} \| \not =  F$ is contradicted.

Let us now consider Condition~\ref{def:correctenesssubpropertysatisfactionTrue} of Theorem~\ref{th:constraintSatisfactionIsCorrect}. 
The proof corresponds to the one proposed for~\ref{def:correctenesssubpropertysatisfactionFalse}, but,  in this case,  sub-property $S_p$ is considered. Furthermore, if the flag  $\flag=T$, the sub-property is possible satisfied. Indeed, in this case, there exists a possibly accepting run in the intersection between the model $\model$ and the property $\mathcal{A}_{\neg \phi}$ that does not depend on the refinement of $b$. Thus, $\phi$ is possibly satisfied since the same run will be present in the intersection between  $\mathcal{N}$ and the property $\mathcal{A}_{\neg \phi}$.

The proof of condition~\ref{def:correctenesssubpropertysatisfactionMaybe}  of Theorem~\ref{th:constraintSatisfactionIsCorrect} follows from the proofs of conditions~\ref{def:correctenesssubpropertysatisfactionFalse} and~\ref{def:correctenesssubpropertysatisfactionTrue}.
\end{proof}

\begin{theorem}[Checking a replacement complexity]
\label{th:modelcheckingreplacement}
The complexity of the model checking procedure depends on the size of the automata $\mathcal{U}$  and $\mathcal{O}$, which in the worst case is $\mathcal{O}(|Q_{\mathcal{R}}| \cdot |Q_{\mathcal{P}}|+|\Delta_{\mathcal{R}}| \cdot |\Delta_{\mathcal{P}}|+|\Delta^{{inR}}| \cdot |\Delta^{in\mathcal{S}}|+|\Delta^{outR}| \cdot |\Delta^{out\mathcal{S}}|+(|\Delta^{out\mathcal{S}}| \cdot  |\Delta^{in\mathcal{S}}|)\cdot(|\Delta^{outR}| \cdot  |\Delta^{inR}|))$.
\end{theorem}

\begin{proof} 
The size of the automata described in Theorem~\ref{th:modelcheckingreplacement} is justified by the following statements. 
The size of the automaton obtained by considering the automaton $\mathcal{M}$ associated with the replacement $\mathcal{R}$ of the box $b$ and the automaton $\mathcal{P}$ associated with the sub-property $\mathcal{S}$ contains in the worst case $|Q_{\mathcal{R}}| \cdot |Q_{\mathcal{P}}|$ states and $|\Delta_{\mathcal{R}}| \cdot |\Delta_{\mathcal{P}}|$ transitions. 
This automaton can be reached through a set of transitions which are obtained by the synchronous execution of an incoming transition of the replacement and the sub-property, leading in the worst case to $|\Delta^{inR}| \cdot |\Delta^{in\mathcal{S}}|$ transitions. 
Similarly, the automaton can be left through a set of transitions obtained by the synchronous execution of an outgoing transition of the replacement and of the sub-property generating in the worst case $|\Delta^{{outR}}| \cdot |\Delta^{out\mathcal{S}}|$ transitions. 
Finally, each pair outgoing/incoming transition contained in the reachability relation of the sub-property can be synchronized with every  pair outgoing/incoming transition of the replacement, leading to $(|\Delta^{out\mathcal{S}}| \cdot  |\Delta^{in\mathcal{S}}|)\cdot(|\Delta^{outR}| \cdot  |\Delta^{inR}|))$ transitions. 
\end{proof}

%%% stuttering fictitious event

\bibliographystyle{plain}

\bibliography{bibliography}

\clearpage
%----------------------------------------------------------------------------------------
 \phantomsection\addcontentsline{toc}{chapter}{\listtheoremname}
\listoftheorems[numwidth=2.7em,ignoreall,show={mydef,constraint,lemma,theorem}]

\end{document}